\RequirePackage{fix-cm}
\documentclass[smallextended]{svjour3}       % onecolumn (second format)
\smartqed  % flush right qed marks, e.g. at end of proof
\usepackage{graphicx}
\usepackage{amsmath}
\usepackage{amssymb}
\usepackage{physics}
\usepackage{enumitem}
\usepackage{caption,booktabs}
\usepackage{url}
 \journalname{xxx}
\begin{document}

\title{Solving the chemical master equation for monomolecular reaction systems analytically: \\ a Doi-Peliti path integral view\thanks{This work was supported by NSF Grant \# DMS 1562078.}
}

\titlerunning{Monomolecular reaction systems: a Doi-Peliti path integral view}        % if too long for running head

\author{John J. Vastola}


\institute{John J. Vastola \at
              Department of Physics and Astronomy, Vanderbilt University, \\
              Nashville, Tennessee, USA \\
              \email{john.j.vastola@vanderbilt.edu}       
}

\date{Received: date / Accepted: date}

\maketitle

\begin{abstract}
The chemical master equation (CME) is a fundamental description of interacting molecules commonly used to model chemical kinetics and noisy gene regulatory networks. Exact time-dependent solutions of the CME---which typically consists of infinitely many coupled differential equations---are rare, and are valuable for numerical benchmarking and getting intuition for the behavior of more complicated systems. Jahnke and Huisinga's landmark calculation of the exact time-dependent solution of the CME for monomolecular reaction systems is one of the most general analytic results known; however, it is hard to generalize, because it relies crucially on properties of monomolecular reactions. In this paper, we rederive Jahnke and Huisinga's result on the time-dependent probability distribution and moments of monomolecular reaction systems using the Doi-Peliti path integral approach, which reduces solving the CME to evaluating many integrals. While the Doi-Peliti approach is less intuitive, it is also more mechanical, and hence easier to generalize. To illustrate how the Doi-Peliti approach can go beyond the method of Jahnke and Huisinga, we also find an explicit and exact time-dependent solution to a problem involving an autocatalytic reaction that Jahnke and Huisinga identified as not solvable using their method. We also find a formal exact time-dependent solution for any CME whose list of reactions involves only zero and first order reactions, which may be the most general result currently known.
\end{abstract}

%%%%%%%%%%%%%%%%%%%%%%%%%%%%%%%%%%%%%%%%%%%%%%%%

\section{Introduction}
\label{sec:intro}

The chemical master equation (CME) provides a fundamental description of well-mixed molecules interacting with each other via a set of chemical reactions \cite{mcquarrie1967,gillespie1992,gillespie2000,gillespie2007,gillespie2013,fox2017,munsky2018}. It models dynamics that are discrete (the state of the system is a set of nonnegative integers) and stochastic (chemical reactions occur with some probability). The CME has recently enjoyed tremendous success as a framework for understanding noisy single cell data \cite{neuert2013,munsky2015,fox2016,munsky2018shape,weber2018,fox2019,fox2019temporal}, particularly in simple model organisms like yeast where techniques like single-molecule Fluorescence \textit{in situ} Hybridization (smFISH) allow RNA molecule numbers to be counted accurately \cite{raj2008,femino1998,rahman2013}. Outside of cell and molecular biology, master equations have been successfully used to model population dynamics \cite{ovaskainen2010,melbinger2010,assaf2017}, traffic \cite{nagel1992,mahnke1997,mahnke2005}, and gas phase chemical kinetics \cite{miller2006,glowacki2012,jasper2014}, among other things. 

Although it is very useful for defining discrete stochastic models, the CME generally cannot be solved directly. One typically resorts to an approximate approach, like using Gillespie's algorithm \cite{gillespie1976,gillespie1977} to extract information from many brute force simulations, or using finite state projection \cite{munsky2006,peles2006,fox2016}, or partitioning the system (e.g. low versus high copy number, slow versus fast time scale) \cite{harris2006,harris2009,iyengar2010,bokes2012,hasenauer2014,kan2016}, or solving a continuous approximation to the CME like the chemical Langevin equation \cite{gillespie2000,gillespie2002,grima2011,vastolaPRE2020}. 

Unsurprisingly, exact time-dependent solutions (as opposed to steady state solutions) of the CME are particularly rare, and have only been computed for specific cases. McQuarrie \cite{mcquarrie1967} describes some of the early attempts: in 1940, Max Delbr\"{u}ck evaluated the CME for the autocatalytic reaction $S \to S + S$ \cite{delbruck1940}; in 1954, Renyi solved the binding reaction $A + B \to C$ \cite{renyi1954}; in 1960, Ishida solved the death reaction $S \to \varnothing$ and presented the first CME solution with time-dependent rates \cite{ishida1960}; in 1963 and 1964, McQuarrie et al. solved many simple systems (including $A + A \to B$ and $A + B \to C$) using the method of generating functions \cite{mcquarrie1963,mcquarrie1964}. 

The situation did not change appreciably until Jahnke and Huisinga's landmark paper \cite{jahnke2007}, more than forty years later. Their 2007 paper constituted a major advance in our collective understanding of the CME; they were able to solve the CME for a system with an \textit{arbitrary} number of species experiencing an \textit{arbitrary} number of reactions whose rates have \textit{arbitrary} time-dependence, \textit{provided that} the reactions consisted of some combination of birth ($\varnothing \to S_k$), death ($S_j \to \varnothing$), and conversion ($S_j \to S_k$)\footnote{These systems are called ``monomolecular'' because all allowed reactions have at most one molecule as input, and at most one molecule as output.}. The shocking generality of their result, as well as the explicitness of the solution they wrote down (in Theorem 1 of that paper), was powerful. 

Since 2007, there have been few new results of the same generality. Reis et al. \cite{reis2018} extend Jahnke and Huisinga's result by considering hierarchical first-order reaction networks (which allow a certain subset of first-order reactions that is strictly larger than the set of monomolecular reactions). However, there has not been (for example) any result on the solution to general first-order reactions, or general bimolecular reactions. At present, even finding the exact solutions of simple systems that involve bimolecular reactions is nontrivial: the work of Laurenzi ($A + B \leftrightarrow C$) \cite{laurenzi2000}, as well as Arslan and Laurenzi ($A + B \leftrightarrow A + A$) \cite{arslan2008} are two examples. 

One drawback of Jahnke and Huisinga's paper is that it essentially relied on guessing the solution. It was well-known that Poisson and multinomial distributions solved the CME in special cases, and that these distributions had certain desirable properties (e.g. a Poisson distribution stays a Poisson distribution, and a multinomial distribution stays a multinomial distribution; see Sec. 3 of their paper). To derive their Theorem 1, these properties were exploited, along with the fact that only monomolecular reactions were considered. Of course, their method completely breaks down for a system that is only slightly more complicated; as they point out in Sec. 6, adding an autocatalytic reaction $S \to S + S$ to a system they can easily solve manages to make it beyond the scope of their results. 

Hence, it would be nice if there was a method to obtain their classic result that did not rely on systematic guessing. In this paper, we offer the Doi-Peliti path integral approach to solving the CME as one such method. The Doi-Peliti approach allows one to `turn the crank', so to speak, and generate a time-dependent solution of the CME through a straightforward but difficult calculation. Importantly, it is system-agnostic: one does not need to know properties like `Poisson distributions stay Poisson', or assume the solution takes a certain form. 

Doi-Peliti field theory---which emerged from the pioneering papers of researchers like Doi \cite{doi1976,doi1976second}, Peliti \cite{peliti1985,peliti1985eden,peliti1986}, and Grassberger \cite{grassberger1980,grassberger1982,cardy1985,grassberger1989}---reframes solving the CME as a field theory problem. This enables the use of powerful approximation schemes, like the renormalization group and diagrammatic perturbation theory \cite{mattis1998,lee1994,leecardy1994,lee1995,vanWijland1998,canet2004,canet2006,tauber2005}. While Doi-Peliti field theory is still somewhat obscure in mathematical biology, it has seen the occasional application: e.g. to understand population dynamics given colored noise \cite{fung2017}, age dependent branching processes \cite{greenman2016,greenman2017}, and large deviations in gene regulatory networks \cite{bressloff2014,assaf2017}. Although not Doi-Peliti, a qualitatively similar path integral has been used to solve the CME for a multistep transcription and translation process \cite{albert2019}.  

We will use Doi-Peliti field theory to rederive Jahnke and Huisinga's Theorem 1. Moreover, in order to show that the Doi-Peliti path integral approach is strictly \textit{more} powerful than the one used by Jahnke and Huisinga, we use it to exactly solve a problem they said their method could not, as well as a far more general problem (the CME whose list of reactions consists of any combination of zero and first order reactions). We solve these additional problems in complete generality, and obtain exact time-dependent solutions assuming rates with arbitrary time-dependence. 

The paper is organized as follows. In Sec. \ref{sec:pstatement}, we state the problems we will solve, as well as our main results on their solutions. This includes reviewing the monomolecular CME and Jahnke and Huisinga's solution of it. In Sec. \ref{sec:reframe}, we review the generating function formulation of the CME, and develop the basic machinery of the Doi-Peliti approach to solving the differential equation satisfied by the generating function. This includes introducing several important concepts (including coherent states and two inner products) and constructing the Doi-Peliti path integral. We also describe how information like transition probabilities, moments, and (the usual) generating functions can be extracted from the (bra-ket) generating function. 
Sec. \ref{sec:monocalc} through \ref{sec:firstcalc} contain the Doi-Peliti calculations that back up our main results, with each section focusing on one of the problems introduced in Sec. \ref{sec:pstatement}: (i) monomolecular systems, (ii) one species birth-death-autocatalysis, and (iii) arbitrary combinations of zero and first order reactions. Sec. \ref{sec:propview} describes an alternative view of the propagator, the central object to be calculated in the Doi-Peliti approach; in particular, we show that our results can also be derived by solving a certain partial differential equation (PDE) using the method of characteristics. Finally, in Sec. \ref{sec:discussion}, we discuss the merits and drawbacks of the Doi-Peliti approach to solving the CME, and speculate on how it could be further utilized. 

%In Sec. \ref{sec:propcalc}, we solve the Doi-Peliti path integral. In Sec. \ref{sec:proptoprob}, we convert the path integral solution into a solution of the CME. In Sec. \ref{sec:proptomoments}, we show how to convert the path integral solution directly into the time-dependent moments of the CME. In Sec. \ref{sec:newresult}, we completely solve a problem Jahnke and Huisinga could not solve to demonstrate the flexibility of the Doi-Peliti approach. 

\section{Problem statements and main results}
\label{sec:pstatement}

In this section, we introduce in detail the specific problems we will solve, and we present our main results regarding their solutions. If the reader is \textit{only} interested in the answers, and not how they were obtained, then in some sense the rest of the paper is irrelevant; the validity of any given solution can usually be checked by tedious but straightforward substitution into the CME.

We present results for several systems, in order of increasing complexity: the chemical birth-death process, monomolecular reactions, single species birth-death-autocatalysis, and arbitrary combinations of zero and first order reactions with arbitrarily many species. While the chemical birth-death process is a kind of monomolecular reaction system, we include it here to give mathematicians new to the CME a toy example (that can be stated with minimal notational baggage) of the kind of results we are seeking.

% one can usually validity can usually be easily checked by substituting it directly into the CME.

\subsection{The chemical birth-death process as a prototype}
\label{sec:intro_bd}

The chemical birth-death process is simple enough to be biologically relevant (it can be used to model how the number of a single type of mRNA or protein in a single cell changes stochastically with time in the absence of significant regulation \cite{fox2017,munsky2018,bressloff2017}), but complicated enough to have nontrivial dynamics (the number of molecules does not increase without bound or shrink to zero in the long time limit, allowing there to exist a Poisson-like steady state probability distribution). It is linear (in several distinct but related senses of the word), allowing its associated CME to be exactly solved via a variety of methods (e.g. separation of variables and ladder operators \cite{vastolaADD2019}, and via a path integral different from the one we will discuss here \cite{vastolaPRE2020}). 

It is characterized by the chemical reactions
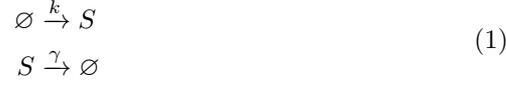
\begin{equation} \label{eq:bdrxns}
\begin{split}
\varnothing &\xrightarrow{k} S  \\
S &\xrightarrow{\gamma} \varnothing 
\end{split}
\end{equation}
where $k(t)$ and $\gamma(t)$ parameterize the small time birth and death rates, respectively\footnote{In particular, if the system has $x$ molecules of species $S$ at time $t$, the probability that the birth reaction happens in a window of time $[t, t + \Delta t)$ is approximately $k(t) \Delta t$, and the probability that the death reaction happens is approximately $\gamma(t) x \Delta t$. See Gillespie \cite{gillespie2000} for more details on the interpretation of and formalism underlying the CME.}. The time-dependence of these parameters is allowed to be arbitrary as long as they both remain nonnegative for all times. The corresponding CME reads
\begin{equation} \label{eq:bdCME}
\begin{split}
\frac{\partial P(x, t)}{\partial t} =& \ k(t) \left[ P(x-1, t) - P(x, t) \right] \\
&+ \gamma(t) \left[ (x+1) P(x+1, t) - x P(x, t) \right] 
\end{split}
\end{equation}
where $P(x, t)$ is the probability that the state of the system is $x \in \mathbb{N} := \{ 0, 1, 2, ... \}$ at time $t \geq t_0$. We are particularly interested in the transition probability $P(x, t; \xi, t_0)$, i.e. the solution of Eq. \ref{eq:bdCME} whose probability distribution is initially certain ($P(x, t_0) = \delta(x - \xi)$ for some $\xi \in \mathbb{N}$); if the transition probability is known, the solution to Eq. \ref{eq:bdCME} for an arbitrary initial distribution $P_0(x)$ can be written
\begin{equation}
P(x, t) = \sum_{y = 0}^{\infty} P(x, t; y, t_0) P_0(y) \ .
\end{equation}
In practice, we are also interested in several other properties of the solution. In particular, we are interested in the long time behavior described by the steady state probability distribution
\begin{equation}
P_{ss}(x) := \lim_{t \to \infty} P(x, t; \xi, t_0) \ ,
\end{equation}
moments like
\begin{equation}
\begin{split}
\langle x(t) \rangle &:= \sum_{x = 0}^{\infty} x P(x, t) \\
\langle x(t) [ x(t) - 1 ] \rangle &:= \sum_{x = 0}^{\infty} x (x - 1) P(x, t) \ ,
\end{split}
\end{equation}
and the real or complex-valued probability generating function 
\begin{equation}
\psi(g, t) := \sum_{x = 0}^{\infty} P(x, t) \ g^x \ ,
\end{equation}
whose derivatives correspond to various moments of interest. In fact, knowing the generating function is equivalent to knowing $P(x, t)$, and its time evolution is described by a PDE analogous to the CME:
\begin{equation} \label{eq:bdGFeqn}
\frac{\partial \psi(g, t)}{\partial t} = k(t) [g - 1] \psi(g, t) - \gamma(t) [g - 1] \frac{\partial \psi(g, t)}{\partial g} \ .
\end{equation}
The initial condition corresponding to $P(x, t_0) = \delta(x - \xi)$ is $\psi(g, t_0) = g^{\xi}$, as can be verified using the definition of $\psi$. We can also take its long time limit, which we will denote by $\psi_{ss}(g)$. 

%We will denote the solution of Eq. \ref{eq:bdCME} whose probability distribution is initially certain (i.e. $P(x, t_0) = \delta(x - \xi)$ for some $\xi \in \mathbb{N}$) as $P(x, t; \xi, t_0)$, and its long time limit as
%
%
%Probability generating function:
%
%
%Can also write an equation for $\psi$ equivalent to the CME:

The main result for the chemical birth-death process is the following. 

%- P_trans, P_ss, first moment, second moment, generating function
\begin{theorem}[Chemical birth-death process] \label{thm:bd}
Let $\lambda(t)$ and $w(t)$ be the solutions of
\begin{equation}
\begin{split}
\dot{\lambda} &= k - \gamma \lambda \ , \ \lambda(t_0) = 0 \\
\dot{w} &= - \gamma w \ , \ w(t_0) = 1 \ .
\end{split}
\end{equation}
Then if $P(x, t_0) = \delta(x - \xi)$ for some $\xi \in \mathbb{N}$, we have:
\begin{enumerate}[label=(\roman*)]
\item 
\begin{equation}P(x, t; \xi, t_0) =  \sum_{k = 0}^{\min(x, \xi)}  \frac{\lambda(t)^{x - k} e^{-\lambda(t)}}{(x - k)!} \binom{\xi}{k} w(t)^k  [1 - w(t)]^{\xi - k} \end{equation}
\item \begin{equation}\expval{x(t)} = \xi w(t) + \lambda(t)
\end{equation}
\item \begin{equation} \expval{x(t)[x(t) - 1]} = w(t)^2 \xi (\xi - 1) + 2 \lambda(t) w(t) \xi + \lambda(t)^2   \end{equation}
\item \begin{equation}  \psi(g, t) = \left[ 1 + (g - 1) w(t) \right]^{\xi}  e^{(g - 1) \lambda(t)}   \end{equation}
\end{enumerate}
\end{theorem}
\begin{proof}
All results can be obtained independently of one another using the Doi-Peliti approach described in the following sections. Alternatively, one can verify directly that $P(x, t; \xi, t_0)$ satisfies Eq. \ref{eq:bdCME} or that $\psi(g, t)$ satisfies Eq. \ref{eq:bdGFeqn}, and then obtain the rest of the results by brute force calculation. \qed
\end{proof}

These results simplify tremendously in the long time limit if $k$ and $\gamma$ are time-independent, essentially because $P(x, t)$ reduces to a Poisson distribution regardless of one's choice of initial distribution $P(x, t_0)$. 

\begin{corollary}[Long time behavior of chemical birth-death process]
Let $k$ and $\gamma$ be time-independent, and define $\mu := k/\gamma$. In the long time limit, we have:
\begin{enumerate}[label=(\roman*)]
\item 
\begin{equation}P_{ss}(x) =   \frac{\mu^{x} e^{- \mu}}{x!}  \end{equation}
\item \begin{equation}\expval{x} = \mu
\end{equation}
\item \begin{equation} \expval{x [x - 1]} = \mu^2   \end{equation}
\item \begin{equation}  \psi_{ss}(g) = e^{(g - 1) \mu}   \end{equation}
\end{enumerate}
\end{corollary}

\begin{proof} Take the $t \to \infty$ limit of the previous results, noting that $w(t) \to 0$ and $\lambda(t) \to \mu$. \qed
\end{proof}
%\begin{theorem}[Chemical birth-death process]
%The solution to Eq. \ref{eq:bdCME} with initial condition $P(x, t_0) = \delta(x - \xi)$ for some $\xi \in \mathbb{N}$ is 
%\begin{equation} 
%P(x, t; \xi, t_0) =  \sum_{k = 0}^{\min(x, \xi)}  \frac{\lambda(t)^{x - k} e^{-\lambda(t)}}{(x - k)!} \binom{\xi}{k} w(t)^k  [1 - w(t)]^{\xi - k} 
%\end{equation}
%for all $t \geq t_0$, where $\lambda(t)$ is the solution of
%\begin{equation}
%\dot{\lambda} = k - \gamma \lambda
%\end{equation} 
%with initial condition $\lambda(t_0) = 0$, and $w(t)$ is the solution of
%\begin{equation}
%\dot{w} = - \gamma w
%\end{equation}
%with initial condition $w(t_0) = 1$. Furthermore, the steady state distribution corresponding to this solution is
%\begin{equation}
%P_{ss}(x) = 
%\end{equation}
%\begin{enumerate}
%\item a
%\item b
%\end{enumerate}
%\end{theorem}

%\begin{equation}
%\begin{split}
%P(x, t; \xi, t_0) &=  \sum_{k = 0}^{\min(x, \xi)} \left[ \frac{\lambda(t)^{x - k} e^{-\lambda(t)}}{(x - k)!} \right] \left[ \binom{\xi}{k} w(t)^k  [1 - w(t)]^{\xi - k} \right]  
%\end{split}
%\end{equation}
%
%
%\begin{equation}
%\begin{split}
%P &=  \sum_{k = 0}^{\min(x, \xi)} \left[ \frac{\lambda(t)^{x - k} e^{-\lambda(t)}}{(x - k)!} \right] \left[ \binom{\xi}{k} w(t)^k  [1 - w(t)]^{\xi - k} \right]  \\
%&= \mathcal{P}(x, \lambda(t)) \star \mathcal{M}(x, \xi, w(t))
%\end{split}
%\end{equation}

Results on the chemical birth-death process are far from new. Still, we ask the reader to keep these results in the back of their mind as we go on to discuss the analogous results for more complicated systems. Because the chemical birth-death process is so fundamental, more complicated results often reduce to these in the appropriate limit.

\subsection{Monomolecular results}
\label{sec:intro_mono}
Heuristically, monomolecular reaction systems are the minimal multi-species generalization of the chemical birth-death process. Like the chemical birth-death process, they exhibit a certain kind of linearity; we will see that their solutions are completely determined by a system of linear ordinary differential equations (ODEs).

%In the same way that the chemical birth-death process is linear in certain senses of the word---in particular, the equations for $\lambda(t)$ and $w(t)$ are linear ordinary differential equations (ODEs)---the solutions of monomolecular reaction systems are also controlled by linear ODEs.

Let us define them. Consider a system with $n$ chemical species $S_1, ..., S_n$, whose reaction list reads
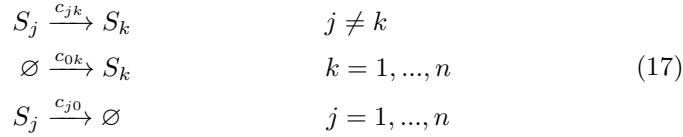
\begin{equation} \label{eq:rxns}
\begin{split}
S_j &\xrightarrow{c_{jk}} S_k \hspace{1in} j \neq k \\
\varnothing &\xrightarrow{c_{0k}} S_k \hspace{1in} k = 1, ..., n \\
S_j &\xrightarrow{c_{j0}} \varnothing \hspace{1.06in} j = 1, ..., n 
\end{split}
\end{equation}
i.e. all possible monomolecular reactions (birth, death, and conversion) are allowed\footnote{The one exception is the trivial conversion reaction $S_k \to S_k$, which is disallowed because including it would be pointless. To ease notation (i.e. to avoid writing $j \neq k$ many times), we follow Jahnke and Huisinga and define $c_{kk} := 0$ for all $k = 1, ..., n$.}. Note that the rates are allowed to have arbitrary time-dependence as long as $c_{jk}(t) \geq 0$ for all $j, k$ and all times $t \geq t_0$. The corresponding CME reads
\begin{equation} \label{eq:CME}
\begin{split}
\frac{\partial P(\mathbf{x},t)}{\partial t} =& \sum_{k=1}^{n} c_{0k}(t) \left[ P(\mathbf{x} - \boldsymbol{\epsilon}_k, t) - P(\mathbf{x},t)  \right]  \\
&+ \sum_{k=1}^{n} c_{k0}(t)\left[ (x_k + 1) P(\mathbf{x} + \boldsymbol{\epsilon}_k, t) - x_k P(\mathbf{x},t)  \right] \\
&+ \sum_{j=1}^{n} \sum_{k=1}^{n}  c_{jk}(t) \left[ (x_j + 1) P(\mathbf{x} + \boldsymbol{\epsilon}_j - \boldsymbol{\epsilon}_k, t) - x_j P(\mathbf{x},t)  \right]
\end{split}
\end{equation}
where $P(\mathbf{x}, t)$ is the probability that the state of the system is $\mathbf{x} := (x_1, ..., x_n)^T \in \mathbb{N}^n$ at time $t \geq t_0$, and where $\boldsymbol{\epsilon}_k$ is the $n$-dimensional vector with a $1$ in the $k$th place and zeros everywhere else. 

The exact solution to Eq. \ref{eq:CME}, given the initial condition $P(\mathbf{x}, t_0) = \delta(\mathbf{x} - \boldsymbol{\xi})$ for some vector $\boldsymbol{\xi} := (\xi_1, ..., \xi_n)^T \in \mathbb{N}^n$, is reported in Theorem 1 of Jahnke and Huisinga \cite{jahnke2007}. In order to state their solution, we will need some notation. 

Define the matrix $\mathbf{A}(t)$ and vector $\mathbf{b}(t)$ by 
\begin{equation}
\begin{split}
A_{j k}(t) &:= c_{k j}(t) \ \text{ for } j \neq k \geq 1 \\
A_{k k}(t) &:= - \sum_{j = 0}^n c_{k j}(t) \ \text{ for } 1 \leq k \leq n \\
\mathbf{b}(t) &:= \begin{pmatrix} c_{0 1}(t) & c_{0 2}(t) & \cdots & c_{0 n}(t)   \end{pmatrix}^T \ .
\end{split}
\end{equation}
The deterministic reaction rate equations corresponding to our reaction list can be written in terms of $\mathbf{A}(t)$ and $\mathbf{b}(t)$ as
\begin{equation} \label{eq:RRE}
\dot{\mathbf{x}} = \mathbf{A}(t) \mathbf{x} + \mathbf{b}(t) \ .
\end{equation}
Because Eq. \ref{eq:RRE} is linear, the solution with initial condition $\boldsymbol{\xi} = (\xi_1, ..., \xi_n)$ can be written as
\begin{equation}
\mathbf{x}(t) = \sum_{k=1}^n  \xi_k \mathbf{w}^{(k)}(t) + \boldsymbol{\lambda}(t)
\end{equation}
where the vectors $\mathbf{w}^{(1)}(t), ..., \mathbf{w}^{(n)}(t)$ and $\boldsymbol{\lambda}(t)$ are defined as
\begin{equation}
\begin{split}
\dot{\mathbf{w}}^{(k)} &= \mathbf{A}(t) \mathbf{w}^{(k)} \ , \ \mathbf{w}^{(k)}(t_0) = \boldsymbol{\epsilon}_k \\
\dot{\boldsymbol{\lambda}} &= \mathbf{A}(t) \boldsymbol{\lambda} + \mathbf{b}(t) \ , \ \boldsymbol{\lambda}(t_0) = \mathbf{0} \ .
\end{split}
\end{equation}
As we will shortly observe, the solution to the deterministic reaction rate equations is intimately related to the solution of the CME (at least for monomolecular reactions). 

Now define the $1$-norm of a vector $\mathbf{x}$ as
\begin{equation}
|\mathbf{x}| := \sum_{k=1}^n |x_k| \ .
\end{equation}
Define, because they will appear throughout this paper, multi-dimensional generalizations of powers, factorials, sums, integrals, and derivatives:
\begin{equation}
\begin{split}
\mathbf{v}^{\mathbf{x}} &:= v_1^{x_1} \cdots v_n^{x_n} \\
\mathbf{x}! &:= x_1! \cdots x_n! \\
\sum_{\mathbf{x}} &:= \sum_{x_1 = 0}^{\infty} \cdots \sum_{x_n = 0}^{\infty} \\
\int d\mathbf{x} &:= \int dx_1 \cdots \int dx_n \\
\left( \frac{d}{d\mathbf{z}} \right)^\mathbf{x} &:= \left( \frac{d}{dz_1} \right)^{x_1} \cdots \left( \frac{d}{dz_n} \right)^{x_n} \ .
\end{split}
\end{equation}
Using the above shorthand, we can define the product Poisson distribution as
\begin{equation} \label{eq:productpoisson}
\mathcal{P}(\mathbf{x}, \boldsymbol{\lambda}) := \frac{\lambda_1^{x_1}}{x_1!} \cdots \frac{\lambda_n^{x_n}}{x_n!} e^{-|\boldsymbol{\lambda}|} = \frac{\boldsymbol{\lambda}^{\mathbf{x}}}{\mathbf{x}!}  e^{-|\boldsymbol{\lambda}|} \ ,
\end{equation}
the multinomial distribution as
\begin{equation} \label{eq:multinomial}
\begin{split}
\mathcal{M}(\mathbf{x}, N, \mathbf{w}) &:= \frac{N! \ \left[ 1 - |\mathbf{w}| \right]^{N - |\mathbf{x}|}}{(N - |\mathbf{x}|)!}  \frac{w_1^{x_1}}{x_1!} \cdots \frac{w_n^{x_n}}{x_n!} \ \text{ if } \ |\mathbf{x}| \leq N \ \text{ and } \ x \in \mathbb{N}^N \\
 &= \frac{N! \ \left[ 1 - |\mathbf{w}| \right]^{N - |\mathbf{x}|}}{(N - |\mathbf{x}|)!}  \frac{\mathbf{w}^{\mathbf{x}}}{\mathbf{x}!}  \ \text{ if } \ |\mathbf{x}| \leq N \ \text{ and } \ x \in \mathbb{N}^N \ ,
\end{split}
\end{equation}
and the convolution of two probability distributions as
\begin{equation}
P_1(\mathbf{x}) \star P_2(\mathbf{x}) := \sum_{\mathbf{z}} P_1(\mathbf{z}) P_2(\mathbf{x} - \mathbf{z}) = \sum_{\mathbf{z}} P_1(\mathbf{x} - \mathbf{z}) P_2(\mathbf{z}) 
\end{equation}
where the sum is over all $\mathbf{z} \in \mathbb{N}^n$ such that $\mathbf{x} - \mathbf{z} \in \mathbb{N}^n$. As in the one-dimensional case, we can also define the probability generating function
\begin{equation}
\psi(\mathbf{g}, t) := \sum_{\mathbf{x}} P(\mathbf{x}, t) \ \mathbf{g}^{\mathbf{x}}
\end{equation}
which satisfies the PDE
\begin{equation} \label{eq:monoGFeqn}
\begin{split}
\frac{\partial \psi(\mathbf{g}, t)}{\partial t} =& \sum_{k=1}^{n} c_{0k}(t) \left[ g_k - 1 \right]  \psi(\mathbf{g}, t) \\
&- \sum_{k=1}^{n} c_{k0}(t)\left[ g_k - 1  \right] \frac{\partial \psi(\mathbf{g}, t)}{\partial g_k} \\
&+ \sum_{j=1}^{n} \sum_{k=1}^{n}  c_{jk}(t) \left[ g_k - g_j \right] \frac{\partial \psi(\mathbf{g}, t)}{\partial g_j} \ .
\end{split}
\end{equation}
With all of that notation defined, we are ready to state the main result for monomolecular reaction systems, which was originally proved by Jahnke and Huisinga\footnote{The generating function was not directly computed by them, but is essentially trivial to compute given their results.}.

\begin{theorem}[Jahnke-Huisinga monomolecular] \label{thm:mono}
Let $\boldsymbol{\lambda}(t)$ and $\mathbf{w}(t)$ be the solutions of
\begin{equation}   \label{eq:wlambda}
\begin{split}
\dot{\boldsymbol{\lambda}} &= \mathbf{A}(t) \boldsymbol{\lambda} + \mathbf{b}(t) \ , \ \boldsymbol{\lambda}(t_0) = \mathbf{0} \\
\dot{\mathbf{w}}^{(k)} &= \mathbf{A}(t) \mathbf{w}^{(k)} \ , \ \mathbf{w}^{(k)}(t_0) = \boldsymbol{\epsilon}_k \ .
\end{split}
\end{equation}
Then if $P(\mathbf{x}, t_0) = \delta(\mathbf{x} - \boldsymbol{\xi})$ for some $\boldsymbol{\xi} \in \mathbb{N}^n$, we have:
\begin{enumerate}[label=(\roman*)]
\item 
\begin{equation} \label{eq:fullsln}
P(\mathbf{x}, t; \boldsymbol{\xi}, t_0) = \mathcal{P}(\mathbf{x}, \boldsymbol{\lambda}(t)) \star \mathcal{M}(\mathbf{x}, \xi_1, \mathbf{w}^{(1)}(t)) \star \cdots \star \mathcal{M}(\mathbf{x}, \xi_n, \mathbf{w}^{(n)}(t))
\end{equation}
\item \begin{equation}
\expval{x_j(t)}  = \sum_{k=1}^n  {\xi}_k w^{(k)}_j(t) + \lambda_j(t)
\end{equation}
\item \begin{equation}
\mathrm{Cov}(x_j, x_{\ell}) =  \begin{cases} 
      \sum_{k = 1}^n \xi_k w^{(k)}_j \left[ 1 - w^{(k)}_j  \right] + \lambda_j & j = \ell \\
      - \sum_{k = 1}^n \xi_k  w^{(k)}_j w^{(k)}_{\ell} & j \neq \ell
   \end{cases} \ .
\end{equation}
\item \begin{equation}  \psi(\mathbf{g}, t) = \prod_{k = 1}^n \ \left[ 1 + ( \mathbf{g} - \mathbf{1}) \cdot \mathbf{w}^{(k)}(t)  \right]^{\xi_k}   \ e^{(\mathbf{g} - \mathbf{1}) \cdot \boldsymbol{\lambda}(t)}   \end{equation}
\end{enumerate}
\end{theorem}
\begin{proof}
All results can be obtained independently of one another using the Doi-Peliti approach described in the following sections. Alternatively, one can verify directly that $P(\mathbf{x}, t; \boldsymbol{\xi}, t_0)$ satisfies Eq. \ref{eq:CME} or that $\psi(\mathbf{g}, t)$ satisfies Eq. \ref{eq:monoGFeqn} with the correct initial condition, and then obtain the rest of the results by brute force calculation. \qed
\end{proof}

%\begin{theorem}[Jahnke-Huisinga monomolecular]
%The solution to Eq. \ref{eq:CME} with initial condition $P(\mathbf{x}, t_0) = \delta(\mathbf{x} - \boldsymbol{\xi})$ for some vector $\boldsymbol{\xi} := (\xi_1, ..., \xi_n) \in \mathbb{N}^n$ is 
%\begin{equation} 
%P(\mathbf{x}, t; \boldsymbol{\xi}, t_0) = \mathcal{P}(\mathbf{x}, \boldsymbol{\lambda}(t)) \star \mathcal{M}(\mathbf{x}, \xi_1, \mathbf{w}^{(1)}(t)) \star \cdots \star \mathcal{M}(\mathbf{x}, \xi_n, \mathbf{w}^{(n)}(t))
%\end{equation}
%for all $t \geq t_0$, where $P(\mathbf{x}, t; \boldsymbol{\xi}, t_0)$ denotes the probability that the system transitions from state $\boldsymbol{\xi}$ to state $\mathbf{x}$ in time $T := t - t_0$.
%\end{theorem}

\subsection{Birth-death-autocatalysis results}

In section 6 of their classic paper \cite{jahnke2007}, Jahnke and Huisinga solve the CME corresponding to the autocatalytic reaction $S \to S + S$ exactly; however, they note that adding birth and death reactions yields a system not amenable to their approach. In this section, we present the exact time-dependent solution to this problem, whose reactions read
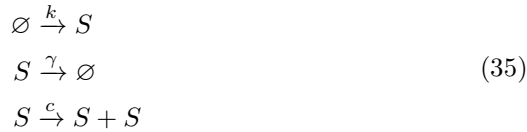
\begin{equation} \label{eq:newrxns}
\begin{split}
\varnothing &\xrightarrow{k} S  \\
S &\xrightarrow{\gamma} \varnothing  \\
S &\xrightarrow{c} S + S  
\end{split}
\end{equation}
where the parameters controlling the rates of birth, death, and autocatalysis are all allowed to have arbitrary time-dependence as long as they are nonnegative for all times.  The CME reads
\begin{equation} \label{eq:newCME}
\begin{split}
\frac{\partial P(x, t)}{\partial t} =& \ k(t) \left[ P(x-1, t) - P(x, t) \right] \\
&+ \gamma(t) \left[ (x+1) P(x+1, t) - x P(x, t) \right] \\
&+ c(t) \left[ (x-1) P(x-1, t) - x P(x, t) \right] 
\end{split}
\end{equation}
where $P(x, t)$ is the probability that the state of the system is $x \in \mathbb{N}$ at time $t \geq t_0$. Meanwhile, the PDE satisfied by the probability generating function $\psi$ reads
\begin{equation} \label{eq:bdaGFeqn}
\begin{split}
\frac{\partial \psi(g, t)}{\partial t} =& \ k(t) [g - 1] \psi(g, t) - \gamma(t) [g - 1] \frac{\partial \psi(g, t)}{\partial g} \\
&+ c(t) [g - 1] g \frac{\partial \psi(g, t)}{\partial g} \ .
\end{split}
\end{equation}
Our main result on the solution of this system is the following. 

%
%\begin{equation} \label{eq:newDE}
%\dot{q}(s) = \left[ c(t - s + t_0) - \gamma(t - s + t_0) \right] q(s) + i c(t - s + t_0) \ q(s)^2 
%\end{equation}
%where $s \in [t_0, t]$ and $q(t_0) = p_f$. As can be verified by substitution, Eq. \ref{eq:newDE} is solved by
%\begin{equation}
%q(s) = \frac{w(s)}{\frac{1}{p_f} - i \int_{t_0}^s c(t - t' + t_0) w(t') \ dt'}
%\end{equation}
%where $w(t)$ is the solution to
%\begin{equation}
%\dot{w}(s) = [c(t - s + t_0) - \gamma(t - s + t_0)] \ w(s)
%\end{equation}
%with $w(t_0) = 1$ (c.f. Eq. \ref{eq:wlambda}), i.e.
%\begin{equation}
%w(s) = e^{\int_{t_0}^s c(t - t' + t_0) - \gamma(t - t' + t_0) \ dt'} \ .
%\end{equation}

\begin{theorem}[Birth-death-autocatalysis] \label{thm:bda}
Let $q(s)$ and $w(s)$ be the solutions of
\begin{equation}
\begin{split}
\dot{q} &= \left[ c(t - s + t_0) - \gamma(t - s + t_0) \right] q(s) + i c(t - s + t_0) \ q(s)^2  \ , \ q(t_0) = p_f \\
\dot{w} &= [c(t - s + t_0) - \gamma(t - s + t_0)] \ w(s) \ , \ w(t_0) = 1 
\end{split}
\end{equation}
for arbitrary $p_f \in \mathbb{R}$, which can be explicitly written as
\begin{equation}
\begin{split}
q(s) &= \frac{w(s)}{\frac{1}{p_f} - i \int_{t_0}^s c(t - t' + t_0) w(t') \ dt'} \\
w(s) &= e^{\int_{t_0}^s c(t - t' + t_0) - \gamma(t - t' + t_0) \ dt'} \ .
\end{split}
\end{equation}
Then if $P(x, t_0) = \delta(x - \xi)$ for some $\xi \in \mathbb{N}$, we have:
\begin{equation} \label{eq:contoursln2}
\begin{split}
P(x, t; \xi, t_0) =& \frac{1}{2 \pi} \int_{-\infty}^{\infty} dp_f \ \frac{\left[ 1 + i q(t) \right]^{\xi} e^{i \int_{t_0}^t k(t - s + t_0) q(s) ds}}{(1 + i p_f)^{x + 1}} \\
\psi(g, t) =& \left[ 1 + \frac{w(t)}{\frac{1}{g-1} - \int_{t_0}^t c(t - t' + t_0) w(t') \ dt'} \right]^{\xi} \times \\
&\times \exp\left\{ \int_{t_0}^t  \frac{k(t - s + t_0) w(s)}{\frac{1}{g-1} - \int_{t_0}^s c(t - t' + t_0) w(t') \ dt'} \ ds   \right\} \ .
\end{split}
\end{equation}
Moreover, if the parameters $k$, $\gamma$, and $c$ are all time-independent and non-zero, the function $w(s)$ is explicitly
\begin{equation}
w(s) = e^{(c - \gamma) (s - t_0)}
\end{equation}
and the transition probability can be rewritten as
\begin{equation}
\begin{split}
P =& \left( \frac{\frac{\gamma}{c} - 1}{\frac{\gamma}{c} - w} \right)^{k/c}  \frac{\left( 1 - w \right)^{x - \xi}}{\left(\frac{\gamma}{c} - w \right)^x}  \times \\
 &\times \sum_{j = 0}^{\xi} \binom{\xi}{j} \frac{\left( j + k/c \right)_x}{x!} \left[ 1 - \frac{\gamma}{c} w \right]^{\xi - j}  \left[ \frac{w \left( \frac{\gamma}{c} - 1 \right)^2}{\frac{\gamma}{c} - w} \right]^j  
\end{split}
\end{equation}
where $(y)_x := (y)(y + 1) \cdots (y + x - 1)$ is the Pochhammer symbol/rising factorial. The generating function in this case reduces to
\begin{equation}
\psi(g) = \left[ \frac{1 + (g-1) \frac{c - \gamma w}{c - \gamma}}{1 - (g - 1) \frac{c (w - 1)}{c - \gamma}} \right]^{\xi} \ \frac{1}{\left[ 1 - (g - 1) \frac{c (w - 1)}{c - \gamma}  \right]^{k/c}} \ .
\end{equation}

\end{theorem}
\begin{proof} The transition probability and probability generating function can be obtained independently of one another using the Doi-Peliti approach described in the following sections. Alternatively, one can verify directly that the probability generating function solves Eq. \ref{eq:bdaGFeqn}, and use the definition of the generating function to find the transition probability. Because the expression for the transition probability is somewhat complicated, verifying it directly is not recommended. \qed
\end{proof}

It is expected that this solution reduces to familiar distributions in certain limits; in particular, as Jahnke and Huisinga originally point out, it should interpolate between a binomial distribution, a Poisson distribution, and a negative binomial distribution. Indeed it does, with these special cases corresponding to the $k = c = 0$ (pure death), $\gamma = c = 0$ (pure birth), and $k = \gamma = 0$ (pure autocatalysis) limits, respectively. We formalize this in the following corollary.

\begin{corollary}[Limiting behavior of the birth-death-autocatalysis transition probability] \label{cor:three_cases}

The transition probability $P(x, t; x_0, t_0)$ becomes (i) binomial in the limit that $k = c = 0$, (ii) Poisson in the limit that $\gamma = c = 0$, and (iii) negative binomial in the limit that $k = \gamma = 0$. That is, 
\begin{equation}
\lim_{k, c \to 0} P(x, t; x_0, t_0)  = \binom{\xi}{x} \left[ w(t) \right]^x \left[ 1 - w(t) \right]^{\xi - x}
\end{equation}
for $x \leq \xi$ and $0$ otherwise, i.e. a binomial distribution; 
\begin{equation}
\lim_{\gamma, c \to 0} P(x, t; x_0, t_0) = \frac{\lambda(t)^{x - \xi} e^{- \lambda(t)}}{(x - \xi)!} 
\end{equation}
for $x \geq \xi$ and $0$ otherwise, i.e. a (shifted) Poisson distribution; and
\begin{equation}
\lim_{k, \gamma \to 0} P(x, t; x_0, t_0) = \binom{x - 1}{\xi - 1} \left[ w(t) \right]^{\xi} \left[ 1 - w(t) \right]^{x - \xi}
\end{equation}
which is nonzero only for $x \geq \xi$. As Jahnke and Huisinga note in their Sec. 6, this is a shifted negative binomial distribution.
\end{corollary}
\begin{proof}
These limits can be taken directly, and are relatively straightforward; see Sec. \ref{sec:bdacalc} for the calculations. \qed
\end{proof}
It is also true that the solution of this problem reduces to a birth-death process in the $c \to 0$ limit (i.e. the limit in which autocatalysis no longer happens). This limit is harder to take properly (it is not clear how to take it using $P(x, t; x_0, t_0)$, so the generating function must be used), and the resulting distribution is less special than the Poisson, binomial, or negative binomial distributions, but it is still interesting; for these reasons, we formalize it as its own corollary below.

\begin{corollary}[No autocatalysis limit reduces to birth-death]
In the $c \to 0$ limit, the transition probability $P(x, t; x_0, t_0)$ and generating function $\psi(g, t)$ reduce to the ones for the chemical birth-death process described in Theorem \ref{thm:bd}. 
\end{corollary}
\begin{proof}
The easiest way to do this is to show that the generating functions correspond in the $c \to 0$ limit. Because this is straightforward, we will not show the calculation explicitly. \qed
\end{proof}

If autocatalysis happens tends to happen more frequently than degradation (i.e. if $c > \gamma$), the number of molecules almost surely blows up to infinity in the long time limit. However, if degradation tends to overtake autocatalysis (i.e. if $\gamma > c$), then the steady state distribution exists and is nontrivial. 

\begin{corollary}[Long time behavior of birth-death-autocatalysis]
Let $k$, $\gamma$, and $c$ be time-independent, and suppose that $\gamma > c$. In the long time limit, we have:
\begin{equation}
\begin{split}
P_{ss}(x) = \left( \frac{\gamma - c}{\gamma} \right)^{k/c} \frac{\left( \frac{c}{\gamma} \right)^x}{x!} \left( \frac{k}{c}  \right)_x \\
\psi_{ss}(g) =  \left( \frac{\gamma - c}{\gamma} \right)^{k/c} \frac{1}{\left[ 1 - \frac{c g}{\gamma} \right]^{k/c}} \ .
\end{split}
\end{equation}
Moreover, these reduce to the $P_{ss}$ and $\psi_{ss}$ for the chemical birth-death process in the $c \to 0$ limit, i.e. 
\begin{equation}
\begin{split}
\left( \frac{\gamma - c}{\gamma} \right)^{k/c} \frac{\left( \frac{c}{\gamma} \right)^x}{x!} \left( \frac{k}{c}  \right)_x &\to \frac{\mu^x}{x!} e^{- \mu} \\
\left( \frac{\gamma - c}{\gamma} \right)^{k/c} \frac{1}{\left[ 1 - \frac{c g}{\gamma} \right]^{k/c}} &\to  e^{(g - 1) \mu} \ .
\end{split}
\end{equation}
\end{corollary}

\begin{proof} The simplest way to find $P_{ss}$ is to solve the steady state CME directly (i.e. set $\partial P/\partial t = 0$ and solve the resulting recurrence relation). Alternatively, noting that $w \to 0$, one can straightforwardly take the $t \to \infty$ limit of our result from Theorem \ref{thm:bda}. The $c \to 0$ is also easy to take. \qed
\end{proof}

%Birth-death-autocatalysis generating function:
%\begin{equation}
%\psi(g) = \left[ \frac{1 + (g-1) \frac{c - \gamma w}{c - \gamma}}{1 - (g - 1) \frac{c (w - 1)}{c - \gamma}} \right]^{\xi} \ \frac{1}{\left[ 1 - (g - 1) \frac{c (w - 1)}{c - \gamma}  \right]^{k/c}}
%\end{equation}

\subsection{Zero and first order reactions}

%While obtaining the solution to the birth-death-autocatalysis system is an interesting product of applying the Doi-Peliti approach, it is also very special: it is a one-dimensional problem involving only three chemical reactions. 

While obtaining the solution to the birth-death-autocatalysis system is certainly interesting, the problem itself is quite special: it is one-dimensional, and involves only three chemical reactions. What else can the Doi-Peliti approach be used to solve? What kind of sets of chemical reactions are tractable?

The full potential of the Doi-Peliti approach is not clear. As a partial answer to this question, however, we offer a result of somewhat shocking generality: a formal solution to the CME of any system whose reaction list only contains zero and first order reactions. By zero order reactions, we mean reactions like
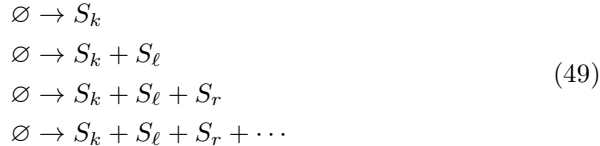
\begin{equation}
\begin{split}
\varnothing &\rightarrow S_k \\
\varnothing &\rightarrow S_k + S_{\ell} \\
\varnothing &\rightarrow S_k + S_{\ell} + S_r \\
\varnothing &\rightarrow S_k + S_{\ell} + S_r + \cdots
\end{split}
\end{equation}
and so on, i.e. reactions requiring no molecules as input. By first order reactions, we mean reactions like
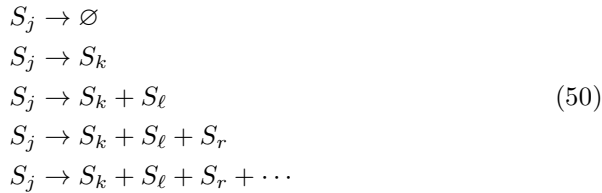
\begin{equation}
\begin{split}
S_j &\rightarrow \varnothing \\
S_j &\rightarrow S_k \\
S_j &\rightarrow S_k + S_{\ell}  \\
S_j &\rightarrow S_k + S_{\ell} + S_r \\
S_j &\rightarrow S_k + S_{\ell} + S_r + \cdots 
\end{split}
\end{equation}
and so on, i.e. reactions requiring exactly one molecule as input. The birth reactions described in the previous sections are examples of zero order reactions, while the death and conversion reactions are examples of first order reactions. Other biologically relevant examples of first order reactions include catalytic production ($S_j \to S_j + S_k$, $j \neq k$) and splitting ($S_j \to S_k + S_{\ell}$, $j \neq k$, $j \neq \ell$).

The list of all possible zero and first order reactions also includes many reactions that are almost certainly \textit{not} biologically relevant---for example, the reaction where one molecule splits into \textit{exactly} one hundred molecules with no intermediate splitting.

The CME of this system is somewhat tedious to write down, so we will instead note that the PDE satisfied by the generating function can be written in the form
\begin{equation} \label{eq:arb_gf_PDE}
\begin{split}
\frac{\partial \psi(\mathbf{g}, t)}{\partial t} =& \sum_{\nu_1, ..., \nu_n} \alpha_{\nu_1, ..., \nu_n}(t) \left( g_1 - 1 \right)^{\nu_1} \cdots \left(g_n - 1 \right)^{\nu_n} \psi(\mathbf{g}, t) \\
&+ \sum_k \sum_{\nu_1, ..., \nu_n} \beta^k_{\nu_1, ..., \nu_n}(t) \left( g_1 - 1 \right)^{\nu_1} \cdots \left( g_n - 1 \right)^{\nu_n} \frac{\partial \psi(\mathbf{g}, t)}{\partial g_k}
\end{split}
\end{equation}
where the precise form of the coefficients $\alpha_{\nu_1, ..., \nu_n}(t)$ and $\beta^j_{\nu_1, ..., \nu_n}(t)$ are determined by the details of one's list of reactions. Our main result for this class of systems is the following.

\begin{theorem}[Arbitrary combinations of zero and first order reactions] \label{thm:arb}

Let $\mathbf{q}(s)$ be the solution of
\begin{equation} \label{eq:zero_one_ODEs}
\dot{q}_j(s) = - i \sum_{\nu_1, ..., \nu_n} \beta^j_{\nu_1, ..., \nu_n}(t - s + t_0) \left[ i q_1(s) \right]^{\nu_1} \cdots \left[ i q_n(s) \right]^{\nu_n} \ , \ q_j(t_0) = p^f_j
\end{equation}
for some $\mathbf{p}^f \in \mathbb{R}^n$, with $s \in [t_0, t]$. Then if $P(\mathbf{x}, t_0) = \delta(\mathbf{x} - \boldsymbol{\xi})$ for some $\boldsymbol{\xi} \in \mathbb{N}^n$, we have
\begin{equation} 
P = \int_{\mathbb{R}^n} \frac{d\mathbf{p}^f}{(2\pi)^n} \ \frac{\left[ \mathbf{1} + i \mathbf{q}(t) \right]^{\boldsymbol{\xi}} e^{ \int_{t_0}^t \sum \alpha_{\nu_1, ..., \nu_n}(t - s + t_0) \left[ i q_1(s) \right]^{\nu_1} \cdots \left[ i q_n(s) \right]^{\nu_n}   \ ds }}{(\mathbf{1} + i \mathbf{p}^f)^{\mathbf{x} + \mathbf{1}}} 
\end{equation}
for the transition probability, and 
\begin{equation} 
\begin{split}
\psi(\mathbf{g}, t) =&  \left[ \mathbf{1} + i \mathbf{q}(t)  \right]^{\boldsymbol{\xi}}   \times  \\ &\times  \left.  e^{ \int_{t_0}^t \sum \alpha_{\nu_1, ..., \nu_n}(t - s + t_0) \left[ i q_1(s) \right]^{\nu_1} \cdots \left[ i q_n(s) \right]^{\nu_n}   \ ds } \right|_{\mathbf{p}^f = - i (\mathbf{g} - \mathbf{1})}
\end{split}
\end{equation}
for the probability generating function. 
\end{theorem}
\begin{proof}
In principle, the transition probability or generating function could be verified by substitution directly; however, this seems very difficult. A better way is discussed in Sec. \ref{sec:propview}.
 \qed
\end{proof}

In some sense all other results in this paper are corollaries of this result; still, it is helpful to study the simpler cases in their own right, both to double-check the correctness of this more general result, and to develop a sense for how to derive these solutions. 

While this result is perhaps shockingly general, it is also incredibly formal. For most systems of interest, it is likely that reducing the problem of solving the CME (an infinite number of coupled linear ODEs) to the problem of solving $n$ coupled nonlinear ODEs is not much of an improvement. However, in some cases legitimate simplification seems possible: monomolecular systems and the birth-death-autocatalysis system addressed in the previous two theorems are clearly examples. It is not clear whether there are large classes of more complicated systems for which the ODEs given by Eq. \ref{eq:zero_one_ODEs} are solvable, but searching for them seems like a promising avenue for future research on solving the CME. 

While we have not pursued this question, it is also possible that solving these ODEs numerically could yield new insights for efficiently solving the CME on computers.

\section{Reframing the problem and basic Doi-Peliti formalism}
\label{sec:reframe}

\subsection{A brief digression on notation}
\label{sec:digression}

Although it is not normally used in the study of stochastic processes, it is the author's strong belief that the bra-ket notation originally developed for quantum mechanics is most appropriate here. Because this notation makes it harder for most mathematicians to read this paper, here we will briefly argue why this is necessary. The reader who wishes to review the basics of bra-ket notation, and see how it compares to more standard mathematical notation for things like vector spaces and inner products, should refer to Appendix \ref{sec:notation_comp}. 

Given that the problems we are attempting to solve are quite complicated, carefully choosing notation is important; a bad choice of notation would clutter our already complicated arguments, making them nearly impossible to understand. We would like notation that (i) \textit{generalizes} cleanly to complicated systems in arbitrarily many dimensions; (ii) \textit{simplifies} the construction of the Doi-Peliti path integral; and (iii) is \textit{suggestive} of the operations we want to take, and not suggestive of operations that are not valid. 
 
%it eases notation, makes it easy to repeatedly apply the identity operator (c.f. the derivation of the path integral expression for the propagator $U$), and is suggestive for the inner products we are using. 
Let us say more about each heuristic requirement:
\begin{enumerate}
\item \textbf{Generalizes}: We will usually work with systems for which there are $n$ distinct chemical species, where $n \geq 1$ is some positive integer. In the next section, we will see that this forces us to work in a Hilbert space where each basis vector can be identified with an element of $\mathbb{N}^n$. We will eventually need notation for each basis vector, as well as for sums over $\mathbb{N}^n$, integrals over $\mathbb{R}^n$, and eigenvectors with eigenvalues $\mathbf{z} \in \mathbb{C}^n$. Denoting basis vectors, eigenvectors, and things like sums and integrals cleanly in arbitrarily many dimensions is easy using bra-ket notation.
\item \textbf{Simplifies}: Constructing the Doi-Peliti path integral involves using many identity operators/resolutions of the identity (see Sec. \ref{sec:roi}). This is cleanest with bra-ket notation, and using alternative notation obfuscates these steps.
\item \textbf{Suggestive/not confusing}: We will have to compute many inner products, as well as different kinds of inner products. Bra-ket notation allows them to be denoted simply, e.g. the inner product of $\ket{x}$ and $\ket{y}$ is $\braket{x}{y}$. If we used generating function notation, where we have $g^x$ instead of $\ket{x}$ and $g^y$ instead of $\ket{y}$, we would have to define strange operations like $g^x \cdot g^y = x! \delta(x - y)$. Moreover, this notation suggests operations like $g^x g^y = g^{x + y}$ are valid, although they are not. Vector notation (using e.g. $e_x$ and $e_y$ to denote basis vectors) would be somewhat confusing, because we are already considering vectors like $\mathbf{x} \in \mathbb{N}^n$ to denote particular states of our system. 
\end{enumerate}
Aside from issues of notation, there is a `deeper' reason this path integral requires special notation, whereas for others (see e.g. \cite{bressloff2014}) standard notation and a Chapman-Kolmogorov-based argument suffices. Most of the time, when path integrals are applied to stochastic processes or mathematical biology, what one is \textit{really} doing is applying the Chapman-Kolmogorov equation many times. This has the interpretation that one is imagining all possible paths from one state to another state and appropriately discretizing them. 

This kind of path integral is qualitatively different. It involves expanding an abstract object (rather than the transition probability itself) in terms of coherent states (which we will define later), which are themselves kind of abstract Poisson-like distributions. There does not seem to be the same obvious interpretation linking this path integral to the Chapman-Kolmogorov equation, or to all possible paths through state space (the space of all possible configurations of the system, i.e. $\mathbb{N}^n$ for a system with $n$ distinct chemical species). \\

%- conducive to inner products
%- reasonable labeling of all basis vectors
%- should not suggest invalid operations (addition, multiplication)
%
%
%NEED TO ADD A SMALL SECTION TO THE BEGINNING HERE, WHERE I EXPLAIN THE NEED FOR BRA-KET NOTATION AND HOW IT WORKS
%
%THEN REFER THE READER TO AN APPENDIX WHERE NOTATION PARALLELS BETWEEN MATH (INCLUDING STOCHASTIC PROCESSES) AND QUANTUM MECHANICS ARE DISCUSSED IN MORE DETAIL

\subsection{Hilbert space, the generating function, and basic operators}

In order to apply the Doi-Peliti technique, we first need to rewrite the CME in terms of states and operators in a certain Hilbert space. Consider an infinite-dimensional Hilbert space spanned by the $\ket{\mathbf{x}}$ states/basis vectors (where $\mathbf{x} = (x_1, ..., x_n)^T \in \mathbb{N}^n$), in which an arbitrary state $\ket{\phi}$ is written
\begin{equation} \label{eq:arbstate}
\ket{\phi} =  \sum_{x_1 = 0}^{\infty} \cdots \sum_{x_n = 0}^{\infty} c(\mathbf{x}) \ket{\mathbf{x}} 
\end{equation}
for some generally complex-valued coefficients $c(\mathbf{x})$. These states can be added and multiplied by (complex) scalars in the usual way. There is one basis vector for every possible state of the system (recall that the state space of the system is $\mathbb{N}^n$), e.g. $\ket{0 0 \cdots 0}$, $\ket{0 1 \cdots 0}$, $\ket{20, 45, 1, \cdots 10}$, and so on; one interpretation of these objects is that they encode a certain generalization of probability distributions on the state space, given that they assign every $\mathbf{x} \in \mathbb{N}^n$ a complex number. 

It should be noted that the basis vectors cannot be combined since they represent distinct directions in the Hilbert space, i.e. $\ket{\mathbf{x}} + \ket{\mathbf{y}} \neq \ket{\mathbf{x} + \mathbf{y}}$. We will denote the zero vector by $0$, which we emphasize for clarity's sake is distinct from the basis vector $\ket{\mathbf{0}}$ (e.g. $\ket{\mathbf{0}} + 0 = \ket{\mathbf{0}}$). The relevant inner products (without which this would just be a vector space) will be described in a few sections. 

To ease notation, we remind the reader that we will write
\begin{equation}
\sum_{\mathbf{x}} := \sum_{x_1 = 0}^{\infty} \cdots \sum_{x_n = 0}^{\infty} \ .
\end{equation}
The state we are principally interested in is the generating function, which is essentially the function $\psi(\mathbf{g}, t)$ described earlier, but using different notation. For more information on their correspondence, see Appendix \ref{sec:notation_comp}.

\begin{definition}
The \textit{generating function} is defined to be the state
\begin{equation} \label{eq:gf}
\ket{\psi(t)} := \sum_{\mathbf{x}}  P(\mathbf{x}, t) \ket{\mathbf{x}} 
\end{equation}
where $P(\mathbf{x}, t)$ is some solution to the CME (i.e. its precise form depends on the chosen initial condition $P(x, t_0)$). 
\end{definition}
For the rest of this paper, we will only be concerned with the case where $P(\mathbf{x}, t_0) = \delta(\mathbf{x} - \boldsymbol{\xi})$ for some $\boldsymbol{\xi} \in \mathbb{N}^n$, so we will always assume that $\ket{\psi(t_0)} = \ket{\boldsymbol{\xi}}$.

%which, by construction, contains exactly the same information that the probability distribution $P(\mathbf{x}, t)$ does. Although expressed in terms of different notation, this generating function is equivalent to the function $\psi(\mathbf{g}, t)$ described in Sec. \ref{sec:pstatement}. 

Because the generating function $\ket{\psi(t)}$ depends on $P(\mathbf{x}, t)$, whose dynamics are controlled by the CME, $\ket{\psi(t)}$ also has dynamics; we can write the equation controlling them (its `equation of motion') in the form 
\begin{equation} \label{eq:EOM}
\frac{\partial \ket{\psi}}{\partial t} = \hat{H} \ket{\psi} 
\end{equation}
where the Hamiltonian operator $\hat{H}$ is a linear operator whose precise form depends on the CME. For the reader familiar with quantum mechanics, this is analogous to the equation of motion for a quantum mechanical state. In any case, it is this equation that we will solve instead of the CME.

It may or may not be helpful for the reader to think of $\hat{H}$ as a (possibly infinite-dimensional) matrix. Although it is infinite-dimensional in essentially every case we care about in this paper, it would literally be a matrix if we were solving a CME with a finite state space. One example of a problem with a finite state space is the pure conversion process ($A \leftrightarrow B$), which involves $A$ molecules and $B$ molecules randomly converting between each other; it has a finite state space because the total number of molecules remains constant. 

Pressing the analogy between $\hat{H}$ and matrices, we have the usual formal solution for the generating function $\ket{\psi(t)}$ in terms of the (time-ordered) exponential of $\hat{H}$.

\begin{proposition}[Formal solution for the generating function] \label{prop:formal}
The equation of motion for the generating function $\ket{\psi(t)}$ (Eq. \ref{eq:EOM}) has the formal solution
\begin{equation} \label{eq:EOMformal}
\begin{split}
\ket{\psi(t)} &= \hat{T} e^{\int_{t_0}^t \hat{H}(t') dt'} \ket{\psi(t_0)}  \\
&= \sum_{j = 0}^{\infty} \frac{1}{j!} \int_{t_0}^t \cdots \int_{t_0}^t \hat{T} \left[ \hat{H}(t_1) \cdots \hat{H}(t_n) \right] \ dt_1 \cdots dt_n \\
&= 1 + \int_{t_0}^t \hat{H}(t_1) dt_1 + \frac{1}{2} \int_{t_0}^t \int_{t_0}^t \hat{T} \left[ \hat{H}(t_1) \hat{H}(t_2) \right] dt_1 dt_2 + \cdots
\end{split}
\end{equation}
where $\hat{T}$ is the time-ordering symbol, whose action on a product of operators is defined to be
\begin{equation}
\hat{T} \left[ \hat{\mathcal{A}}_1(t) \hat{\mathcal{A}}_2(t') \right] := \left\{ \begin{array}{cr} 
\hat{\mathcal{A}}_1(t) \hat{\mathcal{A}}_2(t') & \hspace{0.5in} t \geq t' \\
\hat{\mathcal{A}}_2(t') \hat{\mathcal{A}}_1(t) & \hspace{0.5in} t < t'
\end{array} \right. \ .
\end{equation}
\end{proposition}
\begin{proof}
Substitute this expression for the generating function directly into Eq. \ref{eq:EOM}. While the presence of the time-ordering symbol makes this more subtle than it would be in the case of a time-independent Hamiltonian (i.e. in the case where the reaction parameters were all time-independent), this exercise is standard, so we will not go through this in detail. \qed
\end{proof}

\begin{corollary}
For time-independent $\hat{H}$, the above formal solution reduces to 
\begin{equation}
\ket{\psi(t)} = e^{\hat{H} (t - t_0)} \ket{\psi(t_0)} \ .
\end{equation}
\end{corollary}
\begin{proof}
One can either show that this solves the equation of motion directly, or simplify the result above. \qed
\end{proof}
Fortunately, we will never have to work with a time-ordered exponential of operators directly. The first salient consequence of the formal solution for us is that
\begin{equation}
\ket{\psi(t + \Delta t)} \approx \left[ 1 + \hat{H}(t) \Delta t \right] \ket{\psi(t)}
\end{equation}
for sufficiently small $\Delta t$, with the approximation becoming exact in the $\Delta t \to 0$ limit. Notice that this also matches what we would find by naively approximating the time derivative with a finite difference in Eq. \ref{eq:EOM}.

This formal solution motivates defining the time evolution operator, which carries the solution at time $t_1$ (the state $\ket{\psi(t_1)}$) to the solution at time $t_2$ (the state $\ket{\psi(t_2)}$).

\begin{definition} 
The \textit{time evolution operator} $\hat{U}(t_2, t_1)$ is defined as
\begin{equation}
\begin{split}
\hat{U}(t_2, t_1) :=& \ \hat{T} e^{\int_{t_1}^{t_2} \hat{H}(t') dt'}  \\
=& 1 + \int_{t_1}^{t_2} \hat{H}(s_1) ds_1 + \frac{1}{2} \int_{t_1}^{t_2} \int_{t_1}^{t_2} \hat{T} \left[ \hat{H}(s_1) \hat{H}(s_2) \right] ds_1 ds_2 + \cdots
\end{split}
\end{equation}
for any two times $t_1 \leq t_2$. In terms of the time evolution operator, the formal solution for $\ket{\psi(t)}$ can be written
\begin{equation} \label{eq:EOM_U}
\ket{\psi(t)} = \hat{U}(t, t_0) \ket{\psi(t_0)}  \ .
\end{equation}
\end{definition}
The second salient consequnce of Proposition \ref{prop:formal} is that this operator has an important composition property.
\begin{proposition}[Composition property of the time evolution operator]
The time evolution operator $\hat{U}$ has the property that
\begin{equation} \label{eq:Ucompose}
\hat{U}(t_2, t_1) = \hat{U}(t_2, t') \hat{U}(t', t_1)
\end{equation}
for any time $t'$ with $t_1 \leq t' \leq t_2$. 
\end{proposition}
\begin{proof}
This is most easily seen using the infinite series form of the time evolution operator $\hat{U}$, by expanding both sides and showing that they match at each order. \qed
\end{proof}

Because we are interested in the dynamics of the generating function $\ket{\psi(t)}$, we need to introduce operators to act on it. In particular, we introduce the creation and annihilation operators, which we will later use to write the Hamiltonian operator.
\begin{definition}
Define the \textit{annihilation} and \textit{creation operators} $\hat{a}_j$ and $\hat{\pi}_j$ for all $j = 1, ..., n$ as the operators whose action on a basis vector $\ket{\mathbf{x}}$ is 
\begin{equation} \label{eq:aops}
\begin{split}
\hat{a}_j \ket{\mathbf{x}} &= x_j \ket{\mathbf{x} - \boldsymbol{\epsilon}_j}  \\
\hat{\pi}_j \ket{\mathbf{x}} &= \ket{\mathbf{x} + \boldsymbol{\epsilon}_j}
\end{split}
\end{equation}
where we remind the reader that $\boldsymbol{\epsilon}_j$ is the $n$-dimensional vector with a $1$ in the $j$th place and zeros everywhere else. 
\end{definition}

It is easy to show that these operators satisfy the commutation relations analogous to those seen in quantum mechanics (for example, in the ladder operator treatment of the harmonic oscillator \cite{griffiths2018}, or in the canonical quantization approach to quantum field theory \cite{schwartz2014}). These properties will be used in calculations a few times throughout this paper. 

\begin{proposition}
Recall that, for two operators $\hat{\mathcal{A}}_1$ and $\hat{\mathcal{A}}_2$, their commutator is defined to be $[\hat{\mathcal{A}}_1, \hat{\mathcal{A}}_2] := \hat{\mathcal{A}}_1 \hat{\mathcal{A}}_2 - \hat{\mathcal{A}}_2 \hat{\mathcal{A}}_1$. The creation and annihilation operators satisfy the commutation relations
\begin{equation}
[\hat{a}_j, \hat{\pi}_k] = \delta(j - k) \ , \ [\hat{a}_j, \hat{a}_k] = [\hat{\pi}_j, \hat{\pi}_k] = 0 \ .
\end{equation}

\end{proposition}
\begin{proof}
Use their definitions to straightforwardly show this. \qed
\end{proof}

%which are precisely the same as the commutation relations satisfied by the creation and annihilation operators familiar from quantum mechanics and quantum field theory \cite{griffiths2018,schwartz2014}.

%\section{Basic Doi-Peliti formalism}
%\label{sec:dpbasics}

In essence, the Doi-Peliti approach to solving Eq. \ref{eq:EOM} involves using many coherent state `resolutions of the identity' (a phrase we will define shortly) to rewrite Eq. \ref{eq:EOMformal} as a coherent state path integral. Once that path integral is evaluated, quantities like moments and $P(\mathbf{x}, t)$ can be recovered by manipulating the path integral solution in specific ways. In order to follow this prescription, we will need to define coherent states, define inner products, and construct associated resolutions of the identity; that is our next task. 

\subsection{Coherent states}

Because we will be expressing the Hamiltonian operator in terms of creation and annihilation operators, it is convenient to work in terms of states that behave simply when acted upon by these operators. These are the so-called coherent states, which are often used to study the semiclassical limit of quantum mechanics. Here, we will only care about them for their algebraic properties; while their biological meaning is not completely obscure (they are essentially states that correspond to Poisson distributions), thinking about it is not necessary in what follows.

\begin{definition}
Let $\mathbf{z} = (z_1, ..., z_n)^T \in \mathbb{C}^n$. A \textit{coherent state} is a state
\begin{equation}
\ket{\mathbf{z}} := \sum_{\mathbf{y}} c(\mathbf{y}) \ \ket{\mathbf{y}}
\end{equation} 
satisfying 
\begin{equation}
\begin{split}
\hat{a}_j \ket{\mathbf{z}} &= z_j \ket{\mathbf{z}} \hspace{0.5in} \text{ for all } j = 1, ..., n \\
\sum_{\mathbf{y}} c(\mathbf{y}) &= 1
\end{split}
\end{equation}
i.e. it is an eigenstate of all annihilation operators $\hat{a}_j$, and it has a specific normalization. 
\end{definition}
%\begin{equation}
%\hat{a}_j \ket{\mathbf{z}} = z_j \ket{\mathbf{z}}
%\end{equation}
%for all $j = 1, ..., n$, with $\mathbf{z} = (z_1, ..., z_n) \in \mathbb{C}^n$. 

By imposing the eigenstate constraint on an arbitrary state, it is straightforward to determine the coefficients $c(\mathbf{y})$ explicitly. Coherent states can also be written in terms of a specific combination of creation operators acting on the `vacuum' state $\ket{\mathbf{0}}$. We make these statements more precise in the following proposition. 

\begin{proposition}  \label{prop:cs_forms}
The coherent state $\ket{\mathbf{z}}$ can explicitly be written in the following two equivalent forms:
\begin{enumerate}[label=(\roman*)]
\item \begin{equation} \ket{\mathbf{z}} = \sum_{\mathbf{y}} \frac{z_1^{y_1} \cdots z_n^{y_n}}{y_1! \cdots y_n!} e^{-( z_1 + \cdots + z_n)} \ \ket{\mathbf{y}} = \sum_{\mathbf{y}} \frac{\mathbf{z}^{\mathbf{y}}}{\mathbf{y}!} e^{- \mathbf{z} \cdot \mathbf{1}} \ \ket{\mathbf{y}} \end{equation}
\item \begin{equation} \ket{\mathbf{z}} = \sum_{\mathbf{y}} \frac{[z_1 (\hat{\pi}_1 - 1)]^{y_1} \cdots [z_n (\hat{\pi}_n - 1)]^{y_n}}{y_1! \cdots y_n!} \ket{\mathbf{0}} = e^{\mathbf{z} \cdot (\hat{\boldsymbol{\pi}} - \mathbf{1})} \ket{\mathbf{0}} \end{equation}
\end{enumerate}
where $\hat{\boldsymbol{\pi}} := (\hat{\pi}_1, ..., \hat{\pi}_n)^T$. 
\end{proposition}
\begin{proof}
Showing (i) is straightforward. To show (ii), first note that $[\hat{a}_j, (\hat{\pi}_j - 1)^y] = y (\hat{\pi}_j - 1)^{y - 1}$ for all $y \in \mathbb{N}$, a useful commutator result that can be proved by induction. Using this, along with the facts that $\hat{a}_j$ commutes with $\hat{\pi}_k$ for $k \neq j$ and the $\hat{\pi}_k$ all commute with each other, we have
\begin{equation}
\begin{split}
\hat{a}_j \ket{z} &=  e^{\mathbf{z} \cdot (\hat{\boldsymbol{\pi}} - \mathbf{1}) - z_j (\hat{\pi}_j - 1)} \sum_{y_j = 0}^{\infty} \frac{z_j^{y_j}}{y_j!} \hat{a}_j  (\hat{\pi}_j - 1)^{y_j} \ket{\mathbf{0}} \\
&=  e^{\mathbf{z} \cdot (\hat{\boldsymbol{\pi}} - \mathbf{1}) - z_j (\hat{\pi}_j - 1)} \sum_{y_j = 0}^{\infty} \frac{z_j^{y_j}}{y_j!} \left\{ (\hat{\pi}_j - 1)^{y_j} \hat{a}_j   +  y_j (\hat{\pi}_j - 1)^{y_j - 1} \right\} \ket{\mathbf{0}} \\
&= z_j \ e^{\mathbf{z} \cdot (\hat{\boldsymbol{\pi}} - \mathbf{1}) - z_j (\hat{\pi}_j - 1)} \sum_{y_j = 1}^{\infty} \frac{z_j^{y_j - 1} (\hat{\pi}_j - 1)^{y_j - 1}}{(y_j - 1)!} \ket{\mathbf{0}} \\
&= z_j \ket{\mathbf{z}} \ .
\end{split}
\end{equation}
Hence, this expression satisfies the eigenstate constraint. Noting that eigenstates are unique up to a proportionality constant, to show that it satisfies the normalization constraint (and hence is the same as the expression given by (i)), observe that
\begin{equation}
\begin{split}
e^{\mathbf{z} \cdot (\hat{\boldsymbol{\pi}} - \mathbf{1})} \ket{\mathbf{0}} &= \sum_{\mathbf{y}} \frac{z_1^{y_1} \cdots z_n^{y_n} \left[ (-1)^{y_1 + \cdots + y_n} + \cdots \right]}{y_1! \cdots y_n!} \ket{\mathbf{0}} \\
&= e^{- \mathbf{z} \cdot \mathbf{1}} \ket{\mathbf{0}} + \cdots
\end{split}
\end{equation}
i.e. the coefficient of $\ket{\mathbf{0}}$ is $e^{- \mathbf{z} \cdot \mathbf{1}}$, because every other term in the above expansion contains a creation operator. Because this is the same as the coefficient of $\ket{\mathbf{0}}$ in (i), we have equivalence. \qed
\end{proof}
In what follows, we will generally reserve the letters $\mathbf{z}$ and $\mathbf{p}$ for coherent states. 

%By imposing this constraint on an arbitrary state (Eq. \ref{eq:arbstate}), it is straightforward to show that the explicit form for $\ket{\mathbf{z}}$ must be proportional to
%\begin{equation}  \label{eq:cs}
%\ket{\mathbf{z}} = \sum_{\mathbf{y}} \frac{z_1^{y_1} \cdots z_n^{y_n}}{y_1! \cdots y_n!} e^{-( z_1 + \cdots + z_n)} \ \ket{\mathbf{y}} = \sum_{\mathbf{y}} \frac{\mathbf{z}^{\mathbf{y}}}{\mathbf{y}!} e^{- \mathbf{z} \cdot \mathbf{1}} \ \ket{\mathbf{y}}
%\end{equation}
%where the factor $e^{- \mathbf{z} \cdot \mathbf{1}}$ is a specific overall constant chosen for our later convenience, and where we again use the shorthand $\mathbf{x}! := x_1! \cdots x_n!$ and $\mathbf{z}^{\mathbf{y}} := z_1^{y_1} \cdots z_n^{y_n}$ to ease notation. We can also write a coherent state as
%\begin{equation} \label{eq:zpi}
%\ket{\mathbf{z}} = e^{\mathbf{z} \cdot (\hat{\boldsymbol{\pi}} - \mathbf{1})} \ket{\mathbf{0}}
%\end{equation}
%where $\hat{\boldsymbol{\pi}} := (\hat{\pi}_1, ..., \hat{\pi}_n)$. 

\subsection{Inner products}
\label{sec:innerproducts}

Now we will define two inner products on our Hilbert space: the exclusive product, and the Grassberger-Scheunert product. Both were introduced by Grassberger and Scheunert in a 1980 paper that clearly describes their motivation and properties \cite{grassberger1980}; we are calling their ``inclusive'' inner product the Grassberger-Scheunert product to recognize their contribution. 

Briefly, the exclusive product is useful for computing $P(\mathbf{x}, t)$, while the Grassberger-Scheunert product is useful for simplifying path integral calculations (specifically, we avoid having to perform a ``Doi shift'' \cite{cardy2008,weber2017}; see Eq. 3.4 of Peliti \cite{peliti1985} for an example of the Doi shift) and computing moments. We will use both inner products in solving the CME.

In this section and the following sections, the reader should keep in mind that $\mel{\mathbf{x}}{\hat{\mathcal{A}}(t)}{\mathbf{y}}$, where $\hat{\mathcal{A}}(t)$ is some possibly time-dependent operator, means the same as $\langle \mathbf{e}_{\mathbf{x}}, \hat{\mathcal{A}}(t) \mathbf{e}_{\mathbf{y}} \rangle$ in more standard notation. See Appendix \ref{sec:notation_comp} for more details.

%Define the exclusive product of two basis states by

\begin{definition}
Let $\ket{\mathbf{x}}$ and $\ket{\mathbf{y}}$ be basis vectors. Their \textit{exclusive product} is defined to be
\begin{equation} \label{eq:ex_basis}
\braket{\mathbf{x}}{\mathbf{y}}_{ex} := \mathbf{x}! \  \delta(\mathbf{x} - \mathbf{y}) \ .
\end{equation}
Extending this by linearity, define the exclusive product of two arbitrary states $\ket{\phi_1}$ and $\ket{\phi_2}$ as (c.f. Eq. \ref{eq:arbstate})
\begin{equation} \label{eq:ex}
\braket{\phi_2}{\phi_1}_{ex} =  \sum_{\mathbf{x}} \mathbf{x}! \  c_2^*(\mathbf{x}) c_1(\mathbf{x})  \ .
\end{equation}
\end{definition}

\begin{definition}
Let $\ket{\mathbf{x}}$ and $\ket{\mathbf{y}}$ be basis vectors, and define $\hat{\mathbf{a}} := (\hat{a}_1, ..., \hat{a}_n)^T$. Their \textit{Grassberger-Scheunert product} is defined to be
\begin{equation} \label{eq:in_basis}
\braket{\mathbf{x}}{\mathbf{y}} := \matrixel{\mathbf{x}}{e^{\hat{\boldsymbol{\pi}} \cdot \mathbf{1}} e^{\hat{\mathbf{a}} \cdot \mathbf{1}}}{\mathbf{y}}_{ex} = \sum_{\mathbf{k}} \frac{\mathbf{x}! \ \mathbf{y}!}{(\mathbf{x} - \mathbf{k})! \ (\mathbf{y} - \mathbf{k})! \ \mathbf{k}!}
\end{equation}
where the sum on the right is over all values of $\mathbf{k} \in \mathbb{N}^n$ with $k_j \leq \min(x_j, y_j)$ for all $j = 1, ..., n$. Extending this by linearity, define the Grassberger-Scheunert product of two arbitrary states $\ket{\phi_1}$ and $\ket{\phi_2}$ as
\begin{equation} \label{eq:in}
\braket{\phi_2}{\phi_1} =  \matrixel{\phi_2}{e^{\hat{\boldsymbol{\pi}} \cdot \mathbf{1}} e^{\hat{\mathbf{a}} \cdot \mathbf{1}}}{\phi_1}_{ex}   \ .
\end{equation}
\end{definition}

While it is not obvious just from looking at them, it is straightforward to show that the operator-based and sum-based definitions are equivalent (see Grassberger and Scheunert \cite{grassberger1980} and the appendix to Peliti \cite{peliti1985}). 

The primary reason these inner products are useful to define is that the creation and annihilation operators behave well under Hermitian conjugation with respect to them. 

\begin{proposition}[Hermitian conjugates of creation and annihilation operators]
Let $\ket{\mathbf{x}}$ and $\ket{\mathbf{y}}$ be basis vectors. With respect to the exclusive product, $\hat{a}_j$ and $\hat{\pi}_j$ are Hermitian conjugates of each other for all $j = 1, ..., n$, i.e.
\begin{equation}
\begin{split}
\left( \hat{a}_j \right)^{\dag} &= \hat{\pi}_j \\
\mel{\mathbf{x}}{\hat{a}_j}{\mathbf{y}}_{ex} &= \mel{\mathbf{y}}{\hat{\pi}_j}{\mathbf{x}}_{ex} \ .
\end{split}
\end{equation}
With respect to the Grassberger-Scheunert product, the Hermitian conjugate of $\hat{a}_j$ is
\begin{equation}
\left( \hat{a}_j \right)^{\dag} = \hat{\pi}_j - 1 
\end{equation}
for all $j = 1, ..., n$, i.e.
\begin{equation}
\mel{\mathbf{x}}{\hat{a}_j}{\mathbf{y}} = \mel{\mathbf{y}}{\hat{\pi}_j - 1}{\mathbf{x}} \ .
\end{equation}
\end{proposition}
\begin{proof}
Showing that $\hat{\pi}_j$ and $\hat{a}_j$ are Hermitian conjugates with respect to the exclusive product is straightforward given their definitions, so we will show that $\left( \hat{a}_j \right)^{\dag} = \hat{\pi}_j - 1$ with respect to the Grassberger-Scheunert product. 

Recall the result mentioned in the proof of Proposition \ref{prop:cs_forms} that  $[\hat{a}_j, (\hat{\pi}_j - 1)^y] = y (\hat{\pi}_j - 1)^{y - 1}$ for all $y \in \mathbb{N}$. Using just the same argument, one can show  $[\hat{a}_j, \hat{\pi}_j^y] = y (\hat{\pi}_j)^{y - 1}$ for all $y \in \mathbb{N}$. This, in turn, can be used to prove that
\begin{equation}
e^{\hat{\pi}_j} \hat{a}_j  = (\hat{a}_j - 1) e^{\hat{\pi}_j} \ .
\end{equation}
Let $\ket{\mathbf{x}}$ and $\ket{\mathbf{y}}$ be arbitrary basis vectors. Now we can say that
\begin{equation}
\begin{split}
\mel{\mathbf{x}}{\hat{a}_j}{\mathbf{y}} =& \mel{\mathbf{x}}{e^{\hat{\boldsymbol{\pi}} \cdot \mathbf{1}} e^{\hat{\mathbf{a}} \cdot \mathbf{1}} \hat{a}_j}{\mathbf{y}}_{ex} \\
=& \mel{\mathbf{x}}{e^{\hat{\pi}_j} \hat{a}_j \ e^{\hat{\boldsymbol{\pi}} \cdot \mathbf{1} - \hat{\pi}_j} e^{\hat{\mathbf{a}} \cdot \mathbf{1}} }{\mathbf{y}}_{ex} \\
=& \mel{\mathbf{x}}{ (\hat{a}_j - 1) e^{\hat{\pi}_j}  \ e^{\hat{\boldsymbol{\pi}} \cdot \mathbf{1} - \hat{\pi}_j} e^{\hat{\mathbf{a}} \cdot \mathbf{1}} }{\mathbf{y}}_{ex} \\
=& \mel{\mathbf{x} + \boldsymbol{\epsilon_j}}{ e^{\hat{\boldsymbol{\pi}} \cdot \mathbf{1}} e^{\hat{\mathbf{a}} \cdot \mathbf{1}} }{\mathbf{y}}_{ex} - \mel{\mathbf{x}}{ e^{\hat{\boldsymbol{\pi}} \cdot \mathbf{1}} e^{\hat{\mathbf{a}} \cdot \mathbf{1}} }{\mathbf{y}}_{ex}  \\
=& \braket{\mathbf{x} + \boldsymbol{\epsilon_j}}{\mathbf{y}} - \braket{\mathbf{x}}{\mathbf{y}}
\end{split}
\end{equation}
where we have used the fact that $\hat{\pi}_j$ and $\hat{a}_j$ are Hermitian conjugates with respect to the exclusive product in the next to last step. But this is the same as
\begin{equation}
\begin{split}
\mel{\mathbf{y}}{\hat{\pi}_j - 1}{\mathbf{x}} =& \braket{\mathbf{y}}{\mathbf{x} + \boldsymbol{\epsilon_j}} - \braket{\mathbf{y}}{\mathbf{x}}  
\end{split}
\end{equation}
because the Grassberger-Scheunert product of two basis vectors is symmetric. Hence, $\hat{a}_j$ and $\hat{\pi}_j - 1$ are Hermitian conjugates with respect to the Grassberger-Scheunert product. \qed
\end{proof}
Now let us compute some inner products that we will use later. 

\begin{proposition}[Useful inner products]
Let $\ket{\mathbf{x}}$ be a basis vector, and let $\ket{\mathbf{z}}$ and $\ket{\mathbf{p}}$ be coherent states. Then
\begin{enumerate}[label=(\roman*)]
\item \begin{equation} 
\braket{\mathbf{x}}{\mathbf{z}}_{ex}  = \mathbf{z}^{\mathbf{x}} e^{- \mathbf{z} \cdot \mathbf{1}} 
\end{equation}
\item \begin{equation} 
\braket{\mathbf{x}}{\mathbf{z}} =  \left( \mathbf{1} + \mathbf{z} \right)^{\mathbf{x}}
\end{equation}
\item \begin{equation} 
\braket{\mathbf{p}}{\mathbf{z}}_{ex} = e^{\mathbf{p}^* \cdot \mathbf{z} - \left( \mathbf{p}^* + \mathbf{z} \right) \cdot \mathbf{1}} \ .
\end{equation}
\item \begin{equation}
\braket{\mathbf{p}}{\mathbf{z}} =  e^{\mathbf{p}^* \cdot \mathbf{z}}
\end{equation}
\end{enumerate}
Moreover, we remind the reader that other results (e.g. $\braket{\mathbf{z}}{ \mathbf{x}} = (\mathbf{1} + \mathbf{z}^*)^{\mathbf{x}}$) can be obtained from the above ones by taking a complex conjugate.
\end{proposition}
\begin{proof}
First, the exclusive product of a basis state $\ket{\mathbf{x}}$ with a coherent state $\ket{\mathbf{z}}$ is
\begin{equation} \label{eq:zex}
\braket{\mathbf{x}}{\mathbf{z}}_{ex} = \sum_{\mathbf{y}} \frac{\mathbf{z}^{\mathbf{y}}}{\mathbf{y}!} e^{- \mathbf{z} \cdot \mathbf{1}} \ \braket{\mathbf{x}}{\mathbf{y}}_{ex} = \sum_{\mathbf{y}} \frac{\mathbf{z}^{\mathbf{y}}}{\mathbf{y}!} e^{- \mathbf{z} \cdot \mathbf{1}} \ \mathbf{x}! \  \delta_{\mathbf{x} \mathbf{y}} = \mathbf{z}^{\mathbf{x}} e^{- \mathbf{z} \cdot \mathbf{1}} \ .
\end{equation}
Next, the Grassberger-Scheunert product of a basis state $\ket{\mathbf{x}}$ with a coherent state $\ket{\mathbf{z}}$ is
\begin{equation} \label{eq:zin}
\begin{split}
\braket{\mathbf{x}}{\mathbf{z}} =& \matrixel{\mathbf{x}}{e^{\hat{\boldsymbol{\pi}} \cdot \mathbf{1}} e^{\hat{\mathbf{a}} \cdot \mathbf{1}}}{\mathbf{z}}_{ex}   \\
=& e^{\mathbf{z} \cdot \mathbf{1}} \matrixel{\mathbf{x}}{e^{\hat{\boldsymbol{\pi}} \cdot \mathbf{1}}}{\mathbf{z}}_{ex} \\
=& e^{(\mathbf{z} + \mathbf{1}) \cdot \mathbf{1}} \matrixel{\mathbf{x}}{e^{\left(\mathbf{z} + \mathbf{1}\right) \cdot (\hat{\boldsymbol{\pi}} - \mathbf{1})}}{\mathbf{0}}_{ex} \\
=& e^{(\mathbf{z} + \mathbf{1}) \cdot \mathbf{1}}  \braket{\mathbf{x}}{\mathbf{z} + \mathbf{1}}_{ex} \\
=& \left( \mathbf{1} + \mathbf{z} \right)^{\mathbf{x}}
\end{split}
\end{equation}
where we have used that $\ket{\mathbf{z}}$ is an eigenstate of the annihilation operators $\hat{a}_j$, the operator representation of $\ket{\mathbf{z}}$ from Proposition \ref{prop:cs_forms}, and Eq. \ref{eq:zex}. The exclusive product of two coherent states is
\begin{equation} \label{eq:zzex}
\braket{\mathbf{p}}{\mathbf{z}}_{ex} = \sum_{\mathbf{y}} \frac{\left( \mathbf{p}^* \right)^{\mathbf{y}}}{\mathbf{y}!} e^{- \mathbf{p}^* \cdot \mathbf{1}} \ \braket{\mathbf{y}}{\mathbf{z}}_{ex} = \sum_{\mathbf{y}} \frac{\left( \mathbf{p}^* \right)^{\mathbf{y}} \mathbf{z}^{\mathbf{y}} }{\mathbf{y}!} e^{- \left( \mathbf{p}^* + \mathbf{z} \right) \cdot \mathbf{1}} = e^{\mathbf{p}^* \cdot \mathbf{z} - \left( \mathbf{p}^* + \mathbf{z} \right) \cdot \mathbf{1}} \ .
\end{equation}
Finally, the Grassberger-Scheunert product of two coherent states is
\begin{equation}
\braket{\mathbf{p}}{\mathbf{z}} = \matrixel{\mathbf{p}}{e^{\hat{\boldsymbol{\pi}} \cdot \mathbf{1}} e^{\hat{\mathbf{a}} \cdot \mathbf{1}}}{\mathbf{z}}_{ex}  = e^{\left( \mathbf{p}^* + \mathbf{z} \right) \cdot \mathbf{1}} \braket{\mathbf{p}}{\mathbf{z}}_{ex} =  e^{\mathbf{p}^* \cdot \mathbf{z}}
\end{equation}
where we have used Eq. \ref{eq:zzex}.  \qed
\end{proof}

\subsection{Resolution of the identity}
\label{sec:roi}

The phrase `resolution of the identity' refers to a useful way to write the identity operator. In our case, we would like to write the identity operator in terms of coherent states, which will allow us to construct the Doi-Peliti path integral. The relevant proposition, using coherent states and the Grassberger-Scheunert product, is the following.

%The coherent states we defined, along with the Grassberger-Scheunert product, can be used to construct a resolution of the identity. It is

\begin{proposition}[Identity operator in terms of coherent states] \label{prop:roi}
Let $\ket{\mathbf{x}}$ be a basis vector, and $\ket{\mathbf{z}}$ and $\ket{- i\mathbf{p}}$ be coherent states. Then
\begin{equation}
\ket{\mathbf{x}} = \int_{[0, \infty)^n} d\mathbf{z} \int_{\mathbb{R}^n} \frac{d\mathbf{p}}{(2\pi)^n}   \ \ket{\mathbf{z}} \braket{- i \mathbf{p}}{\mathbf{x}} e^{-i \mathbf{z} \cdot \mathbf{p}} 
\end{equation}
i.e. 
\begin{equation} \label{eq:ROI}
1 = \int_{[0, \infty)^n} d\mathbf{z} \int_{\mathbb{R}^n} \frac{d\mathbf{p}}{(2\pi)^n}   \ \ket{\mathbf{z}} \bra{- i \mathbf{p}} e^{-i \mathbf{z} \cdot \mathbf{p}} 
\end{equation}
is the identity operator (because the relationship holds for basis vectors, it holds for all states by linearity). 
\end{proposition}

\begin{proof}
To establish Eq. \ref{eq:ROI}, first observe that
\begin{equation} \label{eq:ROIwork1}
\begin{split}
& \int_{[0, \infty)^n} d\mathbf{z} \int_{\mathbb{R}^n} \frac{d\mathbf{p}}{(2\pi)^n}   \ \ket{\mathbf{z}} \braket{- i \mathbf{p}}{\mathbf{x}} e^{-i \mathbf{z} \cdot \mathbf{p}} \\
=& \int_{[0, \infty)^n} d\mathbf{z} \int_{\mathbb{R}^n} \frac{d\mathbf{p}}{(2\pi)^n}   \ \ket{\mathbf{z}} \left( \mathbf{1} + i \mathbf{p} \right)^{\mathbf{x}} e^{-i \mathbf{z} \cdot \mathbf{p}} \\
=& \sum_{\mathbf{y}} \frac{1}{\mathbf{y}!} \ \ket{\mathbf{y}} \int_{[0, \infty)^n} d\mathbf{z} \ \mathbf{z}^{\mathbf{y}} \int_{\mathbb{R}^n} \frac{d\mathbf{p}}{(2\pi)^n}   \ \left( \mathbf{1} + i \mathbf{p} \right)^{\mathbf{x}} e^{-\mathbf{z} \cdot \left( \mathbf{1} + i \mathbf{p} \right)} \\
=& \sum_{\mathbf{y}} \frac{1}{\mathbf{y}!} \ \ket{\mathbf{y}} \int_{[0, \infty)^n} d\mathbf{z} \ \mathbf{z}^{\mathbf{y}} \left( - \frac{d}{d\mathbf{z}} \right)^\mathbf{x} \int_{\mathbb{R}^n} \frac{d\mathbf{p}}{(2\pi)^n}   \ e^{-\mathbf{z} \cdot \left( \mathbf{1} + i \mathbf{p} \right)} \\
\end{split} 
\end{equation}
for all basis kets $\ket{\mathbf{x}}$, where we remind the reader of the shorthand
\begin{equation} \label{eq:ddzshorthand}
\left( \frac{d}{d\mathbf{z}} \right)^\mathbf{x} := \left( \frac{d}{dz_1} \right)^{x_1} \cdots \left( \frac{d}{dz_n} \right)^{x_n}
\end{equation}
used to ease notation. Integrate the last line of Eq. \ref{eq:ROIwork1} by parts to obtain
\begin{equation} \label{eq:ROIwork2}
\begin{split}
& \sum_{\mathbf{y}} \frac{1}{\mathbf{y}!} \ \ket{\mathbf{y}} \int_{[0, \infty)^n} d\mathbf{z} \ \left( \frac{d}{d\mathbf{z}} \right)^\mathbf{x} \left[ \mathbf{z}^{\mathbf{y}} \right]  \ e^{-\mathbf{z}} \delta(\mathbf{z}) \\
=& \sum_{\mathbf{y}} \frac{1}{\mathbf{y}!} \ \ket{\mathbf{y}}  \ \mathbf{x}! \ \delta(\mathbf{y} - \mathbf{x}) \\
=& \ket{\mathbf{x}}
\end{split} 
\end{equation}
which confirms that Eq. \ref{eq:ROI} is a resolution of the identity. \qed
\end{proof}

We can use the coherent state resolution of the identity (Eq. \ref{eq:ROI}) we just constructed to rewrite our formal solution for $\ket{\psi(t)}$ (Eq. \ref{eq:EOMformal}). Applying it twice, we have the following result.

\begin{corollary}[Generating function in terms of coherent states]   \label{cor:Utogf}
The generating function can be written in the form
\begin{equation} \label{eq:psiU}
\ket{\psi(t)} = \int  \frac{d\mathbf{z}^f d\mathbf{p}^f}{(2\pi)^n} \frac{d\mathbf{z}^0 d\mathbf{p}^0}{(2\pi)^n}  \ \ket{\mathbf{z}^f} \matrixel{- i \mathbf{p}^f}{\hat{U}(t, t_0)}{\mathbf{z}^0} \braket{- i \mathbf{p}^0}{\psi(t_0)}   e^{-i \mathbf{p}^0 \cdot \mathbf{z}^0 -i \mathbf{p}^f \cdot \mathbf{z}^f } \ .
\end{equation}
\end{corollary}
\begin{proof}
Apply Proposition \ref{prop:roi} to the formal solution for $\ket{\psi(t)}$ (c.f. Eq. \ref{eq:EOM_U}) twice. \qed
\end{proof}

The object that appears in the middle of this expression is sufficiently important that it deserves its own name.

\begin{definition}
The \textit{propagator} is defined as the matrix element
\begin{equation}
U( i \mathbf{p}^f, t; \mathbf{z}^0, t_0) := \matrixel{- i \mathbf{p}^f}{\hat{U}(t, t_0)}{\mathbf{z}^0} \ .
\end{equation}
where $\ket{- i \mathbf{p}^f}$ and $\ket{\mathbf{z}^0}$ are coherent states, and $\hat{U}(t, t_0)$ is the time evolution operator. Usually, we will refer to it using the abbreviated notation $U( i \mathbf{p}^f, \mathbf{z}^0)$.
\end{definition}

%\section{Doi-Peliti path integral construction}
%\label{sec:constructpathint}

Now we will construct a coherent state path integral expression for the propagator---this is one of the most important equations in this paper, as it forms the basis of Doi-Peliti path integral calculations. 

\begin{proposition}[Path integral expression for the propagator]
The propagator $U( i \mathbf{p}^f, t; \mathbf{z}^0, t_0)$ is equal to the path integral
\begin{equation} \label{eq:Upathint}
\begin{split}
U = \ & \lim_{N \to \infty} \int \prod_{\ell = 1}^{N-1} \frac{d\mathbf{z}^{\ell} d\mathbf{p}^{\ell} }{(2\pi)^n} \exp\left\{ \sum_{\ell = 1}^{N-1}  - i \mathbf{p}^{\ell} \cdot ( \mathbf{z}^{\ell} - \mathbf{z}^{\ell - 1}) + \Delta t \ \mathcal{H}(i \mathbf{p}^{\ell}, \mathbf{z}^{\ell - 1}, t_{\ell - 1}) \right. \\
 & \left. + \Delta t \ \mathcal{H}(i \mathbf{p}^f, \mathbf{z}^{N - 1}, t_{\ell-1}) + i \mathbf{p}^f \cdot \mathbf{z}^{N-1}  \right\} 
\end{split}
\end{equation}
where $\Delta t := (t - t_0)/N$ and the Hamiltonian kernel $\mathcal{H}$ is defined as
\begin{equation}
\mathcal{H}(i \mathbf{p}, \mathbf{z}, t) := \matrixel{- i \mathbf{p}}{\hat{H}(t)}{\mathbf{z}} \ e^{- i \mathbf{p} \cdot \mathbf{z}}
\end{equation}
where $\ket{- i \mathbf{p}}$ and $\ket{\mathbf{z}}$ are coherent states  with $\mathbf{p}, \mathbf{z} \in \mathbb{R}^n$.
\end{proposition}

\begin{proof}
First write the time evolution operator $U(t, t_0)$ as a product of many time evolution operators using the composition property (Eq. \ref{eq:Ucompose}):
\begin{equation}
\hat{U}(t, t_0) = \hat{U}(t, t_{N-1}) \hat{U}(t_{N-1}, t_{N-2})  \cdots \hat{U}(t_1, t_0)
\end{equation}
where $t_{\ell} := t_0 + \ell \Delta t$ for $\ell = 0, ..., N$, and $\Delta t := (t - t_0)/N$. Now insert $(N-1)$ resolutions of the identity to write
\begin{equation}
U = \int  \prod_{\ell = 1}^{N-1} \frac{d\mathbf{z}^{\ell} d\mathbf{p}^{\ell}}{(2\pi)^n}    \matrixel{- i \mathbf{p}^f}{\hat{U}(t, t_{N-1})}{\mathbf{z}^{\ell - 1}} \cdots  \matrixel{- i \mathbf{p}^1}{\hat{U}(t_1, t_0)}{\mathbf{z}^0}  e^{-i \sum_{\ell = 1}^{N-1} \mathbf{p}^{\ell} \cdot \mathbf{z}^{\ell} } \ .
\end{equation}
To arrive at our desired path integral, all we must do is compute the matrix elements in the above equation. Assuming that $N$ is large enough that $\Delta t$ is very small, we have that
\begin{equation} \label{eq:Utaylor}
\hat{U}(t_{\ell}, t_{\ell - 1}) \approx 1 + \hat{H}(t_{\ell-1}) \Delta t
\end{equation}
i.e. $\hat{U}$ is equal to its first order Taylor expansion. Moreover, this inequality becomes exact in the $N \to \infty$ limit. Using this,
\begin{equation}
\matrixel{- i \mathbf{p}^{\ell}}{\hat{U}(t_{\ell}, t_{\ell - 1})}{\mathbf{z}^{\ell-1}} \approx e^{i \mathbf{p}^{\ell} \cdot \mathbf{z}^{\ell-1}} + \Delta t \matrixel{- i \mathbf{p}^{\ell}}{\hat{H}(t_{\ell-1})}{\mathbf{z}^{\ell-1}}  \ .
\end{equation}
By the definition of the Hamiltonian kernel,
\begin{equation}
\begin{split}
\matrixel{- i \mathbf{p}^{\ell}}{\hat{U}(t_{\ell}, t_{\ell - 1})}{\mathbf{z}^{\ell-1}} &\approx e^{i \mathbf{p}^{\ell} \cdot \mathbf{z}^{\ell-1}} \left[ 1 + \mathcal{H}(i \mathbf{p}^{\ell}, \mathbf{z}^{\ell-1}, t_{\ell-1}) \Delta t \right] \\
&\approx e^{i \mathbf{p}^{\ell} \cdot \mathbf{z}^{\ell-1} + \Delta t \mathcal{H}(i \mathbf{p}^{\ell}, \mathbf{z}^{\ell-1}, t_{\ell-1})} 
\end{split}
\end{equation}
where we have again used the fact that $\Delta t$ is small. Putting all of these matrix elements together, our final coherent state path integral expression for $U(i \mathbf{p}^f, \mathbf{z}^0)$ reads
\begin{equation} 
\begin{split}
U = \ & \lim_{N \to \infty} \int \prod_{\ell = 1}^{N-1} \frac{d\mathbf{z}^{\ell} d\mathbf{p}^{\ell} }{(2\pi)^n} \exp\left\{ \sum_{\ell = 1}^{N-1}  - i \mathbf{p}^{\ell} \cdot ( \mathbf{z}^{\ell} - \mathbf{z}^{\ell - 1}) + \Delta t \ \mathcal{H}(i \mathbf{p}^{\ell}, \mathbf{z}^{\ell - 1}, t_{\ell - 1}) \right. \\
 & \left. + \Delta t \ \mathcal{H}(i \mathbf{p}^f, \mathbf{z}^{N - 1}, t_{\ell-1}) + i \mathbf{p}^f \cdot \mathbf{z}^{N-1}  \right\}
\end{split}
\end{equation}
where the $N \to \infty$ limit must be taken so that the approximation we made in Eq. \ref{eq:Utaylor} becomes exact. \qed
\end{proof}

\subsection{Grassberger-Scheunert creation operators}
\label{sec:gscreation}

%ARBITRARY OPERATOR RESULTS IN THIS SECTION SHOULD BE MORE GENERAL!!! \\

As we noted in Sec. \ref{sec:innerproducts}, the Hermitian conjugate of the annihilation operator $\hat{a}_j$ with respect to the Grassberger-Scheunert product is $\hat{\pi}_j - 1$ for all $j = 1, ..., n$. Motivated by this, we define the Grassberger-Scheunert creation operators.

\begin{definition} The \textit{Grassberger-Scheunert} creation operators are defined to be
\begin{equation}
\hat{a}_j^+ := \hat{\pi}_j - 1
\end{equation}
for all $j = 1, ..., n$. 
\end{definition}

In the rest of the article, we will take `creation operator' without qualification to mean one of these operators.

All Hamiltonians we consider may be expressed in terms of creation operators and annihilation operators. For example, the Hamiltonian operator corresponding to the chemical birth-death process (c.f. Sec. \ref{sec:intro_bd}) can be shown to read
\begin{equation}
\hat{H} = k(t) \hat{a}^+ - \gamma(t) \hat{a}^+ \hat{a}
\end{equation}
which is a specific case of a result we derive later (see Sec. \ref{sec:monocalc}). Note that this expression is `normal ordered'---all creation operators are to the left of all annihilation operators. For all (possibly time-dependent) operators $\hat{\mathcal{A}}(t)$ in this form, i.e.
\begin{equation}
\hat{\mathcal{A}}(t) := \sum_{\nu_1, ..., \nu_n, \rho_1, ..., \rho_n} d^{\nu_1, ..., \nu_n}_{\rho_1, ..., \rho_n}(t) \ (\hat{a}_1^+)^{\nu_1} \cdots (\hat{a}_n^+)^{\nu_n} (\hat{a}_1)^{\rho_1} \cdots (\hat{a}_n)^{\rho_n}   \ ,
\end{equation}
coherent state matrix elements are easily evaluated by exploiting that $(\hat{a}_j)^{\dag} = \hat{a}^+_j$ and that the coherent states are eigenstates of the annihilation operators. The particular result is the following.

\begin{proposition}[Coherent state matrix elements of normal ordered operators] \label{prop:csmel_normal} 
Let $\ket{\mathbf{z}}$ and $\ket{\mathbf{p}}$ be coherent states, and $\hat{\mathcal{A}}(t)$ be an arbitrary (possibly time-dependent) operator that is normal ordered (i.e. all creation operators are to the left of all annihilation operators). The coherent state matrix element $\matrixel{\mathbf{p}}{\hat{\mathcal{A}}(t)}{\mathbf{z}}$ can be evaluated to be
\begin{equation} \label{eq:howtomel}
\begin{split}
\matrixel{\mathbf{p}}{\hat{\mathcal{A}}(t)}{\mathbf{z}} &=  e^{\mathbf{p}^* \cdot \mathbf{z}} \sum_{\nu_1, ..., \nu_n, \rho_1, ..., \rho_n} d^{\nu_1, ..., \nu_n}_{\rho_1, ..., \rho_n}(t) \ (p_1^*)^{\nu_1} \cdots (p_n^*)^{\nu_n} (z_1)^{\rho_1} \cdots (z_n)^{\rho_n}    \ .
\end{split}
\end{equation}
\end{proposition}
\begin{proof}
Use the linearity of the inner product to take the sum out, then use the facts that $(\hat{a}_j)^{\dag} = \hat{a}^+_j$, $\hat{a}_k \ket{\mathbf{z}} = z_k \ket{\mathbf{z}}$, and $\bra{\mathbf{p}} \hat{a}^+_k = p_k^* \bra{\mathbf{p}}$. Finally, note that $\braket{\mathbf{p}}{\mathbf{z}} = e^{\mathbf{p}^* \cdot \mathbf{z}} $. \qed
\end{proof}
We will use this result in the calculation sections to derive many Hamiltonian kernels.

\subsection{Probability distribution and moments}
\label{sec:howtoprob}

We need some way to extract information (like the transition probability $P(\mathbf{x}, t)$ or factorial moments) from the generating function $\ket{\psi(t)}$. It turns out that we can achieve this using the exclusive product \cite{peliti1985} and Grassberger-Scheunert product \cite{grassberger1980}. 

\begin{proposition}[Extracting transition probabilities and moments from the generating function] \label{prop:gfto}
Transition probabilities can be obtained from the generating function using the exclusive product, and factorial moments can be obtained from the generating function using the Grassberger-Scheunert product and the annihilation operators. In particular,
\begin{equation}\label{eq:gftop}
P(\mathbf{x}, t) = \frac{\braket{\mathbf{x}}{\psi(t)}_{ex}}{\mathbf{x}!} 
\end{equation} 
and
\begin{equation} \label{eq:gftomoments}
\begin{split}
\expval{x_j(t)} &= \matrixel{\mathbf{0}}{\hat{a}_j}{\psi(t)} \\
\expval{x_j(t) x_k(t)} &= \matrixel{\mathbf{0}}{\hat{a}_j \hat{a}_k}{\psi(t)} \\
\expval{x_j(t) [x_j(t) - 1]} &= \matrixel{\mathbf{0}}{\hat{a}^2_j}{\psi(t)} \\
\expval{x_j(t) [x_j(t) - 1] [x_j(t) - 2]} &= \matrixel{\mathbf{0}}{\hat{a}^3_j}{\psi(t)} \ .
\end{split}
\end{equation}
\end{proposition}
\begin{proof}
By the definition of the exclusive product,
\begin{equation}
\frac{\braket{\mathbf{x}}{\psi(t)}_{ex}}{\mathbf{x}!} = \sum_{\mathbf{y}} P(\mathbf{y}, t) \frac{\braket{\mathbf{x}}{\mathbf{y}}_{ex}}{\mathbf{x}!} =  \sum_{\mathbf{y}} P(\mathbf{y}, t) \delta(\mathbf{x} - \mathbf{y}) = P(\mathbf{x}, t) \ .
\end{equation}
By the explicit definition of the Grassberger-Scheunert product of two basis vectors (c.f. Eq. \ref{eq:in_basis}), note that $\braket{\mathbf{0}}{\mathbf{x}}$ for all $\mathbf{x} \in \mathbb{N}^n$. Then
\begin{equation}
\begin{split}
\matrixel{\mathbf{0}}{\hat{a}_j}{\psi(t)} &= \bra{\mathbf{0}} \sum_{\mathbf{x}} P(\mathbf{x}, t) \hat{a}_j \ket{\mathbf{x}} \\
&= \bra{\mathbf{0}} \sum_{\mathbf{x}} P(\mathbf{x}, t)  x_j \ket{\mathbf{x} - \boldsymbol{\epsilon}_j} \\
&=  \sum_{\mathbf{x}} P(\mathbf{x}, t)  x_j \braket{\mathbf{0}}{\mathbf{x} - \boldsymbol{\epsilon}_j}  \\
&= \sum_{\mathbf{x}} x_j P(\mathbf{x}, t) \\
&= \langle x_j(t) \rangle \ .
\end{split}
\end{equation}
The other expectation value formulas can be demonstrated in a similar fashion. \qed
\end{proof}
With this done, we have all of the machinery necessary to solve the problems identified in Sec. \ref{sec:pstatement}.

\section{Monomolecular calculations}
\label{sec:monocalc}

In this section, we present the calculations relevant to proving the formulas from Theorems \ref{thm:bd} and \ref{thm:mono} using the Doi-Peliti approach. First, we derive the Hamiltonian operator and use it to compute the Hamiltonian kernel. Then we evaluate the path integral expression for the propagator $U(i \mathbf{p}^f, \mathbf{z}^0)$. Finally, we use the explicit form of the propagator to derive the transition probability and several moments. We do not explicitly show how to compute the generating function directly from the propagator, because it is very similar to the other calculations. 

\subsection{The Hamiltonian operator and Hamiltonian kernel}

We would like an equation equivalent to the CME (Eq. \ref{eq:CME}) that is satisfied by the generating function $\ket{\psi(t)}$ (Eq. \ref{eq:gf}). Doing so involves a straightforward calculation, which we spell out here for the sake of illustration.

\begin{lemma}  \label{lem:monoHderiv}
The Hamiltonian operator corresponding to the monomolecular CME (Eq. \ref{eq:CME}) is
\begin{equation} \label{eq:Hdef}
\begin{split}
\hat{H} =& \ \sum_{k=1}^{n} c_{0k}(t) \left[ \hat{\pi}_k - 1 \right] - \sum_{k=1}^{n} c_{k0}(t)\left[ \hat{\pi}_k - 1 \right] \hat{a}_k \\
&+ \sum_{j=1}^{n} \sum_{k=1}^{n}  c_{jk}(t) \left[ \hat{\pi}_k  - \hat{\pi}_j  \right] \hat{a}_j  \ .
\end{split}
\end{equation}
\end{lemma}

\begin{proof}
First, take the time derivative of $\ket{\psi(t)}$:
\begin{equation}
\begin{split}
\frac{\partial \ket{\psi}}{\partial t} =& \sum_{\mathbf{x}}  \frac{\partial P(\mathbf{x}, t)}{\partial t} \ket{\mathbf{x}}  \\
=& \sum_{\mathbf{x}}  \left\{ \sum_{k=1}^{n} c_{0k}(t) \left[ P(\mathbf{x} - \boldsymbol{\epsilon}_k, t) - P(\mathbf{x},t)  \right]  \right.  \\
&+ \sum_{k=1}^{n} c_{k0}(t)\left[ (x_k + 1) P(\mathbf{x} + \boldsymbol{\epsilon}_k, t) - x_k P(\mathbf{x},t)  \right] \\
&+ \left. \sum_{j=1}^{n} \sum_{k=1}^{n}  c_{jk}(t) \left[ (x_j + 1) P(\mathbf{x} + \boldsymbol{\epsilon}_j - \boldsymbol{\epsilon}_k, t) - x_j P(\mathbf{x},t)  \right] \right\} \ket{\mathbf{x}} 
\end{split}
\end{equation}
where we have used Eq. \ref{eq:CME}. Reindex the sums over $\mathbf{x}$ so that this expression reads
\begin{equation}
\begin{split}
\frac{\partial \ket{\psi}}{\partial t} =& \sum_{\mathbf{x}}  \left\{ \sum_{k=1}^{n} c_{0k}(t) \left[ \ \ket{\mathbf{x} + \boldsymbol{\epsilon}_k}  - \ket{\mathbf{x}}  \ \right]  \right.  \\
&+ \sum_{k=1}^{n} c_{k0}(t)\left[ \ x_k  \ket{\mathbf{x} - \boldsymbol{\epsilon}_k}  - x_k \ket{\mathbf{x}}  \ \right] \\
&+ \left. \sum_{j=1}^{n} \sum_{k=1}^{n}  c_{jk}(t) \left[ \  x_j  \ket{\mathbf{x} - \boldsymbol{\epsilon}_j + \boldsymbol{\epsilon}_k}  - x_j \ket{\mathbf{x}} \ \right] \right\} P(\mathbf{x},t) \ .
\end{split}
\end{equation}
Using the creation and annihilation operators we defined earlier, the right-hand side can be written as
\begin{equation}
\begin{split}
& \sum_{\mathbf{x}}  \left\{ \sum_{k=1}^{n} c_{0k}(t) \left[ \ \hat{\pi}_k - 1  \ \right]   + \sum_{k=1}^{n} c_{k0}(t)\left[ \ \hat{a}_k - \hat{\pi}_k \hat{a}_k  \ \right] \right. \\
& \hspace{0.5in} + \left.  \sum_{j=1}^{n} \sum_{k=1}^{n}  c_{jk}(t) \left[ \  \hat{a}_j \hat{\pi}_k  - \hat{\pi}_j \hat{a}_j \ \right] \right\} P(\mathbf{x},t) \ket{\mathbf{x}} \\
=& \left\{ \sum_{k=1}^{n} c_{0k}(t) \left[ \ \hat{\pi}_k - 1  \ \right]   + \sum_{k=1}^{n} c_{k0}(t)\left[ \ \hat{a}_k - \hat{\pi}_k \hat{a}_k  \ \right] \right. \\
& \hspace{0.5in} + \left.  \sum_{j=1}^{n} \sum_{k=1}^{n}  c_{jk}(t) \left[ \  \hat{a}_j \hat{\pi}_k  - \hat{\pi}_j \hat{a}_j \ \right] \right\} \ket{\psi(t)} \ .
\end{split}
\end{equation}
Comparing this with the definition of the Hamiltonian operator (c.f. Eq. \ref{eq:EOM}), we have our result. \qed
\end{proof}
The Hamiltonian can be written more compactly in terms of the Grassberger-Scheunert creation operators:

\begin{corollary}
In terms of the Grassberger-Scheunert creation operator, the Hamiltonian is
\begin{equation} \label{eq:HdefGS}
\hat{H} = \sum_{k=1}^{n} c_{0k}(t)  \hat{a}_k^+  - \sum_{k=1}^{n} c_{k0}(t) \hat{a}_k^+  \hat{a}_k + \sum_{j=1}^{n} \sum_{k=1}^{n}  c_{jk}(t) \left[ \hat{a}^+_k  - \hat{a}^+_j \right] \hat{a}_j \ .
\end{equation}
\end{corollary}
\begin{proof}
Start with the result above and make the identification $\hat{a}_j^+ = \hat{\pi}_j - 1$. \qed
\end{proof}
Note that this expression is `normal ordered'---all creation operators are to the left of all annihilation operators. This allows us to use Proposition \ref{prop:csmel_normal} to compute the Hamiltonian kernel. 

\begin{corollary}
The Hamiltonian kernel for the monomolecular CME is
\begin{equation} \label{eq:Hkernel}
\begin{split}
- i \mathcal{H}(i \mathbf{p}^{\ell}, \mathbf{z}^{\ell-1}, t_{\ell-1}) =&  \ \sum_{k=1}^{n} c_{0k}(t_{\ell-1}) p^{\ell}_k  - \sum_{k=1}^{n} c_{k0}(t_{\ell-1}) p^{\ell}_k  z^{\ell-1}_k \\
&+ \sum_{j=1}^{n} \sum_{k=1}^{n}  c_{jk}(t_{\ell-1}) \left[  p^{\ell}_k  - p^{\ell}_j \right] z^{\ell-1}_j \ .
\end{split}
\end{equation}
\end{corollary}
\begin{proof}
Make the identifications $\hat{a}_j^+ \to i p^{\ell}_j$ and $\hat{a}_j \to z^{\ell-1}_j$ in the Hamiltonian above. \qed
\end{proof}
%%%%%%%%%%%%%%%%%%%%%%%%%%%%%

%%%%%%%%%%%%%%%%%

\subsection{Evaluating the propagator path integral}

In this section, we will evaluate the path integral expression for the propagator $U(i \mathbf{p}^f, \mathbf{z}^0)$ (Eq. \ref{eq:Upathint}) given our specific dynamics, which are captured by the Hamiltonian kernel $\mathcal{H}$ (Eq. \ref{eq:Hkernel}).

%%%%%%%%%%

\begin{lemma}[Monomolecular propagator]  \label{lem:mono_prop}
The propagator for the monomolecular system is
\begin{equation} \label{eq:Usln2}
U(i \mathbf{p}^f, \mathbf{z}^0) = e^{i \mathbf{p}^f \cdot  \mathbf{z}(t)} 
\end{equation}
where
\begin{equation} \label{eq:zdecomp}
\mathbf{z}(t) := \sum_{k=1}^n  z^0_k \mathbf{w}^{(k)}(t) + \boldsymbol{\lambda}(t)
\end{equation}
with $\mathbf{w}^{(k)}(t)$ and $\boldsymbol{\lambda}(t)$ as defined in Theorem \ref{thm:mono}. 
\end{lemma}

\begin{proof} 
Begin with the path integral expression for $U$ (Eq. \ref{eq:Upathint}). Let us first integrate over the $p^{\ell}_k$ (where $\ell \in \left\{ 1, ..., N-1\right\}$ and $k \in \left\{ 1, ..., n \right\}$). For fixed $\ell$ and $k$, these integrals look like
\begin{equation}
\int_{-\infty}^{\infty} \frac{dp^{\ell}_k}{2\pi} \exp\left\{ - i p^{\ell}_k \left[ (z_k^{\ell} - z_k^{\ell - 1}) - \Delta t \left( \ c_{0k}^{\ell - 1}  - c_{k0}^{\ell - 1} z_k^{\ell - 1} + \sum_{j=1}^{n} c_{jk}^{\ell - 1}  z_j^{\ell - 1}  - c_{kj}^{\ell - 1} z_k^{\ell - 1}  \ \right)  \ \right] \right\}
\end{equation}
where $c_{jk}^{\ell-1}$ is shorthand for $c_{jk}(t_{\ell-1})$. Using the usual integral representation of the Dirac delta function, these integrals are easily done to obtain $n \cdot (N-1)$ delta function constraints:
\begin{equation}
\delta\left[ (z_k^{\ell} - z_k^{\ell - 1}) - \Delta t \left( \ c_{0k}^{\ell - 1}  - c_{k0}^{\ell - 1} z_k^{\ell - 1} + \sum_{j=1}^{n} c_{jk}^{\ell - 1}  z_j^{\ell - 1}  - c_{kj}^{\ell - 1} z_k^{\ell - 1}  \ \right)  \ \right] \ .
\end{equation}
Fortunately, that is \textit{exactly} how many integrals we have left to do. Notice that the constraints force
\begin{equation} \label{eq:timestep}
z_k^{\ell} = z_k^{\ell - 1} + \Delta t \left( \ c_{0k}^{\ell - 1}  - c_{k0}^{\ell - 1} z_k^{\ell - 1} + \sum_{j=1}^{n} c_{jk}^{\ell - 1}  z_j^{\ell - 1}  - c_{kj}^{\ell - 1} z_k^{\ell - 1}  \ \right) 
\end{equation}
which exactly corresponds to taking an Euler time step given the deterministic dynamics described by the reaction rate equations, Eq. \ref{eq:RRE}. What remains of our calculation is to evaluate
\begin{equation} 
\begin{split}
U = \ & \lim_{N \to \infty} \exp\left\{ \Delta t \ \mathcal{H}(i \mathbf{p}^f, \mathbf{z}^{N - 1}, t_{\ell-1}) + i \mathbf{p}^f \cdot \mathbf{z}^{N-1}  \right\} 
\end{split}
\end{equation}
given Eq. \ref{eq:timestep}, the constraint on $\mathbf{z}^{N-1}$ relating it (via $(N-1)$ Euler time steps) to $\mathbf{z}^0$. We have
\begin{equation}
\begin{split}
& i \mathbf{p}^f \cdot \mathbf{z}^{N-1} + \Delta t \ \mathcal{H}(i \mathbf{p}^f, \mathbf{z}^{N - 1}, t_{\ell-1}) \\
=& i \sum_{k = 1}^n p^f_k \left\{ z_k^{N-1} +  \Delta t \left[ \ c_{0k}^{N - 1}  - c_{k0}^{N - 1} z_k^{N - 1} + \sum_{j=1}^{n} c_{jk}^{N - 1}  z_j^{N - 1}  - c_{kj}^{N - 1} z_k^{N - 1}  \ \right]   \right\} \\
=& i \sum_{k = 1}^n p^f_k z_k^N
\end{split}
\end{equation}
where we define $z_k^N$ as the result of taking $N$ time steps of length $\Delta t$ according to Eq. \ref{eq:timestep} given the initial condition $z_k^0$. In the $N \to \infty$ limit, $z_k^N \to z_k(t)$, where $z_k(t)$ is defined as the $k$th component of the solution to Eq. \ref{eq:RRE}. As described in Sec. \ref{sec:pstatement}, $\mathbf{z}(t)$ can be decomposed in terms of $\boldsymbol{\lambda}(t)$ and the $\mathbf{w}^{(k)}(t)$. 

%where $\mathbf{w}^{(k)}$ and $\boldsymbol{\lambda}$ are as defined in Eq. \ref{eq:wlambda}. Finally,
%\begin{equation} \label{eq:Usln}
%U(i \mathbf{p}^f, \mathbf{z}^0) = e^{i \mathbf{p}^f \cdot  \mathbf{z}(t)} \ .
%\end{equation}
\qed
\end{proof}
While it may seem that this path integral calculation was completely trivial, that is mostly because we put in the legwork to define and characterize the Grassberger-Scheunert product beforehand. Had we used the exclusive product to construct our path integral, we would either have to perform a hard to justify Doi shift, or deal with extra terms after enforcing the delta function constraints.

%%%%%%%%%%%%%%%%%%%%%%%%%%%%%

%\section{From the propagator to the transition probability}
%\label{sec:proptoprob}

Now that we have computed the propagator $U(i \mathbf{p}^f, \mathbf{z}^0)$, we can relate it to the generating function $\ket{\psi(t)}$ using Corollary \ref{cor:Utogf}. Then, using Proposition \ref{prop:gfto}, the generating function can be used to compute $P(\mathbf{x}, t)$ and various moments. Because the transition probability and moment calculations are somewhat involved, we first present them for the one species system (i.e. the chemical birth-death process). 

\subsection{One species transition probability derivation}
\label{sec:onespeciesptrans}

\begin{lemma}[One species monomolecular transition probability]
For the single species monomolecular system (i.e. the chemical birth-death process), the transition probability $P(x, t; \xi, t_0)$ is
\begin{equation}
\begin{split}
P &=  \sum_{k = 0}^{\min(x, \xi)} \left[ \frac{\lambda(t)^{x - k} e^{-\lambda(t)}}{(x - k)!} \right] \left[ \binom{\xi}{k} w(t)^k  [1 - w(t)]^{\xi - k} \right]  \\
&= \mathcal{P}(x, \lambda(t)) \star \mathcal{M}(x, \xi, w(t)) 
\end{split}
\end{equation}
where $w(t)$ and $\lambda(t)$ are as defined in Theorem \ref{thm:bd}.
\end{lemma}

\begin{proof}
Recall that, since $P(x, t_0) = \delta(x - \xi)$ for some $\xi \geq 0$,
\begin{equation}
\ket{\psi(t_0)} = \ket{\xi} \ .
\end{equation}
Using Eq. \ref{eq:gftop} and Eq. \ref{eq:psiU}, we have
\begin{equation}
\begin{split}
P(x, t; \xi, t_0) &= \frac{1}{x!} \int \frac{dz^f dp^f}{2 \pi} \frac{dz^0 dp^0}{2 \pi} \braket{x}{z^f}_{ex} \ U(ip^f, z^0) \braket{- i p^0}{\psi(t_0)} \ e^{- i p^0 z^0  - i p^f z^f}  \\
&= \frac{1}{x!} \int \frac{dz^f dp^f}{2 \pi} \frac{dz^0 dp^0}{2 \pi} \left(z^f \right)^x e^{- z^f} \ e^{i p^f z(t)} (1 + i p^0)^{\xi} \ e^{- i p^0 z^0  - i p^f z^f} \ .
\end{split}
\end{equation}
The integral over $p^f$ is easily done:
\begin{equation}
\int_{-\infty}^{\infty} \frac{dp^f}{2 \pi} e^{i p^f [ z(t) - z^f ]} = \delta(z(t) - z^f) \ .
\end{equation}
Enforcing the delta function constraint removes the integral over $z^f$. Since $z(t) = z^0 w(t) + \lambda(t)$, 
\begin{equation}
\begin{split}
P &= \frac{1}{x!} \int \frac{dz^0 dp^0}{2 \pi} \left[ z^0 w(t) + \lambda(t) \right]^x e^{- [z^0 w(t) + \lambda(t)]} \ (1 + i p^0)^{\xi} \ e^{- i p^0 z^0} \\
&= \frac{e^{-\lambda(t)}}{x!} \int \frac{dz^0 dp^0}{2 \pi} \left[ z^0 w(t) + \lambda(t) \right]^x e^{z^0 [1 - w(t) ]} \ (1 + i p^0)^{\xi} \ e^{- z^0 [ 1 + i p^0]} \ .
\end{split}
\end{equation}
This can be rewritten as
\begin{equation}
\begin{split}
P &= \frac{e^{-\lambda(t)}}{x!} \int \frac{dz^0 dp^0}{2 \pi} \left[ z^0 w(t) + \lambda(t) \right]^x e^{z^0 [1 - w(t) ]} \ \left( - \frac{d}{dz^0} \right)^{\xi} \ e^{- z^0 [ 1 + i p^0]} \\
&= \frac{e^{-\lambda(t)}}{x!} \int \frac{dz^0 dp^0}{2 \pi} \left(\frac{d}{dz^0} \right)^{\xi} \left\{ \left[ z^0 w(t) + \lambda(t) \right]^x e^{z^0 [1 - w(t) ]} \right\}   \ e^{- z^0 [ 1 + i p^0]} 
\end{split}
\end{equation}
where we integrated by parts in the second step. The $p^0$ integral can now be done:
\begin{equation}
\int_{-\infty}^{\infty} \frac{dp^0}{2 \pi} e^{- i p^0 z^0} = \delta(z^0) \ .
\end{equation}
We now have
\begin{equation}
\begin{split}
P &= \frac{e^{-\lambda(t)}}{x!} \int_0^{\infty} dz^0 \left(\frac{d}{dz^0} \right)^{\xi} \left\{ \left[ z^0 w(t) + \lambda(t) \right]^x e^{z^0 [1 - w(t) ]} \right\}   \ e^{- z^0} \delta(z^0) \ .
\end{split}
\end{equation}
If we can evaluate the derivative, then we can easily evaluate the integral using the delta function. Using the binomial theorem,
\begin{equation}
\left[ z^0 w(t) + \lambda(t) \right]^x = \sum_{k = 0}^x \binom{x}{k} w(t)^k \lambda(t)^{x - k} \left( z^0 \right)^k \ .
\end{equation}
Since
\begin{equation}
\left( z^0 \right)^k e^{z^0 [1 - w(t) ]} = \sum_{j = 0}^{\infty} \frac{\left( z^0 \right)^{j + k} [1 - w(t)]^j}{j!} \ ,
\end{equation}
the derivative of a specific term is
\begin{equation}
\left(\frac{d}{dz^0} \right)^{\xi} \left\{ \left( z^0 \right)^k e^{z^0 [1 - w(t) ]} \right\}  = \sum_{j = 0}^{\infty} \frac{(j+k)(j+k-1) \cdots (j+k-\xi+1)\left( z^0 \right)^{j + k - \xi} [1 - w(t)]^j}{j!} \ .
\end{equation}
When enforcing the delta function constraint that $z^0 = 0$, all terms will disappear from this series except for the constant term. The constant term is the term with $j + k = \xi$, which reads
\begin{equation} \label{eq:constterm}
\frac{\xi!}{(\xi - k)!} [1 - w(t)]^{\xi - k} \theta(\xi - k)
\end{equation}
where the step function $\theta$, defined as
\begin{equation}
\theta(\xi - k) := \begin{cases} 
      1 & k \leq \xi \\
      0 & k > \xi
   \end{cases}
\end{equation}
must be there since the result will be zero if $k > \xi$. Hence,
\begin{equation}
\begin{split}
P &= \frac{e^{-\lambda(t)}}{x!} \sum_{k = 0}^x \binom{x}{k} w(t)^k \lambda(t)^{x - k} \frac{\xi!}{(\xi - k)!} [1 - w(t)]^{\xi - k} \theta(\xi - k) \\
&= e^{-\lambda(t)} \sum_{k = 0}^{\min(x, \xi)} \binom{\xi}{k} w(t)^k \lambda(t)^{x - k} \frac{1}{(x - k)!} [1 - w(t)]^{\xi - k}  \\
&=  \sum_{k = 0}^{\min(x, \xi)} \left[ \frac{\lambda(t)^{x - k} e^{-\lambda(t)}}{(x - k)!} \right] \left[ \binom{\xi}{k} w(t)^k  [1 - w(t)]^{\xi - k} \right]  \\
&= \mathcal{P}(x, \lambda(t)) \star \mathcal{M}(x, \xi, w(t))
\end{split}
\end{equation}
as desired. \qed
\end{proof}

\subsection{General transition probability derivation}
\label{sec:genptrans}

\begin{lemma}[Monomolecular transition probability]
For the general monomolecular system, the transition probability $P(\mathbf{x}, t; \boldsymbol{\xi}, t_0)$ is
\begin{equation}
P =  \mathcal{P}(\mathbf{x}, \boldsymbol{\lambda}(t)) \star \mathcal{M}(\mathbf{x}, \xi_1, \mathbf{w}^{(1)}(t)) \star \cdots \star \mathcal{M}(\mathbf{x}, \xi_n, \mathbf{w}^{(n)}(t)) 
\end{equation}
where $\boldsymbol{\lambda}(t)$ and the $\mathbf{w}^{(j)}(t)$ are as defined in Theorem \ref{thm:mono}.  
\end{lemma}
\begin{proof}
The general case proceeds analogously to the one species case. The main difference is that we must do the appropriate multivariable generalization of each of the steps in the previous subsection (e.g. use the multinomial theorem instead of the binomial theorem). Since $P(\mathbf{x}, t_0) = \delta(\mathbf{x} - \boldsymbol{\xi})$,
\begin{equation}
\ket{\psi(t_0)} = \ket{\boldsymbol{\xi}} \ .
\end{equation}
Using Eq. \ref{eq:gftop} and Eq. \ref{eq:psiU},
\begin{equation}
\begin{split}
P(\mathbf{x}, t; \boldsymbol{\xi}, t_0) &= \frac{1}{\mathbf{x}!} \int \frac{d\mathbf{z}^f d\mathbf{p}^f}{(2 \pi)^n} \frac{d\mathbf{z}^0 d\mathbf{p}^0}{(2 \pi)^n} \braket{\mathbf{x}}{\mathbf{z}^f}_{ex} \ U(i \mathbf{p}^f, \mathbf{z}^0) \braket{- i \mathbf{p}^0}{\psi(t_0)} \ e^{- i \mathbf{p}^0 \cdot \mathbf{z}^0  - i \mathbf{p}^f \cdot \mathbf{z}^f} \\
&= \frac{1}{\mathbf{x}!} \int \frac{d\mathbf{z}^f d\mathbf{p}^f}{(2 \pi)^n} \frac{d\mathbf{z}^0 d\mathbf{p}^0}{(2 \pi)^n} \left(\mathbf{z}^f \right)^\mathbf{x} e^{- \mathbf{z}^f \cdot \mathbf{1}} \ e^{i \mathbf{p}^f \cdot \mathbf{z}(t)} (\mathbf{1} + i \mathbf{p}^0)^{\boldsymbol{\xi}} \ e^{- i \mathbf{p}^0 \cdot \mathbf{z}^0  - i \mathbf{p}^f \cdot \mathbf{z}^f}  \ .
\end{split}
\end{equation}
The integrals over $p^f_1, ..., p^f_n$ yield delta functions:
\begin{equation}
\int \frac{d\mathbf{p}^f}{(2 \pi)^n} e^{i \mathbf{p}^f \cdot [ \mathbf{z}(t) - \mathbf{z}^f ]} = \delta(z_1(t) - z^f_1) \cdots \delta(z_n(t) - z^f_n) = \delta(\mathbf{z}(t) - \mathbf{z}^f) \ .
\end{equation}
Enforcing the delta function constraints removes the integrals over $z_1^f, ..., z^f_n$. Using Eq. \ref{eq:zdecomp},
\begin{equation}
\begin{split}
P &= \frac{1}{\mathbf{x}!} \int \frac{d\mathbf{z}^0 d\mathbf{p}^0}{(2 \pi)^n} \left[ \sum_{k=1}^n  z^0_k \mathbf{w}^{(k)} + \boldsymbol{\lambda} \right]^\mathbf{x} e^{- [\sum_{k=1}^n  z^0_k \mathbf{w}^{(k)} + \boldsymbol{\lambda}] \cdot \mathbf{1}} \ (\mathbf{1} + i \mathbf{p}^0)^{\boldsymbol{\xi}} \ e^{- i \mathbf{p}^0 \cdot \mathbf{z}^0}  \\
&= \frac{e^{-|\boldsymbol{\lambda}(t)|}}{\mathbf{x}!} \int \frac{d\mathbf{z}^0 d\mathbf{p}^0}{(2 \pi)^n} \left[ \sum_{k=1}^n  z^0_k \mathbf{w}^{(k)} + \boldsymbol{\lambda} \right]^\mathbf{x} e^{\sum_{k=1}^n z^0_k \left( 1 - |\mathbf{w}^{(k)}| \right)} \  (\mathbf{1} + i \mathbf{p}^0)^{\boldsymbol{\xi}} \ e^{- \mathbf{z}^0 \cdot [ \mathbf{1} + i \mathbf{p}^0]} \ .
\end{split}
\end{equation}
Reusing the notation we used earlier to denote many derivatives with respect to each variable (Eq. \ref{eq:ddzshorthand}), we can rewrite this result as
\begin{equation}
\begin{split}
P &= \frac{e^{-|\boldsymbol{\lambda}(t)|}}{\mathbf{x}!} \int \frac{d\mathbf{z}^0 d\mathbf{p}^0}{(2 \pi)^n} \left[ \sum_{k=1}^n  z^0_k \mathbf{w}^{(k)} + \boldsymbol{\lambda} \right]^\mathbf{x} e^{\sum_{k=1}^n z^0_k \left( 1 - |\mathbf{w}^{(k)}| \right)}  \ \left( - \frac{d}{d\mathbf{z}^0} \right)^{\boldsymbol{\xi}} \ e^{- \mathbf{z}^0 \cdot [ \mathbf{1} + i \mathbf{p}^0]}\\
&= \frac{e^{-|\boldsymbol{\lambda}(t)|}}{\mathbf{x}!} \int \frac{d\mathbf{z}^0 d\mathbf{p}^0}{(2 \pi)^n} \left( \frac{d}{d\mathbf{z}^0} \right)^{\boldsymbol{\xi}} \left\{ \left[ \sum_{k=1}^n  z^0_k \mathbf{w}^{(k)} + \boldsymbol{\lambda} \right]^\mathbf{x} e^{\sum_{k=1}^n z^0_k \left( 1 - |\mathbf{w}^{(k)}| \right)} \right\}   \ e^{- \mathbf{z}^0 \cdot [ \mathbf{1} + i \mathbf{p}^0]}
\end{split}
\end{equation}
where we integrated by parts many times in the second step. The $p^0_1, ..., p^0_n$ integrals can now be done:
\begin{equation}
\int \frac{d\mathbf{p}^0}{(2 \pi)^n} e^{- i \mathbf{p}^0 \cdot \mathbf{z}^0} = \delta(z^0_1) \cdots \delta(z^0_n) = \delta(\mathbf{z}^0) \ .
\end{equation}
We now have
\begin{equation}
\begin{split}
P &= \frac{e^{-|\boldsymbol{\lambda}(t)|}}{\mathbf{x}!} \int d\mathbf{z}^0  \left( \frac{d}{d\mathbf{z}^0} \right)^{\boldsymbol{\xi}} \left\{ \left[ \sum_{k=1}^n  z^0_k \mathbf{w}^{(k)} + \boldsymbol{\lambda} \right]^\mathbf{x} e^{\sum_{k=1}^n z^0_k \left( 1 - |\mathbf{w}^{(k)}| \right)} \right\}    \ e^{- \mathbf{z}^0 \cdot \mathbf{1}} \delta(\mathbf{z}^0) \ .
\end{split}
\end{equation}
If we can evaluate the derivative, then we can easily evaluate the integral using the delta function. Recall that
\begin{equation}
\begin{split}
\left[ \sum_{k=1}^n  z^0_k \mathbf{w}^{(k)} + \boldsymbol{\lambda} \right]^\mathbf{x} &= \left[ \sum_{k=1}^n  z^0_k w^{(k)}_1 + \lambda_1 \right]^{x_1}  \cdots \left[ \sum_{k=1}^n  z^0_k w^{(k)}_n + \lambda_n \right]^{x_n}  \ .
\end{split}
\end{equation}
Using the multinomial theorem,
\begin{equation}
\begin{split}
\left[ \sum_{k=1}^n  z^0_k w^{(k)}_j + \lambda_j \right]^{x_j} &= \sum_{v^j_1 + \cdots v^j_{n+1} = x_j} \binom{x_j}{v^j_1 \cdots v^j_{n+1}} \left[ z^0_1 w^{(1)}_j \right]^{v^j_1} \cdots \left[ z^0_n w^{(n)}_j  \right]^{v^j_n} \left[ \lambda_j \right]^{v^j_{n+1}} 
\end{split}
\end{equation}
for each $j = 1, ..., n$. Write $|\mathbf{v}_{\ell}| := v^1_{\ell} + \cdots + v^n_{\ell}$. Putting these multinomial expansions together, our integral now involves computing $n$ expressions of the form
\begin{equation}
\left. \left( \frac{d}{dz^0_{\ell}} \right)^{\xi_{\ell}} \left\{ \left[ z^0_{\ell} \right]^{|\mathbf{v}_{\ell}|} e^{z^0_{\ell} \left( 1 - |\mathbf{w}^{(\ell)}| \right)} \right\} \right|_{z^0_{\ell} = 0} = \frac{\xi_{\ell}!}{(\xi_{\ell} - |\mathbf{v}_{\ell}|)!} \left( 1 - |\mathbf{w}^{(\ell)}| \right)^{\xi_{\ell} - |\mathbf{v}_{\ell}|} \theta(\xi_{\ell} - |\mathbf{v}_{\ell}|)
\end{equation}
where we have used the result from earlier (Eq. \ref{eq:constterm}) to evaluate it. When enforcing the delta function constraint that $z^0_{\ell} = 0$ for all $\ell = 1, ..., n$, we get
\begin{equation}
\frac{e^{-|\boldsymbol{\lambda}(t)|}}{\mathbf{x}!} \sum_{v^j_k} \left\{ \prod_{j = 1}^n \binom{x_j}{v^j_1 \cdots v^j_{n+1}}  \left[ w^{(1)}_j \right]^{v^j_1} \cdots \left[ w^{(n)}_j  \right]^{v^j_n} \left[ \lambda_j \right]^{v^j_{n+1}} \frac{\xi_j!}{(\xi_j - |\mathbf{v}_j|)!} \left( 1 - |\mathbf{w}^{(j)}| \right)^{\xi_j - |\mathbf{v}_j|} \theta(\xi_j - |\mathbf{v}_j|) \right\} 
\end{equation}
for $P$. This is the final result, but let us rewrite it so that we recover the result from Theorem 1 (Eq. \ref{eq:fullsln}) of Jahnke and Huisinga's paper. Note that
\begin{equation}
e^{-|\boldsymbol{\lambda}(t)|} \prod_{j = 1}^n \frac{\left[ \lambda_j \right]^{v^j_{n+1}}}{v^j_{n+1}!} =  \frac{\boldsymbol{\lambda}(t)^{\mathbf{v}_{n+1}}}{\mathbf{v}_{n+1}!} e^{-|\boldsymbol{\lambda}(t)|} = \mathcal{P}(\mathbf{v}_{n+1}, \boldsymbol{\lambda}(t)) \ .
\end{equation}
Also,
\begin{equation}
\begin{split}
& \frac{\xi_k! \left( 1 - |\mathbf{w}^{(k)}| \right)^{\xi_k - |\mathbf{v}_k|}}{(\xi_k - |\mathbf{v}_k|)!} \theta(\xi_k - |\mathbf{v}_k|) \prod_{j = 1}^n \frac{\left[ w^{(k)}_j \right]^{v^j_k}}{v^j_k!}  \\
=&  \frac{\xi_k! \left( 1 - |\mathbf{w}^{(k)}| \right)^{\xi_k - |\mathbf{v}_k|} }{(\xi_k - |\mathbf{v}_k|)!}  \theta(\xi_k - |\mathbf{v}_k|) \frac{\left[ \mathbf{w}^{(k)} \right]^{\mathbf{v}_k}}{\mathbf{v}_k!} \\
=& \mathcal{M}(\mathbf{v}_k, \xi_k, \mathbf{w}^{(k)}) \ .
 \end{split}
\end{equation}
We are left with
\begin{equation}
\begin{split}
P &=  \sum_{v^j_k} \mathcal{P}(\mathbf{v}_{n+1}, \boldsymbol{\lambda}(t)) \ \mathcal{M}(\mathbf{v}_1, \xi_1, \mathbf{w}^{(1)}) \cdots \mathcal{M}(\mathbf{v}_n, \xi_n, \mathbf{w}^{(n)})  \\
&=  \sum_{v^j_k} \mathcal{P}(\mathbf{x} - \mathbf{v}_1 - \cdots - \mathbf{v}_n, \boldsymbol{\lambda}(t)) \ \mathcal{M}(\mathbf{v}_1, \xi_1, \mathbf{w}^{(1)}) \cdots \mathcal{M}(\mathbf{v}_n, \xi_n, \mathbf{w}^{(n)})  \\
&=  \mathcal{P}(\mathbf{x}, \boldsymbol{\lambda}(t)) \star \mathcal{M}(\mathbf{x}, \xi_1, \mathbf{w}^{(1)}(t)) \star \cdots \star \mathcal{M}(\mathbf{x}, \xi_n, \mathbf{w}^{(n)}(t))
\end{split}
\end{equation}
which matches Eq. \ref{eq:fullsln}.  \qed
\end{proof}
%\begin{theorem}
%The general solution of the monomolecular system with initial distribution BLANK takes the following form:
%\end{theorem}
%
%\begin{proposition}
%The one-dimensional result is BLANK:
%\end{proposition}

%\section{From the propagator to moments}
%\label{sec:proptomoments}

If we wanted to compute the moments of $P(\mathbf{x}, t)$, we could just use Eq. \ref{eq:fullsln} and carry out the calculation directly; however, the Doi-Peliti approach offers a way to compute moments which bypasses $P(\mathbf{x}, t)$ completely. In other words, if we are \textit{only} interested in moments, the work from the previous section is unnecessary. Instead, we can use Proposition \ref{prop:gfto}. As with the previous calculation, we will warm up with the one species case before treating the multi-species case. 

\subsection{One species moments derivation}
\label{sec:onespeciesmoments}

\begin{lemma}[One species monomolecular moments]
For the single species monomolecular system (i.e. the chemical birth-death process), the first and second factorial moments are
\begin{equation} 
\begin{split}
\expval{x(t)} &= \xi w(t) + \lambda(t) \\
\expval{x(t)[x(t) - 1]} &= w(t)^2 \xi (\xi - 1) + 2 \lambda(t) w(t) \xi + \lambda(t)^2   
\end{split}
\end{equation}
where $w(t)$ and $\lambda(t)$ are as defined in Theorem \ref{thm:bd}.
\end{lemma}

\begin{proof}
Using Eq. \ref{eq:gftomoments}, 
\begin{equation}
\begin{split}
\expval{x(t)} &=  \matrixel{0}{\hat{a}}{\psi(t)} \\
&= \int \frac{dz^f dp^f}{2 \pi} \frac{dz^0 dp^0}{2 \pi} \mel{0}{\hat{a}}{z^f} \ U(ip^f, z^0) \braket{- i p^0}{\psi(t_0)} \ e^{- i p^0 z^0  - i p^f z^f}  \\
&= \int \frac{dz^f dp^f}{2 \pi} \frac{dz^0 dp^0}{2 \pi} z^f  \ e^{i p^f z(t)} (1 + i p^0)^{\xi} \ e^{- i p^0 z^0  - i p^f z^f} \ .
\end{split}
\end{equation}
The $p^f$, $z^f$, and $p^0$ integrals can be done as in Sec. \ref{sec:onespeciesptrans}, leaving
\begin{equation}
\expval{x(t)}  = \int_{0}^{\infty} dz^0 \ \left( \frac{d}{dz^0} \right)^{\xi} \left\{ \left[ z^0 w(t) + \lambda(t) \right] e^{z^0}  \right\} \ e^{- z^0} \delta(z^0) \ .
\end{equation}
The derivative is easily evaluated, and we obtain
\begin{equation}
\expval{x(t)}  = \int_{0}^{\infty} dz^0 \  \left[ \xi w(t) e^{z^0} + z(t) e^{z^0}  \right] \ e^{- z^0} \delta(z^0) = \xi w(t) + \lambda(t)
\end{equation}
which is just the solution to the one species reaction rate equation with $x(t_0) = \xi$, just as expected. The second factorial moment can be computed in similar fashion:
\begin{equation}
\begin{split}
\expval{x(t)[x(t) - 1]} &=  \matrixel{0}{\hat{a}^2}{\psi(t)} \\
&= \int \frac{dz^f dp^f}{2 \pi} \frac{dz^0 dp^0}{2 \pi} \left( z^f \right)^2 \ e^{i p^f z(t)} (1 + i p^0)^{\xi} \ e^{- i p^0 z^0  - i p^f z^f} \\
&= \int_{0}^{\infty} dz^0 \ \left( \frac{d}{dz^0} \right)^{\xi} \left\{ \left[ z^0 w(t) + \lambda(t) \right]^2 e^{z^0}  \right\} \ e^{- z^0} \delta(z^0) \\
&= w(t)^2 \xi (\xi - 1) + 2 \lambda(t) w(t) \xi + \lambda(t)^2 \ .
\end{split}
\end{equation}
\qed
\end{proof}
Higher factorial moments can be computed in exactly the same way.

\subsection{General moments derivation}

Unlike in the one species case, there are many first moments: $\expval{x_1(t)}, ..., \expval{x_n(t)}$. There are also many second moments. To summarize them usefully, we compute the covariance matrix elements (i.e. $\mathrm{Cov}(x_j, x_{\ell}) := \langle x_j(t) x_{\ell}(t) \rangle - \langle x_j(t) \rangle \langle x_{\ell}(t) \rangle$ for all pairs of $j$ and $\ell$). 

\begin{lemma}[Monomolecular moments]
For the general monomolecular system, the first moments, second factorial moments, and covariance matrix elements are given by
\begin{equation} 
\begin{split}
\expval{x_j(t)}  =& \sum_{k=1}^n  {\xi}_k w^{(k)}_j(t) + \lambda_j(t) \hspace{0.5in} j = 1, ..., n \\
\expval{x_j x_{\ell}}  =& \sum_{k=1}^n \sum_{k' = 1}^n  {\xi}_k {\xi}_{k'} w^{(k)}_j w^{(k')}_{\ell} \\
&+ \sum_{k=1}^n {\xi}_k \left[ w^{(k)}_j \lambda_{\ell} + w^{(k)}_{\ell} \lambda_{j} - w^{(k)}_j w^{(k)}_{\ell} \right] + \lambda_j \lambda_{\ell} \hspace{0.5in}  j \neq \ell \\
\expval{x_j(t)[x_{j}(t) - 1]}  =& \sum_{k=1}^n \sum_{k' = 1}^n  {\xi}_k {\xi}_{k'} w^{(k)}_j w^{(k')}_{j} \\
&+ \sum_{k=1}^n {\xi}_k \left[ 2 w^{(k)}_j \lambda_{j} - \left( w^{(k)}_j\right)^2 \right] + \lambda_j^2 \hspace{0.5in} j = 1, ..., n \\
\mathrm{Cov}(x_j, x_{\ell}) =&  \begin{cases} 
      \sum_{k = 1}^n \xi_k w^{(k)}_j \left[ 1 - w^{(k)}_j  \right] + \lambda_j & j = \ell \\
      - \sum_{k = 1}^n \xi_k  w^{(k)}_j w^{(k)}_{\ell} & j \neq \ell
   \end{cases} 
\end{split}
\end{equation}
where $\boldsymbol{\lambda}(t)$ and the $\mathbf{w}^{(j)}(t)$ are as defined in Theorem \ref{thm:mono}. 
\end{lemma}

\begin{proof}
Picking a specific $x_j$ and using Eq. \ref{eq:gftomoments}, we have
\begin{equation}
\begin{split}
\expval{x_j(t)} &=  \matrixel{\mathbf{0}}{\hat{a}_j}{\psi(t)} \\
&= \int \frac{d\mathbf{z}^f d\mathbf{p}^f}{(2 \pi)^n} \frac{d\mathbf{z}^0 d\mathbf{p}^0}{(2 \pi)^n} \mel{\mathbf{0}}{\hat{a}_j}{\mathbf{z}^f} \ U(i \mathbf{p}^f, \mathbf{z}^0) \braket{- i \mathbf{p}^0}{\psi(t_0)} \ e^{- i \mathbf{p}^0 \cdot \mathbf{z}^0  - i \mathbf{p}^f \cdot \mathbf{z}^f} \\
&= \int \frac{d\mathbf{z}^f d\mathbf{p}^f}{(2 \pi)^n} \frac{d\mathbf{z}^0 d\mathbf{p}^0}{(2 \pi)^n} z_j^f  \ e^{i \mathbf{p}^f \cdot \mathbf{z}(t)} (\mathbf{1} + i \mathbf{p}^0)^{\boldsymbol{\xi}} \ e^{- i \mathbf{p}^0 \cdot \mathbf{z}^0  - i \mathbf{p}^f \cdot \mathbf{z}^f} \ .
\end{split}
\end{equation}
The $\mathbf{p}^f$, $\mathbf{z}^f$, and $\mathbf{p}^0$ integrals can be done as in Sec. \ref{sec:genptrans}, yielding
\begin{equation}
\begin{split}
\expval{x_j(t)}  &= \int d\mathbf{z}^0 \ \left( \frac{d}{d\mathbf{z}^0} \right)^{\boldsymbol{\xi}} \left\{ \left[ \sum_{k=1}^n  z^0_k w^{(k)}_j + \lambda_j \right] e^{\mathbf{z}^0 \cdot \mathbf{1}}  \right\} \ e^{- \mathbf{z}^0 \cdot \mathbf{1}} \delta(\mathbf{z}^0) \\
&= \sum_{k=1}^n  {\xi}_k w^{(k)}_j(t) + \lambda_j(t)
\end{split}
\end{equation}
which is the $j$th component of the solution to Eq. \ref{eq:RRE} with $\mathbf{x}(t_0) = \boldsymbol{\xi}$. 

Let us compute $\expval{x_j(t) x_{\ell}(t)}$ for $j \neq \ell$. To start off,
\begin{equation}
\begin{split}
\expval{x_j(t) x_{\ell}(t)} &=  \matrixel{\mathbf{0}}{\hat{a}_j \hat{a}_{\ell}}{\psi(t)} \\
&= \int \frac{d\mathbf{z}^f d\mathbf{p}^f}{(2 \pi)^n} \frac{d\mathbf{z}^0 d\mathbf{p}^0}{(2 \pi)^n} \mel{\mathbf{0}}{\hat{a}_j \hat{a}_{\ell}}{\mathbf{z}^f} \ U(i \mathbf{p}^f, \mathbf{z}^0) \braket{- i \mathbf{p}^0}{\psi(t_0)} \ e^{- i \mathbf{p}^0 \cdot \mathbf{z}^0  - i \mathbf{p}^f \cdot \mathbf{z}^f} \\
&= \int \frac{d\mathbf{z}^f d\mathbf{p}^f}{(2 \pi)^n} \frac{d\mathbf{z}^0 d\mathbf{p}^0}{(2 \pi)^n} z_j^f z_{\ell}^f  \ e^{i \mathbf{p}^f \cdot \mathbf{z}(t)} (\mathbf{1} + i \mathbf{p}^0)^{\boldsymbol{\xi}} \ e^{- i \mathbf{p}^0 \cdot \mathbf{z}^0  - i \mathbf{p}^f \cdot \mathbf{z}^f} \ .
\end{split}
\end{equation}
Proceeding as we just did, we obtain
\begin{equation}
\begin{split}
\expval{x_j x_{\ell}}  &= \int d\mathbf{z}^0 \ \left( \frac{d}{d\mathbf{z}^0} \right)^{\boldsymbol{\xi}} \left\{ \left[ \sum_{k=1}^n  z^0_k w^{(k)}_j + \lambda_j \right] \left[ \sum_{k'=1}^n  z^0_{k'} w^{(k')}_{\ell} + \lambda_{\ell} \right] e^{\mathbf{z}^0 \cdot \mathbf{1}}  \right\} \ e^{- \mathbf{z}^0 \cdot \mathbf{1}} \delta(\mathbf{z}^0) \\
&= \sum_{k=1}^n \sum_{k' = 1}^n  {\xi}_k {\xi}_{k'} w^{(k)}_j w^{(k')}_{\ell} + \sum_{k=1}^n {\xi}_k \left[ w^{(k)}_j \lambda_{\ell} + w^{(k)}_{\ell} \lambda_{j} - w^{(k)}_j w^{(k)}_{\ell} \right] + \lambda_j \lambda_{\ell} \ .
\end{split}
\end{equation}
For the similar case $j = \ell$, we obtain
\begin{equation}
\begin{split}
\expval{x_j(t)[x_{j}(t) - 1]}  &= \int d\mathbf{z}^0 \ \left( \frac{d}{d\mathbf{z}^0} \right)^{\boldsymbol{\xi}} \left\{ \left[ \sum_{k=1}^n  z^0_k w^{(k)}_j + \lambda_j \right]^2 e^{\mathbf{z}^0 \cdot \mathbf{1}}  \right\} \ e^{- \mathbf{z}^0 \cdot \mathbf{1}} \delta(\mathbf{z}^0) \\
&= \sum_{k=1}^n \sum_{k' = 1}^n  {\xi}_k {\xi}_{k'} w^{(k)}_j w^{(k')}_{j} + \sum_{k=1}^n {\xi}_k \left[ 2 w^{(k)}_j \lambda_{j} - \left( w^{(k)}_j\right)^2 \right] + \lambda_j^2 \ .
\end{split}
\end{equation}
Putting these results together, we find that the covariance of $x_j$ and $x_{\ell}$ is
\begin{equation}
\text{Cov}(x_j, x_{\ell}) =  \begin{cases} 
      \sum_{k = 1}^n \xi_k w^{(k)}_j \left[ 1 - w^{(k)}_j  \right] + \lambda_j & j = \ell \\
      - \sum_{k = 1}^n \xi_k  w^{(k)}_j w^{(k)}_{\ell} & j \neq \ell
   \end{cases} \ .
\end{equation}
Hence, we have recovered the moment results from Sec. 4.2 of Jahnke and Huisinga. \qed \end{proof}

%\begin{theorem}
%The time-dependent moments of the monomolecular system with initial distribution BLANK takes the following form:
%\end{theorem}
%
%\begin{proposition}
%The one-dimensional result is BLANK:
%\end{proposition}

\section{Birth-death-autocatalysis calculations}
\label{sec:bdacalc}

In this section, we present the calculations relevant to proving the formulas from Theorem \ref{thm:bda} on the birth-death-autocatalysis system. First, we present the Hamiltonian operator and kernel. Then we evaluate the path integral expression for the propagator $U(i \mathbf{p}^f, \mathbf{z}^0)$. Finally, we use the explicit form of the propagator to derive the transition probability, and several interesting limiting forms of it. We do not explicitly show how to compute the generating function directly from the propagator, because it is very similar to the other calculations.

%In section 6 of their classic paper \cite{jahnke2007}, Jahnke and Huisinga solve the CME corresponding to the autocatalytic reaction $S \to S + S$ exactly; however, they note that adding birth and death reactions yields a system not amenable to their approach. In this section, we present the exact time-dependent solution to this problem, whose reactions read
%\begin{equation} \label{eq:newrxns}
%\begin{split}
%\varnothing &\xrightarrow{k} S  \\
%S &\xrightarrow{\gamma} \varnothing  \\
%S &\xrightarrow{c} S + S  
%\end{split}
%\end{equation}
%where the rates of birth, death, and autocatalysis are all allowed to have arbitrary time-dependence as long as they are nonnegative for all times.  The CME reads
%\begin{equation} \label{eq:newCME}
%\begin{split}
%\frac{\partial P(x, t)}{\partial t} =& \ k(t) \left[ P(x-1, t) - P(x, t) \right] \\
%&+ \gamma(t) \left[ (x+1) P(x+1, t) - x P(x, t) \right] \\
%&+ c(t) \left[ (x-1) P(x-1, t) - x P(x, t) \right] 
%\end{split}
%\end{equation}
%where $P(x, t)$ is the probability that the state of the system is $x \in \mathbb{N}$ at time $t \geq t_0$. 
%
%
%

\subsection{Evaluating the propagator}

%The Hamiltonian operator corresponding to this problem is
%\begin{equation}
%\hat{H} := \hat{a}^+ \left[ k + (c - \gamma) \hat{a} + c \ \hat{a}^+ \hat{a} \right]
%\end{equation}
%in terms of the Grassberger-Scheunert creation and annihilation operators. 

We can straightforwardly go from the CME (Eq. \ref{eq:newCME}) to the Hamiltonian operator and Hamiltonian kernel.

\begin{lemma}
The Hamiltonian operator corresponding to the birth-death-autocatalysis CME (Eq. \ref{eq:newCME}) is
\begin{equation} \label{eq:bda_H}
\begin{split}
\hat{H} = \hat{a}^+ \left[ k + (c - \gamma) \hat{a} + c \ \hat{a}^+ \hat{a} \right] \ .
\end{split}
\end{equation}
\end{lemma}
\begin{proof}
Starting with Eq. \ref{eq:newCME}, follow the argument from Lemma \ref{lem:monoHderiv}, and then substitute in the Grassberger-Scheunert creation operator. \qed
\end{proof}

\begin{corollary}
The Hamiltonian kernel for the birth-death-autocatalysis CME is
\begin{equation} \label{eq:bda_Hkernel}
\begin{split}
\mathcal{H}(i p, z, t) = i p \left[ k + (c - \gamma) z \right] - c \ p^2 z \ .
\end{split}
\end{equation}
\end{corollary}
\begin{proof}
Make the identifications $\hat{a}^+ \to i p$ and $\hat{a} \to z$ in the Hamiltonian above. \qed
\end{proof}

Now we must compute the propagator, a calculation which turns out to be somewhat involved. 

\begin{lemma}[Birth-death-autocatalysis propagator] \label{lem:propcalc_bda}
The propagator for the birth-death-autocatalysis system is
\begin{equation}
U(i p_f, z_0) =  \exp\left\{ i z_0 q(t) + i \int_{t_0}^t k(s) q(t - s + t_0)  \ ds \right\} 
\end{equation}
where $q(s)$ is as in Theorem \ref{thm:bda}.
\end{lemma}

\begin{proof}
The path integral expression for the propagator $U(i p_f, z_0)$ is
\begin{equation}
\begin{split}
U = \ & \lim_{N \to \infty} \int \prod_{\ell = 1}^{N-1} \frac{dz_{\ell} dp_{\ell} }{2\pi} \exp\left\{ \sum_{\ell = 1}^{N-1}  - i p_{\ell} (z_{\ell} - z_{\ell - 1}) + \Delta t \mathcal{H}(i p_{\ell}, z_{\ell - 1}, t_{\ell-1}) \right. \\
 & \left. + \Delta t \mathcal{H}(i p_f, z_{N - 1}, t_{\ell-1}) + i p_f z_{N-1}  \right\}
 \end{split}
\end{equation}
where we have used slightly different notation than before since there is only one chemical species. In order to evaluate this path integral, first integrate over each $z_{\ell}$, and then integrate over each $p_{\ell}$. Collecting terms containing $z_{\ell}$, the integral over each $z_{\ell}$ looks like
\begin{equation}
\begin{split}
& \int_{0}^{\infty} \frac{dz_{\ell}}{2 \pi} \ \exp\left\{ z_{\ell} \left[ - c_{\ell} \Delta t \ p_{\ell + 1}^2  + i (c_{\ell} - \gamma_{\ell}) \Delta t \ p_{\ell + 1} - i (p_{\ell} - p_{\ell + 1}) \right] \right\} \\
=& \frac{1}{2\pi i} \frac{1}{(p_{\ell} - p_{\ell + 1}) - \Delta t \left[   (c_{\ell} - \gamma_{\ell})  \ p_{\ell + 1} + i c_{\ell} \ p_{\ell + 1}^2\right]  } \ .
\end{split}
\end{equation}
The integrals over $p_{\ell}$ can now be done---but they must be done in a specific order. Do the integral over $p_{N-1}$, then $p_{N-2}$, and so on, until the integral over $p_1$ has been done. Each of these integrals is schematically
\begin{equation}
\frac{1}{2\pi i} \int_{-\infty}^{\infty} dp_{\ell} \  \frac{f(p_{\ell})}{(p_{\ell} - p_{\ell + 1}) - \Delta t \left[   (c_{\ell} - \gamma_{\ell})  \ p_{\ell + 1} + i c_{\ell} \ p_{\ell + 1}^2\right]  }
\end{equation}
where the function $f(p_{\ell})$ has no poles. This means that each integral can be evaluated using Cauchy's integral formula, so that the net effect of doing them is to enforce the $(N-1)$ constraints
\begin{equation} \label{eq:newstep}
p_{\ell}  = p_{\ell + 1} + \Delta t \left[   (c_{\ell} - \gamma_{\ell})  \ p_{\ell + 1} + i c_{\ell} \ p_{\ell + 1}^2\right] 
\end{equation}
on the $p_{\ell}$ for $\ell = 1, ..., N-1$. There are no more integrals to do, so all that remains is to evaluate what's left of the propagator using these constraints. Eq. \ref{eq:newstep} looks like an Euler time step, although it is `backwards'---we go from $p_{\ell + 1}$ to $p_{\ell}$ instead of the other way around. Define $q_{N- \ell} := p_{\ell}$ so that it reads
\begin{equation} \label{eq:newstepQ}
q_{N - \ell}  = q_{N - \ell - 1} + \Delta t \left[   (c_{\ell} - \gamma_{\ell})  \ q_{N - \ell - 1} + i c_{\ell} \ q_{N - \ell - 1}^2\right] \ .
\end{equation}
Choosing $\ell = N - n$, we find
\begin{equation} 
q_{n}  = q_{n - 1} + \Delta t \left[   (c_{N - n} - \gamma_{N - n})  \ q_{n - 1} + i c_{N - n} \ q_{n - 1}^2\right] \ .
\end{equation}
This corresponds to dynamics
\begin{equation} \label{eq:newDE}
\dot{q}(s) = \left[ c(t - s + t_0) - \gamma(t - s + t_0) \right] q(s) + i c(t - s + t_0) \ q(s)^2 
\end{equation}
where $s \in [t_0, t]$ and $q(t_0) = p_f$. As can be verified by substitution, Eq. \ref{eq:newDE} is solved by
\begin{equation}
q(s) = \frac{w(s)}{\frac{1}{p_f} - i \int_{t_0}^s c(t - t' + t_0) w(t') \ dt'}
\end{equation}
where $w(t)$ is the solution to
\begin{equation}
\dot{w}(s) = [c(t - s + t_0) - \gamma(t - s + t_0)] \ w(s)
\end{equation}
with $w(t_0) = 1$ (c.f. Eq. \ref{eq:wlambda}), i.e.
\begin{equation}
w(s) = e^{\int_{t_0}^s c(t - t' + t_0) - \gamma(t - t' + t_0) \ dt'} \ .
\end{equation}
The continuous limit of Eq. \ref{eq:newstep} is then $p(s) := q(t - s + t_0)$. With that done, the propagator with most terms integrated out reads
\begin{equation}
U(i p_f, z_0) = \lim_{N \to \infty}  \exp\left\{ i \sum_{\ell = 1}^{N-1} k_{\ell - 1} p_{\ell} \ \Delta t + i p_1 z_0 + \Delta t \left[ i p_1 (c_0 - \gamma_0) z_0 - c_0 p_1^2 z_0 \right]  \right\} \ .
\end{equation}
The term on the right is just another Euler time step, so we can write it as
\begin{equation}
i z_0 \left\{ p_1 + \Delta t \left[ p_1 (c_0 - \gamma_0) - c_0 p_1^2 \right]    \right\} = i z_0 p_0
\end{equation}
where we define
\begin{equation}
p_0 := p_1 + \Delta t \left[ p_1 (c_0 - \gamma_0) - c_0 p_1^2 \right] \ .
\end{equation}
In the limit as $N \to \infty$, $p_0 \to p(t_0) = q(t)$. The term on the left is just a Riemann sum:
\begin{equation}
\begin{split}
\sum_{\ell = 1}^{N-1} k_{\ell - 1} p_{\ell} \ \Delta t \approx \int_{t_0}^t k(s) p(s) \ ds = \int_{t_0}^t k(s) q(t - s + t_0)  \ ds \ .
\end{split}
\end{equation}
Hence, our final answer for the propagator $U$ is
\begin{equation}
U(i p_f, z_0) =  \exp\left\{ i z_0 q(t) + i \int_{t_0}^t k(t - s + t_0) q(s)  \ ds \right\} 
\end{equation}
where we have reparameterized the integral on the right to swap $s$ and $(t - s + t_0)$.  \qed
\end{proof}
As an aside, we note that this calculation closely resembles the Martin-Siggia-Rose-Janssen-De Dominicis path integral computation from our earlier paper \cite{vastola2019gill}: in particular, many applications of Cauchy's integral formula and another `backwards' Euler time step constraint are both involved.

\subsection{Deriving the transition probability}

As in Sec. \ref{sec:onespeciesptrans} and \ref{sec:genptrans}, we will use the propagator derived in the previous section to derive an expression for the transition probability $P(x, t; \xi, t_0)$.

\begin{lemma}[Birth-death-autocatalysis transition probability]
For the birth-death-autocatalysis system, the transition probability $P(x, t; \xi, t_0)$ is
\begin{equation}
\begin{split}
P(x, t; \xi, t_0) =& \frac{1}{2 \pi} \int_{-\infty}^{\infty} dp_f \ \frac{\left[ 1 + i q(t) \right]^{\xi} e^{i \int_{t_0}^t k(t - s + t_0) q(s) ds}}{(1 + i p_f)^{x + 1}} 
\end{split}
\end{equation}
where $q(s)$ is as in Theorem \ref{thm:bda}.
\end{lemma}

\begin{proof}
Since $P(x, t_0) = \delta(x - \xi)$, we have $\ket{\psi_0} = \ket{\xi}$. Using Eq. \ref{eq:gftop} and Eq. \ref{eq:psiU},
\begin{equation}
\begin{split}
P(x, t; \xi, t_0) &= \frac{1}{x!} \int \frac{dz_f dp_f}{2 \pi} \frac{dz_0 dp_0}{2 \pi} \braket{x}{z_f}_{ex} \ U(ip_f, z_0) \braket{- i p_0}{\psi(t_0)} \ e^{- i p_0 z_0  - i p_f z_f}  \\
&= \frac{1}{x!} \int \frac{dz_f dp_f}{2 \pi} \frac{dz_0 dp_0}{2 \pi} \left(z_f \right)^x e^{- z_f} \ e^{ i z_0 q(t) + i \int_{t_0}^t k(t - s + t_0) q(s)  \ ds }  (1 + i p_0)^{\xi} \ e^{- i p_0 z_0  - i p_f z_f} \ .
\end{split}
\end{equation}
The integral over $z_0$ is
\begin{equation}
\int_0^{\infty} \frac{dz_0}{2 \pi} \ e^{- i z_0 \left[ p_0 - q(t) \right]} = \frac{1}{2 \pi i} \frac{1}{p_0 - q(t)} \ .
\end{equation}
The integral over $p_0$ can be performed using Cauchy's integral formula:
\begin{equation}
\frac{1}{2 \pi i} \int_{-\infty}^{\infty} dp_0 \ \frac{(1 + i p_0)^{\xi}}{p_0 - q(t)} = \left[ 1 + i q(t) \right]^{\xi} \ .
\end{equation}
The integral over $z_f$ can be recognized as a Laplace transform:
\begin{equation}
\int_0^{\infty} dz_f \ \left(z_f \right)^x e^{- z_f \left[ 1 + i p_f \right]} = \frac{x!}{\left( 1 + i p_f \right)^{x + 1}} \ .
\end{equation}
Putting these together, we obtain the desired result.
\qed \end{proof}
We will leave our solution in this form, since it is difficult to evaluate the contour integral without knowing the explicit time-dependence of the rates. In the next few sections, we will examine a few special cases.

\subsection{Time-independent rates}

\begin{lemma}[Birth-death-autocatalysis transition probability for time-independent rates]\label{lem:bda_timeind}
Suppose the parameters $k$, $\gamma$, and $c$ are all time-independent and non-zero. Then the transition probability can be rewritten as
\begin{equation}
\begin{split}
P =& \left( \frac{\frac{\gamma}{c} - 1}{\frac{\gamma}{c} - w} \right)^{k/c}  \frac{\left( 1 - w \right)^{x - \xi}}{\left(\frac{\gamma}{c} - w \right)^x}  \times \\
 &\times \sum_{j = 0}^{\xi} \binom{\xi}{j} \frac{\left( j + k/c \right)_x}{x!} \left[ 1 - \frac{\gamma}{c} w \right]^{\xi - j}  \left[ \frac{w \left( \frac{\gamma}{c} - 1 \right)^2}{\frac{\gamma}{c} - w} \right]^j  
\end{split}
\end{equation}
where $(y)_x := (y)(y + 1) \cdots (y + x - 1)$ is the Pochhammer symbol/rising factorial, and where $w(t) = e^{- (\gamma - c) (t - t_0)}$.
\end{lemma}

\begin{proof}
In this case, $q(t)$ reads
\begin{equation}
\begin{split}
\dot{q} &= [c - \gamma] \ q + i c \ q^2 \\
q(t) &= \frac{e^{(c - \gamma) T}}{\frac{1}{p_f} - i \frac{c}{c - \gamma} \left[ e^{(c - \gamma) T} - 1 \right]} = \frac{w(t)}{\frac{1}{p_f} - i \frac{c}{c - \gamma} \left[ w(t) - 1 \right]}
\end{split}
\end{equation}
where $T := t - t_0$. We have
\begin{equation}
\int_{t_0}^t q(s) \ ds = \frac{i}{c} \log\left\{ 1 - \frac{i c}{c - \gamma}  \left[ e^{(c - \gamma) T} - 1 \right]  p_f \right\}
\end{equation}
so that the convolution term from the propagator reads
\begin{equation} \label{eq:convolutionterm}
e^{i k \int_{t_0}^t q(s) \ ds } = \frac{1}{\left[ 1 - \frac{i c}{c - \gamma}  \left[ e^{(c - \gamma) T} - 1 \right]  p_f \right]^{k/c}} = \frac{1}{\left[ 1 - i B(t) p_f \right]^{k/c}} 
\end{equation}
where we define
\begin{equation}
B(t) := \frac{c}{c - \gamma} \left[ w(t) - 1 \right] \ .
\end{equation}
It is important to note that Eq. \ref{eq:convolutionterm} has no poles in the upper half-plane (the region around which we are integrating), regardless of whether $c - \gamma > 0$, $c - \gamma < 0$, or $c = \gamma$. Next, 
\begin{equation}
1 + i q(t) = 1 + \frac{i w(t) p_f}{1 - i B(t) p_f} = \left[ 1 - \frac{w(t)}{B(t)}  \right] + \frac{w(t)}{B(t)} \frac{1}{\left[ 1 - i B(t) p_f \right]}
\end{equation}
so that
\begin{equation}
\begin{split}
\left[ 1 + i q(t) \right]^{\xi} &= \sum_{j = 0}^{\xi} \binom{\xi}{j} \left[ 1 - \frac{w(t)}{B(t)}  \right]^{\xi - j} \left( \frac{w(t)}{B(t)} \right)^j \frac{1}{\left[ 1 - i B(t) p_f \right]^j} \ .
\end{split}
\end{equation}
Putting all these results together, our expression for the transition probability is
\begin{equation}
P =  \sum_{j = 0}^{\xi} \binom{\xi}{j} \left[ 1 - \frac{w(t)}{B(t)}  \right]^{\xi - j} \left( \frac{w(t)}{B(t)} \right)^j \frac{1}{x! \ i^x} \left\{ \frac{x!}{2 \pi i} \int_{-\infty}^{\infty} dp_f \ \frac{1}{\left[ 1 - i B(t) p_f \right]^{j+ k/c}} \frac{1}{(p_f - i)^{x + 1}} \right\} \ .
\end{equation}
Since
\begin{equation}
\frac{d^x}{dp_f^x} \left[ \frac{1}{\left[ 1 - i B(t) p_f \right]^{j+ k/c}} \right]_{p = i} = \frac{i^x B(t)^x}{\left[ 1 + B(t) \right]^{j + k/c + x}} \left( j + \frac{k}{c} \right) \left( j + \frac{k}{c} + 1 \right) \cdots \left( j + \frac{k}{c} + x - 1 \right)
\end{equation}
we have
\begin{equation}
\begin{split}
P &=  \sum_{j = 0}^{\xi} \binom{\xi}{j} \left[ 1 - \frac{w(t)}{B(t)}  \right]^{\xi - j} \left( \frac{w(t)}{B(t)} \right)^j \frac{\left( j + k/c \right)_x}{x!}  \frac{B(t)^x}{\left[ 1 + B(t) \right]^{j + k/c + x}} \\
&= \left( \frac{1 - \frac{\gamma}{c} }{w - \frac{\gamma}{c} } \right)^{k/c} \left( \frac{w - 1}{w - \frac{\gamma}{c}} \right)^x \sum_{j = 0}^{\xi} \binom{\xi}{j} \frac{\left( j + k/c \right)_x}{x!} \left[ 1 - \left( 1 - \frac{\gamma}{c} \right) \frac{w}{w - 1}  \right]^{\xi - j}  \left[ \frac{w \left( 1 - \frac{\gamma}{c} \right)^2}{(w - 1)(w - \frac{\gamma}{c})} \right]^j \\
=& \left( \frac{\frac{\gamma}{c} - 1}{\frac{\gamma}{c} - w} \right)^{k/c}  \frac{\left( 1 - w \right)^{x - \xi}}{\left(\frac{\gamma}{c} - w \right)^x}  \times \\
 &\times \sum_{j = 0}^{\xi} \binom{\xi}{j} \frac{\left( j + k/c \right)_x}{x!} \left[ 1 - \frac{\gamma}{c} w \right]^{\xi - j}  \left[ \frac{w \left( \frac{\gamma}{c} - 1 \right)^2}{\frac{\gamma}{c} - w} \right]^j  
\end{split}
\end{equation}
where $(y)_x := (y)(y + 1) \cdots (y + x - 1)$ is the Pochhammer symbol/rising factorial. This can also be written in terms of the hypergeometric function $_2F_1(a, b; c; x)$. \qed
\end{proof}

%Assuming that $\gamma > c$, taking the time length $T \to \infty$ yields the steady state solution 
%\begin{equation}
%P_{ss}(x) = \left( \frac{\gamma - c}{\gamma} \right)^{k/c} \frac{\left( \frac{c}{\gamma} \right)^x}{x!} \left( \frac{k}{c}  \right)_x \ .
%\end{equation}
%It is easy to check that this solution is normalized and solves the (steady state) CME. In the limit as $c \to 0$, we have a birth-death process, and the steady state probability distribution becomes
%\begin{equation}
%P_{ss}(x) = \frac{\left( \frac{k}{\gamma} \right)^x e^{-k/\gamma}}{x!}
%\end{equation}
%which is Poisson, as expected. In the limit taking $k \to 0$ (while keeping $\gamma$ and $c$ finite), we have $P_{ss}(x) = \delta(x)$, which is also expected (a system with only autocatalytic and death reactions, with $\gamma > c$, has all of its probability concentrated in $x = 0$ at steady state). 

\subsection{Binomial, Poisson, and negative binomial special cases}

\begin{proof}[Corollary \ref{cor:three_cases}]

Return to the original contour integral for time-dependent rates (Eq. \ref{eq:contoursln2}), and set $k = c = 0$, but leave the time-dependence of $\gamma(t)$ arbitrary. We have
\begin{equation}
w(t) := \exp\left[ -\int_{t_0}^t \gamma(t') dt' \right] 
\end{equation}
\begin{equation}
q(t) = w(t) p_f
\end{equation}
\begin{equation}
P(x, t; \xi, t_0) = \frac{1}{2 \pi} \int_{-\infty}^{\infty} dp_f \ \frac{\left[ 1 + i w(t) p_f \right]^{\xi}}{(1 + i p_f)^{x + 1}} \ .
\end{equation}
The function in the numerator has no poles, so the contour integral can easily be evaluated using Cauchy's integral formula. The result is
\begin{equation}
P(x, t; \xi, t_0) = \binom{\xi}{x} \left[ w(t) \right]^x \left[ 1 - w(t) \right]^{\xi - x}
\end{equation}
for $x \leq \xi$ and $0$ otherwise, i.e. a binomial distribution. 

%\subsection{Poisson distribution special case}

Return to the original contour integral for time-dependent rates (Eq. \ref{eq:contoursln2}), and set $\gamma = c = 0$, but leave the time-dependence of $k(t)$ arbitrary. We have
\begin{equation}
\lambda(t) := \int_{t_0}^t k(t') dt' 
\end{equation}
\begin{equation}
q(t) = p_f
\end{equation}
\begin{equation}
P(x, t; \xi, t_0) = \frac{1}{2 \pi} \int_{-\infty}^{\infty} dp_f \ \frac{e^{i \lambda(t) p_f }}{(1 + i p_f)^{x + 1 - \xi}} \ .
\end{equation}
This contour integral can be evaluated using either Cauchy's integral formula or a table of integrals (c.f. Gradshteyn and Ryzhik \cite{gradshteyn2014} ET I 118(3), in section 3.382, on pg. 365). The result is
\begin{equation}
P(x, t; \xi, t_0) = \frac{\lambda(t)^{x - \xi} e^{- \lambda(t)}}{(x - \xi)!} 
\end{equation}
for $x \geq \xi$ and $0$ otherwise, i.e. a (shifted) Poisson distribution.

%\subsection{Negative binomial distribution special case}

Return to the original contour integral for time-dependent rates (Eq. \ref{eq:contoursln}), and set $k = \gamma = 0$, but leave the time-dependence of $c(t)$ arbitrary. In this case, we will define $w(t)$ differently from before as
\begin{equation}
w(t) := \exp\left[ - \int_{t_0}^t c(t') dt' \right] 
\end{equation}
i.e. as the reciprocal of what we previously called $w(t)$. This is to match the result from Jahnke and Huisinga. Now we have
\begin{equation}
q(t) = \frac{w(t)^{-1}}{\frac{1}{p_f} - i \left[ w(t)^{-1} - 1 \right]}
\end{equation}
\begin{equation}
P(x, t; \xi, t_0) = \frac{1}{2 \pi} \int_{-\infty}^{\infty} dp_f \ \frac{1}{(1 + i p_f)^{x - \xi + 1}} \frac{1}{\left[ 1 - i (w(t)^{-1} - 1) p_f \right]^{\xi}} \ .
\end{equation}
The term on the right has no poles in the upper half-plane, so we can evaluate it using Cauchy's formula to find
\begin{equation}
P(x, t; \xi, t_0) = \binom{x - 1}{\xi - 1} \left[ w(t) \right]^{\xi} \left[ 1 - w(t) \right]^{x - \xi}
\end{equation}
which is nonzero only for $x \geq \xi$, i.e. we have a shifted negative binomial distribution. \qed
\end{proof}

\section{Zero and first order calculations}
\label{sec:firstcalc}

In this section, we sketch the calculations necessary for proving the formulas from Theorem \ref{thm:arb} on a system with arbitrary combinations of zero and first order reactions. First, we present the Hamiltonian operator and kernel. Then we sketch how the path integral expression for the propagator $U(i \mathbf{p}^f, \mathbf{z}^0)$ may be evaluated. Finally, we use the explicit form of the propagator to derive the transition probability and generating function. 

%\subsection{Evaluating the propagator}

Writing down the CME directly is difficult; however, we \textit{do} know that a CME involving only zero or first order reactions only as terms proportional to Grassberger-Scheunert creation operators, and has terms with at most one annihilation operator. We can assume it takes the form (c.f. Eq. \ref{eq:arb_gf_PDE})
\begin{equation}
\begin{split}
\hat{H}(t) =& \sum_{\nu_1, ..., \nu_n} \alpha_{\nu_1, ..., \nu_n}(t) \left( \hat{a}^+_1 \right)^{\nu_1} \cdots \left( \hat{a}^+_n \right)^{\nu_n} \\
&+ \sum_k \sum_{\nu_1, ..., \nu_n} \beta^k_{\nu_1, ..., \nu_n}(t) \left( \hat{a}^+_1 \right)^{\nu_1} \cdots \left( \hat{a}^+_n \right)^{\nu_n} \hat{a}_k
\end{split}
\end{equation}
for some coefficients $\alpha_{\nu_1, ..., \nu_n}(t)$ and $\beta^j_{\nu_1, ..., \nu_n}(t)$ that are determined by the details of one's list of reactions. The corresponding Hamiltonian kernel is
\begin{equation}
\begin{split}
\mathcal{H}(i \mathbf{p}, \mathbf{z}, t) =& \matrixel{- i \mathbf{p}}{\hat{H}(t)}{\mathbf{z}} \\
=& \sum_{\nu_1, ..., \nu_n} \alpha_{\nu_1, ..., \nu_n}(t) \left( i p_1 \right)^{\nu_1} \cdots \left( i p_n \right)^{\nu_n} \\
&+ \sum_k \sum_{\nu_1, ..., \nu_n} \beta^k_{\nu_1, ..., \nu_n}(t) \left( i p_1 \right)^{\nu_1} \cdots \left( i p_n \right)^{\nu_n} z_k \ .
\end{split}
\end{equation}
Now we can compute the propagator corresponding to this kernel. 

\begin{lemma}[Zero and first order reactions propagator] \label{lem:ult_prop}
The propagator for the system with arbitrary combinations of zero and first order reactions is
\begin{equation}
U =  \exp\left\{ i \ \mathbf{z}^0 \cdot \mathbf{q}(t) + \int_{t_0}^t \sum_{\nu_1, ..., \nu_n} \alpha_{\nu_1, ..., \nu_n}(t - s + t_0) \left[ i q_1(s) \right]^{\nu_1} \cdots \left[ i q_n(s) \right]^{\nu_n}   \ ds \right\} 
\end{equation}
where $q(s)$ is as defined in Theorem \ref{thm:arb}.
\end{lemma}

\begin{proof} 
We will only sketch this proof, because the argument is exactly the same as the one presented in Lemma \ref{lem:propcalc_bda}---the notation is just more cluttered, because we are now dealing with a mult-species system and an arbitrarily large list of reactions. One may notice, from a careful look at that prior argument, that its success did not depend on the detailed features of the birth-death-autocatalysis system at all; it only depended on the Hamiltonian containing terms at most first order in annihilation operators (i.e. no terms like $\hat{a}_j \hat{a}_k$ or $(\hat{a}_j)^6$ appear). Since this is also true in the current case, we can rerun that argument to find that the propagator can be written in terms of the solution $\mathbf{q}(t)$ to 
\begin{equation}
\dot{q}_j(s) = - i \sum_{\nu_1, ..., \nu_n} \beta^j_{\nu_1, ..., \nu_n}(t - s + t_0) \left[ i q_1(s) \right]^{\nu_1} \cdots \left[ i q_n(s) \right]^{\nu_n} 
\end{equation}
satisfying the initial condition $q_j(t_0) = p^f_j$. As before, the final propagator has two terms. There is one term that comes from $\mathbf{p}^0 = \mathbf{q}(t)$ coupling to $\mathbf{z}^0$, and another term (due to the terms in the Hamiltonian involving no annihilation operators) that becomes a convolution integral. 
\qed
\end{proof}
Next, we will derive the transition probability.
\begin{lemma}[Zero and first order reactions transition probability]
The transition probability for the system with arbitrary combinations of zero and first order reactions is
\begin{equation} 
P = \int_{\mathbb{R}^n} \frac{d\mathbf{p}^f}{(2\pi)^n} \ \frac{\left[ \mathbf{1} + i \mathbf{q}(t) \right]^{\boldsymbol{\xi}} e^{ \int_{t_0}^t \sum \alpha_{\nu_1, ..., \nu_n}(t - s + t_0) \left[ i q_1(s) \right]^{\nu_1} \cdots \left[ i q_n(s) \right]^{\nu_n}   \ ds }}{(\mathbf{1} + i \mathbf{p}^f)^{\mathbf{x} + \mathbf{1}}} 
\end{equation}
where $q(s)$ is as in Theorem \ref{thm:arb}. 
\end{lemma}
\begin{proof}
Begin with the expression for the generating function in terms of the propagator $U$ (c.f. Corollary \ref{cor:Utogf}). As in Sec. \ref{sec:genptrans}, we have
\begin{equation}
\begin{split}
P =& \frac{1}{\mathbf{x}!} \int \frac{d\mathbf{z}^f d\mathbf{p}^f}{(2 \pi)^n} \frac{d\mathbf{z}^0 d\mathbf{p}^0}{(2 \pi)^n} \braket{\mathbf{x}}{\mathbf{z}^f}_{ex} \ U(i \mathbf{p}^f, \mathbf{z}^0) \braket{- i \mathbf{p}^0}{\psi(t_0)} \ e^{- i \mathbf{p}^0 \cdot \mathbf{z}^0  - i \mathbf{p}^f \cdot \mathbf{z}^f} \\
=& \frac{1}{\mathbf{x}!} \int \frac{d\mathbf{z}^f d\mathbf{p}^f}{(2 \pi)^n} \frac{d\mathbf{z}^0 d\mathbf{p}^0}{(2 \pi)^n} \left(\mathbf{z}^f \right)^\mathbf{x} e^{- \mathbf{z}^f \cdot \mathbf{1}} \ e^{i \mathbf{z}^0 \cdot \mathbf{q}(t)} (\mathbf{1} + i \mathbf{p}^0)^{\boldsymbol{\xi}} \\
&\times  \ e^{- i \mathbf{p}^0 \cdot \mathbf{z}^0  - i \mathbf{p}^f \cdot \mathbf{z}^f}  e^{ \int_{t_0}^t \sum_{\nu_1, ..., \nu_n} \alpha_{\nu_1, ..., \nu_n}(t - s + t_0) \left[ i q_1(s) \right]^{\nu_1} \cdots \left[ i q_n(s) \right]^{\nu_n}   \ ds }  \ .
\end{split}
\end{equation}
The integral over $\mathbf{z}^0$ is
\begin{equation}
\int d\mathbf{z}^0 \ e^{- i \mathbf{z}^0 \cdot ( \mathbf{p}^0 - \mathbf{q}(t))} = \frac{1}{i^n ( \mathbf{p}^0 - \mathbf{q}(t))^{\mathbf{1}}} \ ,
\end{equation}
the integral over $\mathbf{p}^0$ is
\begin{equation}
\int \frac{d\mathbf{p}^0}{(2\pi i)^n} \ \frac{(\mathbf{1} + i \mathbf{p}^0)^{\boldsymbol{\xi}}}{( \mathbf{p}^0 - \mathbf{q}(t))^{\mathbf{1}}} = (\mathbf{1} + i \mathbf{q}(t))^{\boldsymbol{\xi}} \ ,
\end{equation}
and the integral over $\mathbf{z}^f$ is
\begin{equation}
\int d\mathbf{z}^f \ \frac{\left(\mathbf{z}^f \right)^\mathbf{x}}{\mathbf{x}!}  e^{- \mathbf{z}^f \cdot (\mathbf{1} + i \mathbf{p}^f)}= \frac{1}{(\mathbf{1} + i \mathbf{p}^f)^{\mathbf{x} + \mathbf{1}}} \ ,
\end{equation}
leaving only the desired integral.\qed
\end{proof}
Finally, we will derive the generating function.
\begin{lemma}[Zero and first order reactions generating function]
The generating function for the system with arbitrary combinations of zero and first order reactions is
\begin{equation}
\begin{split}
\psi(\mathbf{g}, t) =&  \left[ \mathbf{1} + i \mathbf{q}(t)  \right]^{\boldsymbol{\xi}}   \times  \\ &\times  \left.  e^{ \int_{t_0}^t \sum \alpha_{\nu_1, ..., \nu_n}(t - s + t_0) \left[ i q_1(s) \right]^{\nu_1} \cdots \left[ i q_n(s) \right]^{\nu_n}   \ ds } \right|_{\mathbf{p}^f = - i (\mathbf{g} - \mathbf{1})} 
\end{split}
\end{equation}
where $q(s)$ is as in Theorem \ref{thm:arb}. 
\end{lemma}
\begin{proof}
Begin again with the expression for the generating function in terms of the propagator $U$ (c.f. Corollary \ref{cor:Utogf}). For proving this result, it is convenient to switch over to analytic notation, in which 
\begin{equation}
\begin{split}
\ket{\psi(t)} &\to \psi(\mathbf{g}, t) \\
\ket{\mathbf{x}} &\to \mathbf{g}^{\mathbf{x}} \\
\ket{\mathbf{z}} &\to e^{\mathbf{z} \cdot (\mathbf{g} - \mathbf{1})} \ .
\end{split}
\end{equation}
In particular, we have
\begin{equation}
\begin{split}
\psi =& \int \frac{d\mathbf{z}^f d\mathbf{p}^f}{(2 \pi)^n} \frac{d\mathbf{z}^0 d\mathbf{p}^0}{(2 \pi)^n}  \ e^{\mathbf{z}^f \cdot (\mathbf{g} - \mathbf{1})} e^{i \mathbf{z}^0 \cdot \mathbf{q}(t)} (\mathbf{1} + i \mathbf{p}^0)^{\boldsymbol{\xi}} \\
&\times  \ e^{- i \mathbf{p}^0 \cdot \mathbf{z}^0  - i \mathbf{p}^f \cdot \mathbf{z}^f}  e^{ \int_{t_0}^t \sum_{\nu_1, ..., \nu_n} \alpha_{\nu_1, ..., \nu_n}(t - s + t_0) \left[ i q_1(s) \right]^{\nu_1} \cdots \left[ i q_n(s) \right]^{\nu_n}   \ ds }  \ .
\end{split}
\end{equation}
The integral over $\mathbf{z}^0$ is
\begin{equation}
\int d\mathbf{z}^0 \ e^{- i \mathbf{z}^0 \cdot ( \mathbf{p}^0 - \mathbf{q}(t))} = \frac{1}{i^n ( \mathbf{p}^0 - \mathbf{q}(t))^{\mathbf{1}}} \ ,
\end{equation}
the integral over $\mathbf{p}^0$ is
\begin{equation}
\int \frac{d\mathbf{p}^0}{(2\pi i)^n} \ \frac{(\mathbf{1} + i \mathbf{p}^0)^{\boldsymbol{\xi}}}{( \mathbf{p}^0 - \mathbf{q}(t))^{\mathbf{1}}} = (\mathbf{1} + i \mathbf{q}(t))^{\boldsymbol{\xi}} \ ,
\end{equation}
and the integral over $\mathbf{z}^f$ is
\begin{equation}
\int d\mathbf{z}^f \ e^{- \mathbf{z}^f \cdot \left[- (\mathbf{g} - \mathbf{1}) + i \mathbf{p}^f \right]}= \frac{1}{\left[ - (\mathbf{g} - \mathbf{1}) + i \mathbf{p}^f \right]^{\mathbf{1}}} \ ,
\end{equation}
leaving only the integral
\begin{equation}
\int \frac{d\mathbf{p}^f}{(2\pi i)^n} \ \frac{\left[ \mathbf{1} + i \mathbf{q}(t)  \right]^{\boldsymbol{\xi}} e^{ \int_{t_0}^t \sum_{\nu_1, ..., \nu_n} \alpha_{\nu_1, ..., \nu_n}(t - s + t_0) \left[ i q_1(s) \right]^{\nu_1} \cdots \left[ i q_n(s) \right]^{\nu_n}   \ ds }}{\left[ \mathbf{p}^f + i (\mathbf{g} - \mathbf{1})  \right]^{\mathbf{1}}}   \ .
\end{equation}
This integral is a simple contour integral (the numerator is the exponential of a function analytic in $\mathbf{p}^f$), whose evaluation via Cauchy's integral formula corresponds to the desired result. \qed
\end{proof}

\section{Another view of the propagator}
\label{sec:propview}

The mess of formalism aside, a coarse view of what we have been doing is that we have been calculating the propagator $U$, which we remind the reader is defined via
\begin{equation}
U( i \mathbf{p}^f, \mathbf{z}^0) := \matrixel{- i \mathbf{p}^f}{\hat{U}(t, t_0)}{\mathbf{z}^0} 
\end{equation}
where $\ket{\mathbf{z}^0}$ and $\ket{- i \mathbf{p}^f}$ are coherent states. We computed $U$ by evaluating many integrals, and then used the formula (c.f. Corollary \ref{cor:Utogf})
\begin{equation} 
\ket{\psi(t)} = \int  \frac{d\mathbf{z}^f d\mathbf{p}^f}{(2\pi)^n} \frac{d\mathbf{z}^0 d\mathbf{p}^0}{(2\pi)^n}  \ \ket{\mathbf{z}^f} U( i \mathbf{p}^f, \mathbf{z}^0)  \braket{- i \mathbf{p}^0}{\psi(t_0)}   e^{-i \mathbf{p}^0 \cdot \mathbf{z}^0 -i \mathbf{p}^f \cdot \mathbf{z}^f } 
\end{equation}
to recover the generating function $\ket{\psi(t)}$. This expression for $\ket{\psi(t)}$ is then suitably manipulated to directly recover other objects of interest, like moments or transition probabilities. 

Given the relatively simple-looking results we have derived for $U$ (c.f. Lemmas \ref{lem:mono_prop}, \ref{lem:propcalc_bda}, and \ref{lem:ult_prop}), one may wonder whether there is another way to derive it---in particular, does $U$ satisfy some PDE? 

In the following, it will be more convenient to switch to a more standard notation for the probability generating function:
\begin{equation}
\ket{\psi(t)} = \sum_{\mathbf{x}} P(\mathbf{x}, t) \ket{\mathbf{x}} \to \psi(\mathbf{g}, t) := \sum_{\mathbf{x}} P(\mathbf{x}, t) \ \mathbf{g}^{\mathbf{x}}
\end{equation}
which really just amounts to the replacement $\ket{\mathbf{x}} \to \mathbf{g}^{\mathbf{x}}$. This is related to our notation by taking the Grassberger-Scheunert product of the generating function with a coherent state $\ket{\mathbf{g} - \mathbf{1}}$ for some $\mathbf{g} \in \mathbb{R}^n$:
\begin{equation}
\braket{\mathbf{g} - \mathbf{1}}{\psi(t)} = \sum_{\mathbf{x}} P(\mathbf{x}, t) \braket{\mathbf{g} - \mathbf{1}}{\mathbf{x}} = \sum_{\mathbf{x}} P(\mathbf{x}, t) (\mathbf{1} + (\mathbf{g} - \mathbf{1}))^{\mathbf{x}} = \psi(\mathbf{g}, t) \ .
\end{equation}
In this notation, the relationship between the generating function and the propagator reads
\begin{equation}  \label{eq:gftoU_analytic}
\psi(\mathbf{g}, t) = \int  \frac{d\mathbf{z}^f d\mathbf{p}^f}{(2\pi)^n} \frac{d\mathbf{z}^0 d\mathbf{p}^0}{(2\pi)^n}  \ e^{\mathbf{z}^f \cdot (\mathbf{g} - \mathbf{1})} U( i \mathbf{p}^f, \mathbf{z}^0)  \braket{- i \mathbf{p}^0}{\psi(t_0)}   e^{-i \mathbf{p}^0 \cdot \mathbf{z}^0 -i \mathbf{p}^f \cdot \mathbf{z}^f } \ .
\end{equation}
Recall that the generating function $\psi(\mathbf{g}, t)$ satisfies a partial differential equation. For simplicity, suppose we are dealing with the chemical birth-death process, for which the relevant PDE reads
\begin{equation} 
\frac{\partial \psi}{\partial t} = k(t) [g - 1] \psi - \gamma(t) [g - 1] \frac{\partial \psi}{\partial g} \ .
\end{equation}
Substituting this into (the one-dimensional version of) Eq. \ref{eq:gftoU_analytic}, the right-hand side reads 
\begin{equation}  \label{eq:gftoU_EOM}
\int  \frac{dz_f dp_f}{2\pi} \frac{dz_0 dp_0}{2\pi}  U \left\{ k [g - 1]  - \gamma [g - 1] z_f   \right\} \ e^{z_f (g - 1)} e^{- i p_f z_f} \braket{- i p_0}{\psi(t_0)}   e^{-i p_0 z_0} \ .
\end{equation}
But note that
\begin{equation}
\begin{split}
(g - 1) e^{z_f (g - 1)} &= \frac{\partial}{\partial z_f} e^{z_f (g - 1)} \\
z_f e^{- i p_f z_f} &= i \frac{\partial}{\partial p_f} e^{- i p_f z_f} \ .
\end{split}
\end{equation}
Using these identities, integrating by parts, and freely removing boundary terms, the right-hand side now reads
\begin{equation}  
\int  \frac{dz_f dp_f}{2\pi} \frac{dz_0 dp_0}{2\pi} \left\{ i p_f \left[ k U + i \gamma \frac{\partial U}{\partial p_f} \right]  \right\} \ e^{z_f (g - 1)} e^{- i p_f z_f} \braket{- i p_0}{\psi(t_0)}   e^{-i p_0 z_0} \ .
\end{equation}
This suggests that the expression given by Eq. \ref{eq:gftoU_analytic} will solve the equation of motion for $\ket{\psi(t)}$ (Eq. \ref{eq:EOM}) if
\begin{equation}
\frac{\partial U(i p_f, z_0)}{\partial t} = i p_f \left[ k U(i p_f, z_0) + i \gamma \frac{\partial U(i p_f, z_0)}{\partial p_f} \right] \ .
\end{equation}
It is easy to verify that our expression for the propagator of the chemical birth-death process (c.f. Lemma \ref{lem:mono_prop}) does solve this PDE. 

We can generalize this enough for our purposes, although it should be clear that this correspondence holds for any CME (and not just ones involving only zero and first order reactions). 

\begin{proposition}[Propagator PDE]
If the generating function $\psi(\mathbf{g}, t)$ satisfies the PDE given by Eq. \ref{eq:arb_gf_PDE}, then the propagator $U(i \mathbf{p}^f, t; \mathbf{z}^0, t_0)$ satisfies a PDE
\begin{equation}  \label{eq:prop_PDE}
\begin{split}
\frac{\partial U}{\partial t} =& \sum_{\nu_1, ..., \nu_n} \alpha_{\nu_1, ..., \nu_n}(t) \left[ i p^f_1 \right]^{\nu_1} \cdots \left[ i p^f_n \right]^{\nu_n} U \\
&- i \sum_k \sum_{\nu_1, ..., \nu_n} \beta^k_{\nu_1, ..., \nu_n}(t) \left[ i p^f_1 \right]^{\nu_1} \cdots \left[ i p^f_n \right]^{\nu_n} \frac{\partial U}{\partial p^f_k}
\end{split}
\end{equation}
with initial condition $U(i \mathbf{p}^f, t_0; \mathbf{z}^0, t_0) = \exp( i \mathbf{z}^0 \cdot \mathbf{p}^f)$ for arbitrary $\mathbf{p}^f, \mathbf{z}^0 \in \mathbb{R}^n$.
\end{proposition}
\begin{proof} Integrate by parts as in the one-dimensional example. The initial condition comes from the definition of $U$:
\begin{equation}
U( i \mathbf{p}^f, t_0; \mathbf{z}^0, t_0) = \matrixel{- i \mathbf{p}^f}{\hat{U}(t_0, t_0)}{\mathbf{z}^0} = \braket{- i \mathbf{p}^f}{\mathbf{z}^0} = e^{i \mathbf{z}^0 \cdot \mathbf{p}^f} \ .
\end{equation}
\qed
\end{proof}
At this point, we should note that the PDE satisfied by the propagator (Eq. \ref{eq:prop_PDE}) and the PDE satisfied by the generating function (Eq. \ref{eq:arb_gf_PDE}) are equivalent up to a change of variables (i.e. $\mathbf{g} - \mathbf{1} \to i \mathbf{p}^f$). Does this mean that the propagator, along with the entire Doi-Peliti artifice we have constructed, is extraneous? 

While this is a reasonable question to ask, the answer is probably no. It is easy to see that our expressions for the propagator and our expressions for the generating function have tended to look somewhat different, with the latter almost always being more complicated. The main reason for this difference seems to be that the propagator's initial condition is much simpler than the initial condition for the generating function PDE, which usually permits finding explicit solutions of the propagator PDE.

Now that we have this result, how can we connect it with the propagator solution we found in Lemma \ref{lem:ult_prop} (for arbitrary combinations of zero and first order reactions, which includes all other propagators considered in this paper as special cases)? It turns out that there is a straightforward way to do this using the method of characteristics, a standard approach for solving first order PDEs like the one above. 

The method involves supposing that the relevant independent variables (in this case, $\mathbf{p}^f$ and $t$) lie along some parameterized curve. For a pedagogical example applying this method to solve a toy problem in chemical kinetics (the chemical birth-death process with additive noise), see \cite{vastolaADD2019}.

\begin{lemma}[Method of characteristics solution]

The propagator for the system with arbitrary combinations of zero and first order reactions matches the one given by Lemma \ref{lem:ult_prop}.
\end{lemma}
\begin{proof}
Suppose (where we use slightly different notation here, because only the initial condition of the PDE depends on $\mathbf{z}^0$ and $t_0$) that $\mathbf{p}^f$ and $t$ lie along curves parameterized by some parameter $s$, so that
\begin{equation}
\begin{split}
\frac{\partial}{\partial s} \left[ U(\mathbf{p}^f(s), t(s) \right] =& \frac{\partial U}{\partial t} \frac{\partial t}{\partial s} + \sum_{k = 1}^n \frac{\partial U}{\partial p^f_k} \frac{\partial p^f_k}{\partial s} \\
 =& - \frac{\partial U}{\partial t} - i \sum_k \sum_{\nu_1, ..., \nu_n} \beta^k_{\nu_1, ..., \nu_n}(t) \left[ i p^f_1 \right]^{\nu_1} \cdots \left[ i p^f_n \right]^{\nu_n} \frac{\partial U}{\partial p^f_k} \ .
\end{split}
\end{equation}
Choose the curve so that
\begin{equation}
\begin{split}
\frac{\partial t}{\partial s} &= - 1 \\
\frac{\partial p^f_k}{\partial s} &= - i \sum_{\nu_1, ..., \nu_n} \beta^k_{\nu_1, ..., \nu_n}(t(s)) \left[ i p^f_1(s) \right]^{\nu_1} \cdots \left[ i p^f_n(s) \right]^{\nu_n} \ .
\end{split}
\end{equation}
Suppose that we are interested in $U(i \mathbf{p}^f, t_f; \mathbf{z}^0, t_0)$ for some particular final time $t_f$. Solving the equation for $t(s)$, we have
\begin{equation}
t(s) = t_f - s + t_0
\end{equation}
where the arbitrary constant was chosen so that $s \in [t_0, t_f]$ with $t(t_0) = t_f$ and $t(t_f) = t_0$. Then the equation determining $\mathbf{p}^f(s)$ reads
\begin{equation}
\begin{split}
\frac{\partial p^f_k}{\partial s} &= - i \sum_{\nu_1, ..., \nu_n} \beta^k_{\nu_1, ..., \nu_n}(t_f - s + t_0) \left[ i p^f_1(s) \right]^{\nu_1} \cdots \left[ i p^f_n(s) \right]^{\nu_n} \ .
\end{split}
\end{equation}
Notice that this is exactly the same as the equation satisfied by $\mathbf{q}(s)$ (see Theorem \ref{thm:arb}). Moreover, our $\mathbf{p}^f(s)$ and $\mathbf{q}(s)$ satisfy the same initial condition: $\mathbf{p}^f(s = t_0) = \mathbf{p}^f(t(s) = t_f) = \mathbf{p}^f$, since the symbol $\mathbf{p}^f$ means the value corresponding to the evaluation of $U(i \mathbf{p}^f, t_f; \mathbf{z}^0, t_0)$ at the final time $t_f$. This point is somewhat subtle, so convince yourself of it before going forward. 

Hence, we can make the identification $\mathbf{p}^f(s) \to \mathbf{q}(s)$. This means our PDE for $U$ now reads
\begin{equation}
\frac{\partial U}{\partial s} =  \sum_{\nu_1, ..., \nu_n} \alpha_{\nu_1, ..., \nu_n}(t) \left[ i q_1 \right]^{\nu_1} \cdots \left[ i q_n \right]^{\nu_n} U  \ .
\end{equation}
Solving this as usual, we have the solution
\begin{equation}
U(s) =  C \exp\left\{ \int_{t_0}^s \sum_{\nu_1, ..., \nu_n} \alpha_{\nu_1, ..., \nu_n}(t - s + t_0) \left[ i q_1(s) \right]^{\nu_1} \cdots \left[ i q_n(s) \right]^{\nu_n}   \ ds \right\} 
\end{equation}
for some constant $C$ that depends on the initial condition. Implement the initial condition for $s = t_0$ (i.e. $t(s) = t_f$), noting that $\mathbf{p}^f(s = t_0) = \mathbf{q}(s = t_f)$. Finally, evaluate $U(s)$ at $s = t_f$ to obtain
\begin{equation}
U(t_f) =  \exp\left\{ i \mathbf{z}^0 \cdot \mathbf{q}(t_f) + \int_{t_0}^{t_f} \sum_{\nu_1, ..., \nu_n} \alpha_{\nu_1, ..., \nu_n}(t - s + t_0) \left[ i q_1(s) \right]^{\nu_1} \cdots \left[ i q_n(s) \right]^{\nu_n}   \ ds \right\} 
\end{equation}
which is the desired answer. \qed
\end{proof}
After all this, it is natural to ask whether the path integral calculations were necessary if the answer for the propagator can be determined by solving a relatively simple PDE. The author can only note that he was able to come up with this alternative approach only after carefully studying the path integral answer. It is likely that there are other cases where one can `turn the crank' to determine the path integral answer, and then justify that answer using some more conventional method after one realizes why it takes its precise form.

%THIS SOLUTION REALLY CORRESPONDS WITH A METHOD OF CHARACTERISTICS SOLUTION OF THIS PDE.

%%%%%%%%%%%%%%%%%%%%%%%%%%%%%

\section{Discussion}
\label{sec:discussion}

The strength of the Doi-Peliti approach---that calculations require nothing more clever than evaluating many integrals---is probably also its primary weakness. In Jahnke and Huisinga's original paper, they began with proofs of partial results that offered intuition for why their main result is true: in short, Poisson remains Poisson, and multinomial remains multinomial. In contrast, our calculation does not seem to offer such insight en route to the full solution. This may make it easier to generalize to other kinds of systems (as we did in Sec. \ref{sec:bdacalc} and Sec. \ref{sec:firstcalc}), but it is a little unsatisfying. 
%
%As they noted at the end of their Sec. 6, it interpolates between Poisson ($\gamma = c = 0$), binomial ($k = c = 0$), and negative binomial ($k = \gamma = 0$). 

Still, the Doi-Peliti approach \textit{was} able to generate a solution in a nontrivial case where Jahnke and Huisinga's approach broke down, and we showed that it can offer solutions in far more general and nontrivial cases in Sec. \ref{sec:firstcalc}. While the calculation is likely to be tedious, it seems possible that the Doi-Peliti approach could also be used to find explicit generating functions and transition probabilities (i.e. involving the explicit solution for $\mathbf{q}(t)$) for suitable generalizations of (for example) the birth-death-autocatalysis system, like one that involves many birth reactions, death reactions, and reactions of the form $S_j \to S_k + S_{\ell}$. It is not clear what new insights are necessary to solve explicitly for $\mathbf{q}(t)$ in cases like this.

Another obvious objection to the Doi-Peliti approach is that it is not entirely mathematically rigorous: in rederiving Jahnke and Huisinga's result, we freely swapped many improper integrals, frequently utilized the integral representation of the Dirac delta function, and so on. But we did get answers, and the method is likely to yield answers for problems that other methods cannot currently solve. If nothing else, the Doi-Peliti approach can be used as a tool to generate answers, which can be justified as rigorously correct using some other method (e.g. by showing that they solve the CME directly).  

While we did not resort to approximations in this paper, it is worth noting that utilizing Doi-Peliti path integrals enables the use of powerful perturbative and asymptotic expansions. For most systems of interest in mathematical biology (e.g. gene networks with many species and interactions), this is the way in which the Doi-Peliti approach can be practically applied. See Weber and Frey \cite{weber2017}, and Assaf and Meerson \cite{assaf2017}, for recent reviews discussing approximation techniques related to path integral descriptions of the CME. 

The Doi-Peliti path integral is just one example of a stochastic path integral \cite{weber2017,vastola2019}. The Onsager-Machlup \cite{OM1953pt1,OM1953pt2,graham1977,hertz2016} and Martin-Siggia-Rose-Janssen-De Dominicis \cite{msr1973,janssen1976,dd1976,ddpeliti1978,hertz2016} path integrals are two other examples, which offer an alternative to the Fokker-Planck equation in the same way the Doi-Peliti path integral is an alternative to the CME. While exact computations of these path integrals are also tedious, they are just as mechanical---one can `turn the crank' and generate answer, without relying on (for example) \textit{a priori} knowledge of special functions to solve differential equations \cite{vastolaADD2019,vastola2019gill}. 

\section{Conclusion}

We rederived Jahnke and Huisinga's classic result on monomolecular reaction systems using the Doi-Peliti coherent state path integral approach, which reduces solving the CME to the computation of many integrals. In addition, we also derived an explicit exact time-dependent solution to a problem involving an autocatalytic reaction that was beyond the scope of Jahnke and Huisinga's method, and a formal exact solution for systems involving arbitrary combinations of zero and first order reactions. We hope that our calculations, as well as our detailed description of the Doi-Peliti formalism, help make the Doi-Peliti method more accessible to mathematical biologists studying the CME.

%%%%%%%%%%%%%%%%%%%%%%%%%%%%%%%%%%%

\begin{acknowledgements}
This work was supported by NSF Grant \# DMS 1562078.
\end{acknowledgements}

\appendix

\section{Quantum vs standard notation}
\label{sec:notation_comp}

In this paper, we make abundant use of Dirac's bra-ket notation for vectors and inner products. While this notation is standard in quantum mechanics, it is less often used in areas with a more strictly mathematical bent. At the beginning of Sec. \ref{sec:reframe}, we listed a few important reasons for this choice; to reiterate, it eases notation, makes it easy to repeatedly apply the identity operator (c.f. the derivation of the path integral expression for the propagator $U$), and is suggestive for the inner products we are using. 

In this appendix, we will briefly review it and compare it with notation more commonly used in linear algebra and stochastic processes, so that this paper can be easily read by mathematicians unfamiliar with quantum mechanical notation. 

For now, we will work in one dimension for simplicity. Consider a complex vector space $V$ with a countable basis $e_0$, $e_1$, $e_2$, ..., so that an arbitrary state in this space reads
\begin{equation}
\phi = \sum_{x = 0}^{\infty} c(x) e_x
\end{equation}
for some complex coefficients $c(x)$. In terms of bra-ket notation, we would denote the basis vectors (also called `kets' or `states') by $\ket{0}, \ket{1}, \ket{2}$, ... and an arbitrary state by
\begin{equation}
\ket{\phi} = \sum_{x = 0}^{\infty} c(x) \ket{x} \ ,
\end{equation}
which essentially amounts to the identifications $e_x \to \ket{x}$ and $\phi \to \ket{\phi}$. 

Define the inner product $\langle e_x, e_y \rangle := \delta(x - y)$ for all $x, y \in \mathbb{N}$, and extend it to arbitrary states by linearity. Using bra-ket notation, we would write
\begin{equation}
\braket{x}{y} = \delta(x - y) \ .
\end{equation}
At this point, there are not yet any significant differences between the two choices of notation. The significant differences begin when we consider linear functionals like the functional $L_y: V \to \mathbb{C}$ defined by its action on a basis vector $e_x$:
\begin{equation}
L_y(e_x) := \langle e_y, e_x \rangle \ .
\end{equation} 
Using bra-ket notation, we would denote $L_y$ by $\bra{y}$ (this is called a `bra'), and $L_y$ acting on $e_x$ by $\braket{y}{x}$ (the inner product is sometimes called a `bra-ket'). This allows us to represent Fourier-like identities like
\begin{equation}
\phi = \sum_{y = 0}^{\infty} \langle e_y, \phi \rangle e_y
\end{equation}
via
\begin{equation}  \label{eq:bk_fourier}
\ket{\phi} = \sum_{y = 0}^{\infty} \braket{y}{\phi} \ket{y} \ ,
\end{equation}
or more succinctly by defining the operator
\begin{equation}   \label{eq:app_roi}
1 = \sum_{y = 0}^{\infty}  \ket{y} \bra{y}
\end{equation}
which by definition is equal to the identity operator. Equations like these are often called `resolutions of the identity', because they recast the identity operator in some convenient form. The notation above is meant to be highly suggestive; one can imagine it `bumping into' a vector/state $\ket{\phi}$ from the left to recover Eq. \ref{eq:bk_fourier}.

This notation also makes it easy to repeatedly apply resolutions of the identity, and to see what the result will be. Compare
\begin{equation}
\phi = \sum_{y_1, y_2, y_3} \langle e_{y_3}, e_{y_2} \rangle \langle e_{y_2}, e_{y_1} \rangle \langle e_{y_1}, \phi \rangle e_{y_3}
\end{equation}
to
\begin{equation}
\ket{\phi} = \sum_{y_1, y_2, y_3} \ket{y_3} \braket{y_3}{y_2} \braket{y_2}{y_1} \braket{y_1}{\phi} \ .
\end{equation}
The above can be obtained simply by inserting Eq. \ref{eq:app_roi} next to $\ket{\phi}$ many times. 

One helpful feature of bra-ket notation is that eigenvectors are traditionally labeled by their eigenvalues. For example, if $\hat{\mathcal{A}} \phi = \lambda \phi$, it is traditional to write $\phi$ as
\begin{equation}
\phi \to \ket{\lambda} \ ,
\end{equation}
so that $\hat{\mathcal{A}} \ket{\lambda} = \lambda \ket{\lambda}$. We used this throughout the paper to denote coherent states, which we defined to be eigenstates of the annihilation operators. 

Matrix elements---expressions of the form $\langle \phi_2, \hat{\mathcal{A}} \phi_1 \rangle$ for two vectors $\phi_1$ and $\phi_2$ and some operator $\hat{\mathcal{A}}$---are denoted by
\begin{equation}
\mel{\phi_2}{\hat{\mathcal{A}}}{\phi_1} \ .
\end{equation}
This notation is convenient when we are computing matrix elements involving operators and their eigenstates. For example, let $\hat{a}$ be an operator, $\hat{a}^{\dag}$ be its Hermitian conjugate, and $\phi_1 \to \ket{\lambda_1}$ and $\phi_2 \to \ket{\lambda_2}$ be eigenstates with eigenvalues $\lambda_1$ and $\lambda_2$, respectively. Then on the one hand, we have
\begin{equation}
\langle \phi_2, \hat{a}^{\dag} \hat{a} \ \phi_1 \rangle = \langle \hat{a} \phi_2 , \hat{a} \phi_1 \rangle = \lambda_2^* \lambda_1 \langle \phi_2, \phi_1 \rangle
\end{equation}
in standard notation. On the other hand, we have
\begin{equation}
\mel{\lambda_2}{\hat{a}^{\dag} \hat{a}}{\lambda_1} = \lambda_2^* \lambda_1 \braket{\lambda_2}{\lambda_1}
\end{equation}
using bra-ket notation, where we imagine $\hat{a}^{\dag}$ `acting to the left' and $\hat{a}$ `acting to the right'. 

That is about all there is to say about the correspondence between bra-ket notation and typical vector space notation. One should keep in mind that the strength of bra-ket notation is in repeatedly applying the identity operator/resolutions of the identity, which is required to construct the Doi-Peliti path integral. The correspondence is summarized (for arbitrary dimensions, using the notation introduced in Sec. \ref{sec:pstatement}) in Table \ref{table:qm_vs_notation}.

A few words should also be said about the relationship between our generating function and its usual analytic function form. We remind the reader that both are defined (in one dimension again, for simplicity) via
\begin{equation}
\psi(g, t) = \sum_{x = 0}^{\infty} P(x, t) \ g^x \hspace{0.5in} \ket{\psi} = \sum_{x = 0}^{\infty} P(x, t) \ \ket{x}
\end{equation}
where $P(x, t)$ is a solution to the CME. These expressions are completely equivalent, up to the identification $g^x \to \ket{x}$. In fact, the equations of motion they satisfy exactly correspond. For example, in the case of the chemical birth-death process, we remind the reader that $\psi(g, t)$ satisfies the PDE
\begin{equation} 
\frac{\partial \psi(g, t)}{\partial t} = k(t) [g - 1] \psi(g, t) - \gamma(t) [g - 1] \frac{\partial \psi(g, t)}{\partial g} 
\end{equation}
whereas $\ket{\psi}$ satisfies the equation 
\begin{equation}
\frac{\partial \ket{\psi}}{\partial t} = \hat{H} \ket{\psi}
\end{equation}
where in this case the Hamiltonian operator $\hat{H}$ is given (in terms of our original creation and annihilation operators) by
\begin{equation}
\hat{H} = k (\hat{\pi} - 1) - \gamma (\hat{\pi} - 1) \hat{a} \ .
\end{equation}
This is the same as the above PDE, provided one makes the identifications
\begin{equation}
\begin{split}
\hat{\pi} &\to g \\
\hat{a} &\to \frac{\partial}{\partial g} \ .
\end{split}
\end{equation}
It turns out that these identifications work more generally (for arbitrary numbers of dimensions, and an arbitrary list of reactions). Although they are equivalent, one form of the generating function is often more convenient to use than the other. In our case, we use the Hilbert space form for almost the entirety of this paper, because it allows us to exploit bra-ket notation to denote applying many resolutions of the identity (c.f. Sec. \ref{sec:roi}), and to work straightforwardly in terms of matrix elements of the Hamiltonian.

\begin{table}[ht!] 
%\hspace*{-2cm}
\centering
 \begin{tabular}{| c | c | c |} 
 \hline
Object & Bra-ket notation & Standard notation \\ [0.5ex] 
 \hline \hline
basis vector/ket & $\ket{\mathbf{x}}$ & $\mathbf{e}_{\mathbf{x}}$ \\ 
linear functional/bra & $\bra{\mathbf{x}}$ & $L_{\mathbf{x}}: \mathbf{e}_{\mathbf{y}} \mapsto \langle \mathbf{e}_{\mathbf{x}} , \mathbf{e}_{\mathbf{y}} \rangle$ \\ 
zero vector & $0$ & $0$ \\
arbitrary state &  $\displaystyle\ket{\phi} = \sum_{\mathbf{x}} c(\mathbf{x}) \ket{\mathbf{x}}$ &  $\displaystyle\boldsymbol{\phi} = \sum_{\mathbf{x}} c(\mathbf{x}) \ \mathbf{e}_{\mathbf{x}}$ \\
inner product & $\braket{\mathbf{x}}{\mathbf{y}}$ & $\langle \mathbf{e}_{\mathbf{x}} , \mathbf{e}_{\mathbf{y}} \rangle$  \\ 
operator matrix element & $\mel{\mathbf{x}}{\mathcal{A}}{\mathbf{y}}$ & $\langle \mathbf{e}_{\mathbf{x}} , \mathcal{A} \ \mathbf{e}_{\mathbf{y}} \rangle = \langle \mathcal{A}^{\dag} \ \mathbf{e}_{\mathbf{x}} , \mathbf{e}_{\mathbf{y}} \rangle$ \\
generating function & $\displaystyle\ket{\psi(t)} = \sum_{\mathbf{x}} P(\mathbf{x}, t) \ket{\mathbf{x}}$ & $\displaystyle\psi(t) = \sum_{\mathbf{x}} P(\mathbf{x}, t) \ \mathbf{e}_{\mathbf{x}}$ \\
coherent state (c.s.) & $\displaystyle\ket{\mathbf{z}} = \sum_{\mathbf{y}} \frac{\mathbf{z}^{\mathbf{y}}}{\mathbf{y}!} e^{- \mathbf{z} \cdot \mathbf{1}} \ \ket{\mathbf{y}}$ & $\displaystyle\text{cs}(\mathbf{z}) = \sum_{\mathbf{y}} \frac{\mathbf{z}^{\mathbf{y}}}{\mathbf{y}!} e^{- \mathbf{z} \cdot \mathbf{1}} \ \mathbf{e}_{\mathbf{y}}$ \\
c.s. identity operator &  $\displaystyle\ket{\mathbf{x}} = \int_{[0, \infty)^n} d\mathbf{z} \int_{\mathbb{R}^n} \frac{d\mathbf{p}}{(2\pi)^n}   \ \ket{\mathbf{z}} \braket{- i \mathbf{p}}{\mathbf{x}} e^{-i \mathbf{z} \cdot \mathbf{p}} $ & $\displaystyle\mathbf{e}_{\mathbf{x}} = \int_{[0, \infty)^n} d\mathbf{z} \int_{\mathbb{R}^n} \frac{d\mathbf{p}}{(2\pi)^n}   \ \text{cs}(\mathbf{z}) \langle \text{cs}(- i \mathbf{p}), \mathbf{e}_\mathbf{x} \rangle e^{-i \mathbf{z} \cdot \mathbf{p}} $ \\
 \hline
 \end{tabular}
 \caption{Let $\mathbf{x} \in \mathbb{R}^n$, and let the notation be as in Sec. \ref{sec:intro_mono} (e.g. $\mathbf{x}!  := x_1! \cdots x_n!$). This table summarizes the correspondence between quantum and standard notation for several objects discussed in this appendix, as well as objects discussed elsewhere in this paper (e.g. coherent states). }
\label{table:qm_vs_notation}
\end{table}

Finally, we should say that the coherent state resolution of the identity we used many times to construct the Doi-Peliti path integral (c.f. Sec. \ref{sec:roi}) can be written in terms of ordinary functions as
\begin{equation}
g^x = \int_0^{\infty} dz \int_{-\infty}^{\infty} \frac{dp}{2\pi} e^{z (g-1)} (1 + i p)^x e^{- i z p} 
\end{equation}
in one dimension, and
\begin{equation}
\mathbf{g}^{\mathbf{x}} = \int_{[0, \infty)^n} d\mathbf{z} \int_{\mathbb{R}^n} \frac{d\mathbf{p}}{(2\pi)^n}   \ e^{\mathbf{z} \cdot (\mathbf{g} - \mathbf{1})} (\mathbf{1} + i \mathbf{p})^{\mathbf{x}} e^{-i \mathbf{z} \cdot \mathbf{p}} 
\end{equation}
in arbitrarily many dimensions. However, attempting to construct the path integral using this notation instead of bra-ket notation is significantly messier, so we have avoided it. 

\bibliographystyle{spphys}       % APS-like style for physics
%\bibliography{dp_bib, gillnoise_bib, addnoise_bib, re_pathint,last_bib,other_bib, resub_bib}

\begin{thebibliography}{10}
\providecommand{\url}[1]{{#1}}
\providecommand{\urlprefix}{URL }
\expandafter\ifx\csname urlstyle\endcsname\relax
  \providecommand{\doi}[1]{DOI \discretionary{}{}{}#1}\else
  \providecommand{\doi}{DOI \discretionary{}{}{}\begingroup
  \urlstyle{rm}\Url}\fi

\bibitem{mcquarrie1967}
D.A. McQuarrie, Journal of Applied Probability \textbf{4}(3), 413–478 (1967).
\newblock \doi{10.2307/3212214}

\bibitem{gillespie1992}
D.T. Gillespie, Physica A: Statistical Mechanics and its Applications
  \textbf{188}(1), 404  (1992).
\newblock \doi{https://doi.org/10.1016/0378-4371(92)90283-V}.
\newblock
  \urlprefix\url{http://www.sciencedirect.com/science/article/pii/037843719290283V}

\bibitem{gillespie2000}
D.T. Gillespie, The Journal of Chemical Physics \textbf{113}(1), 297 (2000)

\bibitem{gillespie2007}
D.T. Gillespie, Annual Review of Physical Chemistry \textbf{58}(1), 35 (2007)

\bibitem{gillespie2013}
D.T. Gillespie, A.~Hellander, L.R. Petzold, The Journal of Chemical Physics
  \textbf{138}(17), 170901 (2013)

\bibitem{fox2017}
Z.~{Fox}, B.~{Munsky}, arXiv e-prints arXiv:1708.09264 (2017)

\bibitem{munsky2018}
B.~Munsky, W.S. Hlavacek, L.S. Tsimring (eds.), \emph{Quantitative Biology:
  Theory, Computational Methods, and Models} (The MIT Press, 2018)

\bibitem{neuert2013}
G.~Neuert, B.~Munsky, R.Z. Tan, L.~Teytelman, M.~Khammash, A.~van Oudenaarden,
  Science \textbf{339}(6119), 584 (2013).
\newblock \doi{10.1126/science.1231456}

\bibitem{munsky2015}
B.~Munsky, Z.~Fox, G.~Neuert, Methods \textbf{85}, 12  (2015).
\newblock Inferring Gene Regulatory Interactions from Quantitative
  High-Throughput Measurements

\bibitem{fox2016}
Z.~Fox, G.~Neuert, B.~Munsky, The Journal of Chemical Physics \textbf{145}(7),
  074101 (2016)

\bibitem{munsky2018shape}
B.~Munsky, G.~Li, Z.R. Fox, D.P. Shepherd, G.~Neuert, Proceedings of the
  National Academy of Sciences \textbf{115}(29), 7533 (2018).
\newblock \doi{10.1073/pnas.1804060115}

\bibitem{weber2018}
L.~Weber, W.~Raymond, B.~Munsky, Physical Biology \textbf{15}(5), 055001 (2018)

\bibitem{fox2019}
Z.R. Fox, B.~Munsky, PLOS Computational Biology \textbf{15}(1), 1 (2019)

\bibitem{fox2019temporal}
Z.R. Fox, G.~Neuert, B.~Munsky, bioRxiv  (2019).
\newblock \doi{10.1101/812479}.
\newblock
  \urlprefix\url{https://www.biorxiv.org/content/early/2019/10/21/812479}

\bibitem{raj2008}
A.~Raj, P.~van~den Bogaard, S.A. Rifkin, A.~van Oudenaarden, S.~Tyagi, Nature
  Methods \textbf{5}(10), 877 (2008)

\bibitem{femino1998}
A.M. Femino, F.S. Fay, K.~Fogarty, R.H. Singer, Science \textbf{280}(5363), 585
  (1998)

\bibitem{rahman2013}
S.~Rahman, D.~Zenklusen, in \emph{Imaging Gene Expression: Methods and
  Protocols}, ed. by Y.~Shav-Tal (Humana Press, Totowa, NJ, 2013), pp. 33--46

\bibitem{ovaskainen2010}
O.~Ovaskainen, B.~Meerson, Trends in Ecology \& Evolution \textbf{25}(11), 643
  (2010)

\bibitem{melbinger2010}
A.~Melbinger, J.~Cremer, E.~Frey, Phys. Rev. Lett. \textbf{105}, 178101 (2010)

\bibitem{assaf2017}
M.~Assaf, B.~Meerson, Journal of Physics A: Mathematical and Theoretical
  \textbf{50}(26), 263001 (2017)

\bibitem{nagel1992}
{Kai Nagel}, {Michael Schreckenberg}, J. Phys. I France \textbf{2}(12), 2221
  (1992)

\bibitem{mahnke1997}
R.~Mahnke, N.~Pieret, Phys. Rev. E \textbf{56}, 2666 (1997)

\bibitem{mahnke2005}
R.~Mahnke, J.~Kaupužs, I.~Lubashevsky, Physics Reports \textbf{408}(1), 1
  (2005)

\bibitem{miller2006}
J.A. Miller, S.J. Klippenstein, The Journal of Physical Chemistry A
  \textbf{110}(36), 10528 (2006)

\bibitem{glowacki2012}
D.R. Glowacki, C.H. Liang, C.~Morley, M.J. Pilling, S.H. Robertson, The Journal
  of Physical Chemistry A \textbf{116}(38), 9545 (2012)

\bibitem{jasper2014}
A.W. Jasper, K.M. Pelzer, J.A. Miller, E.~Kamarchik, L.B. Harding, S.J.
  Klippenstein, Science \textbf{346}(6214), 1212 (2014)

\bibitem{gillespie1976}
D.T. Gillespie, Journal of Computational Physics \textbf{22}(4), 403  (1976).
\newblock \doi{https://doi.org/10.1016/0021-9991(76)90041-3}.
\newblock
  \urlprefix\url{http://www.sciencedirect.com/science/article/pii/0021999176900413}

\bibitem{gillespie1977}
D.T. Gillespie, The Journal of Physical Chemistry \textbf{81}(25), 2340 (1977)

\bibitem{munsky2006}
B.~Munsky, M.~Khammash, The Journal of Chemical Physics \textbf{124}(4), 044104
  (2006)

\bibitem{peles2006}
S.~Peleš, B.~Munsky, M.~Khammash, The Journal of Chemical Physics
  \textbf{125}(20), 204104 (2006)

\bibitem{harris2006}
L.A. Harris, P.~Clancy, The Journal of Chemical Physics \textbf{125}(14),
  144107 (2006)

\bibitem{harris2009}
L.A. Harris, A.M. Piccirilli, E.R. Majusiak, P.~Clancy, Phys. Rev. E
  \textbf{79}, 051906 (2009).
\newblock \doi{10.1103/PhysRevE.79.051906}.
\newblock \urlprefix\url{https://link.aps.org/doi/10.1103/PhysRevE.79.051906}

\bibitem{iyengar2010}
K.A. Iyengar, L.A. Harris, P.~Clancy, The Journal of Chemical Physics
  \textbf{132}(9), 094101 (2010)

\bibitem{bokes2012}
P.~Bokes, J.R. King, A.T.A. Wood, M.~Loose, Journal of Mathematical Biology
  \textbf{65}(3), 493 (2012)

\bibitem{hasenauer2014}
J.~Hasenauer, V.~Wolf, A.~Kazeroonian, F.J. Theis, Journal of Mathematical
  Biology \textbf{69}(3), 687 (2014)

\bibitem{kan2016}
X.~Kan, C.H. Lee, H.G. Othmer, Journal of Mathematical Biology \textbf{73}(5),
  1081 (2016)

\bibitem{gillespie2002}
D.T. Gillespie, The Journal of Physical Chemistry A \textbf{106}(20), 5063
  (2002)

\bibitem{grima2011}
R.~Grima, P.~Thomas, A.V. Straube, The Journal of Chemical Physics
  \textbf{135}(8), 084103 (2011)

\bibitem{vastolaPRE2020}
J.J. Vastola, W.R. Holmes, Phys. Rev. E \textbf{101}, 032417 (2020)

\bibitem{delbruck1940}
M.~Delbrück, The Journal of Chemical Physics \textbf{8}(1), 120 (1940)

\bibitem{renyi1954}
A.~R{\'e}nyi, Magyar Tud. Akad. Alkalm. Mat. Int. K{\"o}zl \textbf{2}, 93
  (1954)

\bibitem{ishida1960}
K.~Ishida, Bulletin of the Chemical Society of Japan \textbf{33}(8), 1030
  (1960)

\bibitem{mcquarrie1963}
D.A. McQuarrie, The Journal of Chemical Physics \textbf{38}(2), 433 (1963)

\bibitem{mcquarrie1964}
D.A. McQuarrie, C.J. Jachimowski, M.E. Russell, The Journal of Chemical Physics
  \textbf{40}(10), 2914 (1964)

\bibitem{jahnke2007}
T.~Jahnke, W.~Huisinga, Journal of Mathematical Biology \textbf{54}(1), 1
  (2007)

\bibitem{reis2018}
M.~Reis, J.A. Kromer, E.~Klipp, Journal of Mathematical Biology \textbf{77}(2),
  377 (2018)

\bibitem{laurenzi2000}
I.J. Laurenzi, The Journal of Chemical Physics \textbf{113}(8), 3315 (2000)

\bibitem{arslan2008}
E.~Arslan, I.J. Laurenzi, The Journal of Chemical Physics \textbf{128}(1),
  015101 (2008)

\bibitem{doi1976}
M.~Doi, J. Phys. A \textbf{9}(9), 1479 (1976).
\newblock \doi{10.1088/0305-4470/9/9/009}.
\newblock \urlprefix\url{https://doi.org/10.1088%2F0305-4470%2F9%2F9%2F009}

\bibitem{doi1976second}
M.~{Doi}, J. Phys. A \textbf{9}, 1465 (1976).
\newblock \doi{10.1088/0305-4470/9/9/008}

\bibitem{peliti1985}
L.~Peliti, Journal de Physique \textbf{46}(9), 1469 (1985).
\newblock \doi{10.1051/jphys:019850046090146900}.
\newblock
  \urlprefix\url{http://www.edpsciences.org/10.1051/jphys:019850046090146900}

\bibitem{peliti1985eden}
L.~Peliti, Y.C. Zhang, Journal de Physique Lettres \textbf{46}(24), 1151 (1985)

\bibitem{peliti1986}
L.~Peliti, Journal of Physics A: Mathematical and General \textbf{19}(6), L365
  (1986)

\bibitem{grassberger1980}
P.~Grassberger, M.~Scheunert, Fortschritte der Physik \textbf{28}(10), 547
  (1980).
\newblock \doi{10.1002/prop.19800281004}

\bibitem{grassberger1982}
P.~Grassberger, Zeitschrift f{\"u}r Physik B Condensed Matter \textbf{47}(4),
  365 (1982)

\bibitem{cardy1985}
J.L. Cardy, P.~Grassberger, Journal of Physics A: Mathematical and General
  \textbf{18}(6), L267 (1985)

\bibitem{grassberger1989}
P.~Grassberger, Journal of Physics A: Mathematical and General \textbf{22}(23),
  L1103 (1989)

\bibitem{mattis1998}
D.C. Mattis, M.L. Glasser, Rev. Mod. Phys. \textbf{70}, 979 (1998).
\newblock \doi{10.1103/RevModPhys.70.979}.
\newblock \urlprefix\url{https://link.aps.org/doi/10.1103/RevModPhys.70.979}

\bibitem{lee1994}
B.P. Lee, Journal of Physics A: Mathematical and General \textbf{27}(8), 2633
  (1994)

\bibitem{leecardy1994}
B.P. Lee, J.~Cardy, Phys. Rev. E \textbf{50}, R3287 (1994)

\bibitem{lee1995}
B.P. Lee, J.~Cardy, Journal of Statistical Physics \textbf{80}(5), 971 (1995)

\bibitem{vanWijland1998}
F.~van Wijland, K.~Oerding, H.~Hilhorst, Physica A: Statistical Mechanics and
  its Applications \textbf{251}(1), 179  (1998)

\bibitem{canet2004}
L.~Canet, B.~Delamotte, O.~Deloubri\`ere, N.~Wschebor, Phys. Rev. Lett.
  \textbf{92}, 195703 (2004)

\bibitem{canet2006}
L.~Canet, Journal of Physics A: Mathematical and General \textbf{39}(25), 7901
  (2006)

\bibitem{tauber2005}
U.C. Täuber, M.~Howard, B.P. Vollmayr-Lee, Journal of Physics A: Mathematical
  and General \textbf{38}(17), R79 (2005)

\bibitem{fung2017}
T.~Fung, J.P. O'Dwyer, R.A. Chisholm, Journal of Mathematical Biology
  \textbf{74}(1), 289 (2017)

\bibitem{greenman2016}
C.D. Greenman, T.~Chou, Phys. Rev. E \textbf{93}, 012112 (2016)

\bibitem{greenman2017}
C.D. Greenman, Journal of Statistical Mechanics: Theory and Experiment
  \textbf{2017}(3), 033101 (2017)

\bibitem{bressloff2014}
P.C. Bressloff, \emph{Stochastic Processes in Cell Biology} (Springer, 2014).
\newblock \doi{10.1007/978-3-319-08488-6}

\bibitem{albert2019}
J.~Albert, Journal of Mathematical Biology  (2019)

\bibitem{bressloff2017}
P.C. Bressloff, Journal of Physics A: Mathematical and Theoretical
  \textbf{50}(13), 133001 (2017)

\bibitem{vastolaADD2019}
J.J. {Vastola}, arXiv e-prints arXiv:1910.09117 (2019)

\bibitem{griffiths2018}
D.~Griffiths, D.~Schroeter, \emph{Introduction to Quantum Mechanics} (Cambridge
  University Press, 2018)

\bibitem{schwartz2014}
M.D. Schwartz, \emph{{Quantum Field Theory and the Standard Model}} (Cambridge
  University Press, 2014).
\newblock
  \urlprefix\url{http://www.cambridge.org/us/academic/subjects/physics/theoretical-physics-and-mathematical-physics/quantum-field-theory-and-standard-model}

\bibitem{cardy2008}
J.~Cardy, G.~Falkovich, K.~Gawedzki, \emph{Reaction-diffusion processes}
  (Cambridge University Press, 2008), p. 108–161.
\newblock London Mathematical Society Lecture Note Series.
\newblock \doi{10.1017/CBO9780511812149.004}

\bibitem{weber2017}
M.F. Weber, E.~Frey, Reports on Progress in Physics \textbf{80}(4), 046601
  (2017).
\newblock \doi{10.1088/1361-6633/aa5ae2}.
\newblock
  \urlprefix\url{http://stacks.iop.org/0034-4885/80/i=4/a=046601?key=crossref.02abaf744081951aaaca3bafec0e1284}

\bibitem{vastola2019gill}
J.J. {Vastola}, arXiv e-prints arXiv:1910.10807 (2019)

\bibitem{gradshteyn2014}
I.S. Gradshteyn, I.M. Ryzhik, \emph{Table of integrals, series, and products}
  (Academic press, 2014)

\bibitem{vastola2019}
J.J. {Vastola}, W.R. {Holmes}, arXiv e-prints arXiv:1909.12990 (2019)

\bibitem{OM1953pt1}
L.~Onsager, S.~Machlup, Phys. Rev. \textbf{91}, 1505 (1953).
\newblock \doi{10.1103/PhysRev.91.1505}.
\newblock \urlprefix\url{https://link.aps.org/doi/10.1103/PhysRev.91.1505}

\bibitem{OM1953pt2}
S.~Machlup, L.~Onsager, Phys. Rev. \textbf{91}, 1512 (1953).
\newblock \doi{10.1103/PhysRev.91.1512}.
\newblock \urlprefix\url{https://link.aps.org/doi/10.1103/PhysRev.91.1512}

\bibitem{graham1977}
R.~Graham, Zeitschrift f{\"u}r Physik B Condensed Matter \textbf{26}(3), 281
  (1977).
\newblock \doi{10.1007/BF01312935}.
\newblock \urlprefix\url{https://doi.org/10.1007/BF01312935}

\bibitem{hertz2016}
J.A. Hertz, Y.~Roudi, P.~Sollich, Journal of Physics A: Mathematical and
  Theoretical \textbf{50}(3), 033001 (2016).
\newblock \doi{10.1088/1751-8121/50/3/033001}.
\newblock \urlprefix\url{https://doi.org/10.1088%2F1751-8121%2F50%2F3%2F033001}

\bibitem{msr1973}
P.C. Martin, E.D. Siggia, H.A. Rose, Phys. Rev. A \textbf{8}, 423 (1973).
\newblock \doi{10.1103/PhysRevA.8.423}.
\newblock \urlprefix\url{https://link.aps.org/doi/10.1103/PhysRevA.8.423}

\bibitem{janssen1976}
H.K. Janssen, Zeitschrift f{\"u}r Physik B Condensed Matter \textbf{23}(4), 377
  (1976).
\newblock \doi{10.1007/BF01316547}.
\newblock \urlprefix\url{https://doi.org/10.1007/BF01316547}

\bibitem{dd1976}
{DE DOMINICIS, C.}, J. Phys. Colloques \textbf{37}(C1), C1 (1976).
\newblock \doi{10.1051/jphyscol:1976138}.
\newblock \urlprefix\url{https://doi.org/10.1051/jphyscol:1976138}

\bibitem{ddpeliti1978}
C.~De~Dominicis, L.~Peliti, Phys. Rev. B \textbf{18}, 353 (1978).
\newblock \doi{10.1103/PhysRevB.18.353}.
\newblock \urlprefix\url{https://link.aps.org/doi/10.1103/PhysRevB.18.353}

\end{thebibliography}


\end{document}